%% file: GICIFB.tex
\documentclass[journal,final]{IEEEtran}

\usepackage{amsmath,amssymb,mathrsfs,dsfont,amsthm}
\usepackage{epsfig,epsf,subfigure,graphicx,graphics, url,color}

\allowdisplaybreaks

\input{commands.tex}

\graphicspath{{./figs/}}

\begin{document}

\title{Gaussian Interference Channel with Intermittent Feedback}

\author{Can~Karakus,~\IEEEmembership{Student Member,~IEEE,}
        I-Hsiang~Wang,~\IEEEmembership{Member,~IEEE,}
        and~Suhas~Diggavi,~\IEEEmembership{Fellow,~IEEE}%
        \thanks{The research of C. Karakus and S. Diggavi was supported in part by
NSF grant 1314937 and a gift from Intel. The work of I.-H. Wang was supported by Ministry of Science and Technology, Taiwan, under Grants MOST 103-2221-E-002-089-MY2 and MOST 103-2622-E-002-034. This paper was presented in part at 2013 IEEE International Symposium on Information Theory and 2013 Allerton Conference on Communication, Control, and Computing.}
        \thanks{C. Karakus and S. Diggavi are with the Department of Electrical Engineering,
        University of California, Los Angeles, CA 90095 USA (e-mail: karakus@ucla.edu; 	suhas@ee.ucla.edu)}%
        \thanks{I.-H. Wang is with the Department of Electrical Engineering,
        National Taiwan University, Taipei, Taiwan (e-mail: ihwang@ntu.edu.tw)}%
        \thanks{Copyright (c) 2014 IEEE. Personal use of this material is permitted.  However, permission to use this material for any other purposes must be obtained from the IEEE by sending a request to pubs-permissions@ieee.org.}
}

\maketitle

\begin{abstract}
We investigate how to exploit intermittent feedback for interference
management by studying the two-user Gaussian interference channel
(IC).  We approximately characterize (within a universal constant) the
capacity region for the Gaussian IC with intermittent feedback.  We
exactly characterize the the capacity region of the linear
deterministic version of the problem, which gives us insight into the
Gaussian problem.  We find that the characterization only depends on
the forward channel parameters and the marginal probability
distribution of each feedback link.  The result shows that passive and
unreliable feedback can be harnessed to provide multiplicative
capacity gain in Gaussian interference channels.  We find that when
the feedback links are active with sufficiently large probabilities,
the perfect feedback sum-capacity is achieved to within a constant
gap.  In contrast to other schemes developed for interference channel
with feedback, our achievable scheme makes use of
\emph{quantize-map-and-forward} to relay the information obtained
through feedback, performs forward decoding, and does not use
structured codes. We also develop new outer bounds enabling us to
obtain the (approximate) characterization of the capacity region.
\end{abstract}


\section{Introduction} \label{sec:introduction}

\input{Introduction.tex}

\section{System Model} \label{sec:model}

\input{Model.tex}

\section{Main Results} \label{sec:results}

\input{Result.tex}

\section{Motivation of the Coding Scheme} \label{sec:motivation}

\input{Motivation.tex}

\section{Achievability Proof} \label{sec:achievability}
\input{Achievability.tex}

\section{Converse Proof} \label{sec:converse}

\input{Converse.tex}

\section{Discussion and Extensions}  \label{sec:discussion}
\input{Discussion.tex}

\bibliography{Ref}

\appendices

\renewcommand{\thetheorem}{\Alph{section}.\arabic{theorem}}
\renewcommand{\thelemma}{\Alph{section}.\arabic{lemma}}
\renewcommand{\thecorollary}{\Alph{section}.\arabic{corollary}}
\renewcommand{\theremark}{\Alph{section}.\arabic{remark}}
\renewcommand{\theclaim}{\Alph{section}.\arabic{claim}}

\section{Proof of Lemma \ref{lem:listsize}} \label{sec:ap_listsize}
\input{AP_Listsize.tex}

\section{Proofs of Lemmas \ref{lem:weak_fb}, \ref{lem:weak_hk}, \ref{lem:strong_fb} and \ref{lem:strong_hk}} \label{sec:ap_achievability}
\input{AP_Achievability.tex}

\section{Evaluation of Rate Regions} \label{sec:ap_evaluation}
\input{AP_Evaluation.tex}

\section{Proofs of Outer Bounds \eqref{eq:ob_ldc_Ri}, \eqref{eq:ob_ldc_R1R2}, \eqref{eq:ob_ldc_2R1R2}, and \eqref{eq:ob_ldc_R12R2}} \label{sec:ap_ob_ldc}
\input{AP_OuterBoundLDC.tex}

\section{Proofs of Outer Bounds \eqref{eq:ob_g_Ri}, \eqref{eq:ob_g_R1R2}, \eqref{eq:ob_g_2R1R2}, and \eqref{eq:ob_g_R12R2}} \label{sec:ap_ob_g}
\input{AP_OuterBoundGaussian.tex}

\section{Gap Analysis} \label{sec:ap_gap}
\input{AP_GapAnalysis.tex}

\section{Proofs of Corollaries \ref{cor:threshold_ldc} and \ref{cor:threshold}} \label{sec:ap_threshold}
\input{AP_Threshold.tex}

\begin{IEEEbiography}[\raisebox{15mm}{\includegraphics[width=1in,height=1.25in,clip,keepaspectratio]{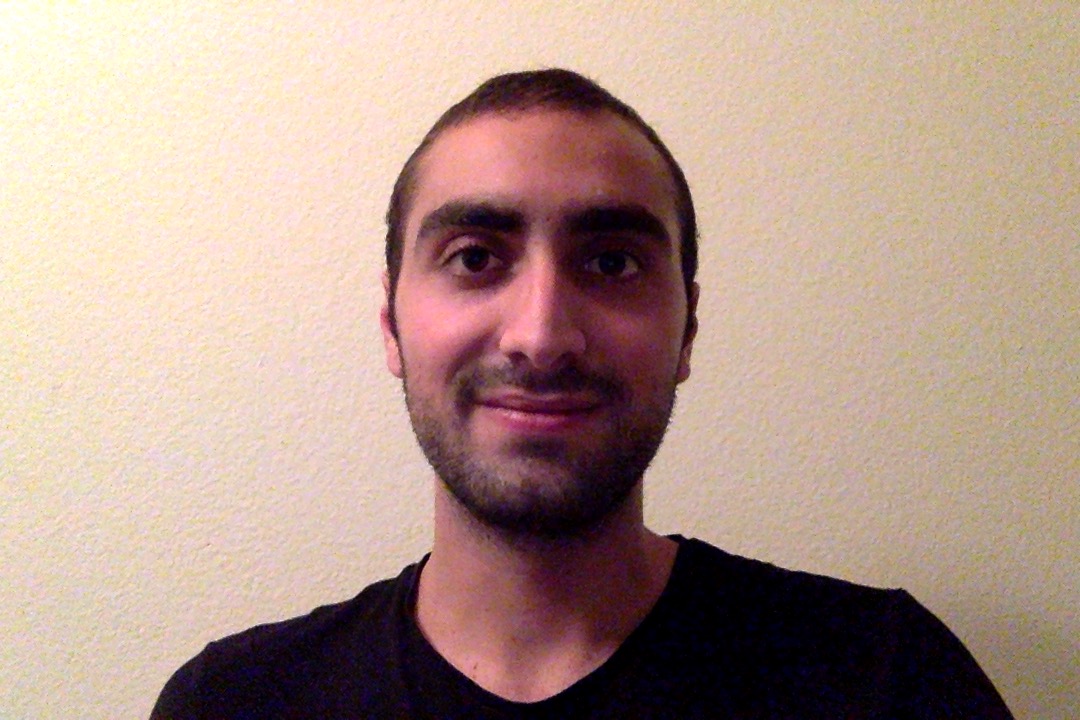}}]{Can Karakus}
received his B.S. degree from Bilkent University, Turkey, in 2011, his M.S. degree from University of California, Los Angeles (UCLA), USA, in 2013, and currently working towards his Ph.D. degree at UCLA, all in electrical engineering. He was a recipient of UCLA Graduate Division Fellowship in 2011, and UCLA Electrical Engineering Department Fellowship in 2013. His research interests include information theory, wireless networks and algorithms.
\end{IEEEbiography}

\begin{IEEEbiography}[{\includegraphics[width=1in,height=1.25in,clip,keepaspectratio]{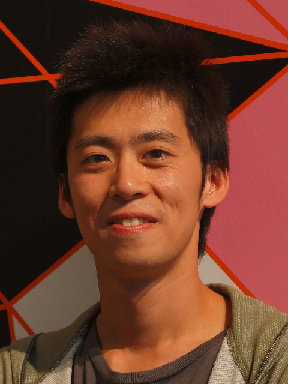}}]{I-Hsiang Wang}
received the B.Sc. degree in electrical engineering from National Taiwan University, Taiwan, in 2006. He received a Ph.D. degree in electrical engineering and computer sciences from the University of California at Berkeley, USA, in 2011. From 2011 to 2013, he was a postdoctoral researcher at \`{E}cole Polytechnique F\`{e}d\`{e}rale de Lausanne, Switzerland. Since 2013, he has been at the Department of Electrical Engineering in National Taiwan University, where he is now an assistant professor. His research interests include network information theory, wireless networks, coding theory, and network coding. He received a 2-year Vodafone Graduate Fellowship in 2006. He was a finalist of the Best Student Paper Award of IEEE International Symposium on Information Theory, 2011.
\end{IEEEbiography}

\begin{IEEEbiography}[{\includegraphics[width=1in,height=1.25in,clip,keepaspectratio]{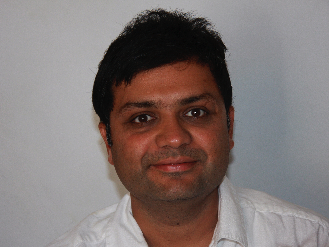}}]{Suhas Diggavi}
received the B. Tech. degree in electrical
engineering from the Indian Institute of Technology, Delhi, India, and
the Ph.D. degree in electrical engineering from Stanford University,
Stanford, CA, in 1998.

After completing his Ph.D., he was a Principal Member Technical Staff
in the Information Sciences Center, AT\&T Shannon Laboratories, Florham
Park, NJ. After that he was on the faculty of the School of Computer
and Communication Sciences, EPFL, where he directed the Laboratory for
Information and Communication Systems (LICOS).  He is currently a
Professor, in the Department of Electrical Engineering, at the
University of California, Los Angeles, where he directs the
Information Theory and Systems laboratory.  His research interests
include wireless network information theory, wireless networking
systems, network data compression and network algorithms.

He is a co-recipient of the 2013 IEEE Information Theory Society \&
Communications Society Joint Paper Award, the 2013 ACM International
Symposium on Mobile Ad Hoc Networking and Computing (MobiHoc) best
paper award, the 2006 IEEE Donald Fink prize paper award, 2005 IEEE
Vehicular Technology Conference best paper award and the Okawa
foundation research. He is a Fellow of the IEEE, a distinguished
lecturer for the information theory society and has served on the
editorial board for Transactions on Information Theory, ACM/IEEE
Transactions on Networking, IEEE Communication Letters, a guest editor
for IEEE Selected Topics in Signal Processing. He served as the
Technical Program Co-Chair for 2012 IEEE Information Theory Workshop
(ITW) and the Technical Program Co-Chair for the 2015 IEEE International
Symposium on Information Theory (ISIT).
\end{IEEEbiography}
\vfill

\end{document}

%% file: commands.tex
\usepackage{xifthen}

\newtheorem{lemma}{Lemma}[section]
\newtheorem{theorem}{Theorem}[section]
\newtheorem{corollary}{Corollary}[section]
\newtheorem{remark}{Remark}[section]
\newtheorem{claim}{Claim}[section]

\renewcommand{\thetheorem}{\arabic{section}.\arabic{theorem}}
\renewcommand{\thelemma}{\arabic{section}.\arabic{lemma}}
\renewcommand{\thecorollary}{\arabic{section}.\arabic{corollary}}
\renewcommand{\theremark}{\arabic{section}.\arabic{remark}}
\renewcommand{\theclaim}{\arabic{section}.\arabic{claim}}

\newcommand{\SNR}{\mathsf{SNR}}
\newcommand{\INR}{\mathsf{INR}}
\newcommand{\aaaa}{\mathrm{(a)}}
\newcommand{\bbbb}{\mathrm{(b)}}
\newcommand{\cccc}{\mathrm{(c)}}
\newcommand{\dddd}{\mathrm{(d)}}
\newcommand{\eeee}{\mathrm{(e)}}
\newcommand{\ffff}{\mathrm{(f)}}
\newcommand{\gggg}{\mathrm{(g)}}
\newcommand{\hhhh}{\mathrm{(h)}}

\newcommand{\E}[1]{\mathbb{E} \left[ #1 \right]}
\newcommand{\Prob}[1]{\mathbb{P} \left( #1 \right)}
\newcommand{\lp}{\left(}
\newcommand{\rp}{\right)}
\newcommand{\lb}{\left[}
\newcommand{\rb}{\right]}
\newcommand{\lbp}{\left\{}
\newcommand{\rbp}{\right\}}
\newcommand{\ul}{\underline}

\newcommand{\mcal}{\mathcal}

\newcommand{\what}{\widehat}
\newcommand{\wtild}{\widetilde}
\newcommand{\mb}{\mathbf}
\newcommand{\mbb}{\mathbb}
\newcommand{\msf}{\mathsf}

\newcommand{\ijj}{(i,j)=(1,2),(2,1)}
\newcommand{\eqFunc}{\overset{\mathrm{f}}{=}}

\newcommand{\Es}[1]{\mathbb{E}_{S^N} \left[ #1 \right]}
\newcommand{\hs}[1]{h_S \left( #1 \right)}
\newcommand{\hsc}[1]{h_S \left( #1 \right)}

\newcommand{\cgauss}[2]{\mathcal{CN}\lp#1,#2\rp}
\newcommand{\intypset}[1]{\left( #1 \right) \in \mathcal{T}^{(N)}_\epsilon}
\newcommand{\nintypset}[1]{\left( #1 \right) \notin \mathcal{T}^{(N)}_\epsilon}

\newcommand{\ind}[1]{\mathds{1}_{\left \{ #1 \right\} }}
\newcommand{\inds}[1]{\mathds{1}_{ #1  }}

%% file: Introduction.tex
The simplest information theoretic model for studying interference is
the two-user Gaussian \emph{interference channel} (IC). It has been
shown that feedback can provide an unbounded gain in capacity for
two-user Gaussian interference channels \cite{SuhTse_11}, in contrast
to point-to-point memoryless channels, where feedback gives no
capacity gain \cite{Shannon_56}, and multiple-access channels, where
feedback can at most provide \emph{power} gain \cite{Ozarow_84}. This
has been demonstrated when the feedback is unlimited, perfect, and
free of cost in \cite{SuhTse_11}.  Given the optimistic result
obtained under this setting, a natural question arises: Can feedback
be leveraged for interference management under imperfect feedback
models?

There have been several pieces of work so far, attempting to answer
this question. Vahid \emph{et al.} \cite{VahidSuh_12} considered a
rate-limited feedback model, where the feedback links are modeled as
fixed-capacity deterministic bit pipes. They developed a scheme based
on decode-and-forward at transmitters and lattice coding to extract
the helping information in the feedback links, and showed that it
achieves the sum-capacity to within a constant gap. The work in
\cite{LeTandon_12} studied a deterministic model motivated by passive
feedback over AWGN channels, and \cite{SahaiAggarwal_09, SuhWang_12} studied
the two-way interference channel, where the feedback is provided
through a backward interference channel that occupies the same
resource as the forward channel. \cite{LeTandon_12, SahaiAggarwal_09} and
\cite{SuhWang_12} only dealt with the linear deterministic model
\cite{AvestimehrDiggavi_09} of the Gaussian IC.

In this paper, we investigate how to exploit \emph{intermittent}
feedback for managing interference \cite{KarakusWang_13, KarakusWang_13_2}.  Such intermittent feedback could
occur in several situations. For example, one could use a side-channel
such as WiFi for feedback; in this case since the WiFi channel is best-effort, 
dropped packets might cause intermittent feedback.  In other
situations, control mechanisms in higher network layers could cause
the feedback resource to be available intermittently. For the feedback links, Bernoulli
processes $\{S_1[t]\}$ and $\{S_2[t]\}$ control the presence of
feedback for user $1$ and $2$, respectively. The two processes can be
dependent, but their joint distribution is i.i.d. over time. We assume
that the receivers are \emph{passive}: they simply feed back their
received signals to the transmitters without any processing. In
other words, each transmitter receives from feedback an observation of
the channel output of its own receiver through an erasure channel,
with unit delay. We focus on the passive feedback model as the
intermittence of feedback is motivated by the availability of feedback
resources (either through use of best-effort WiFi for feedback or
through feedback resource scheduling). Therefore, it might be that the
time-variant statistics of the intermittent feedback are not \emph{a
  priori} available at the receiver, precluding active
coding\footnote{In general, the statistics of $S_1[t]$ and $S_2[t]$ can have arbitrary time-
dependence, which could be unavailable at the receivers before feedback
transmission, but this information could be learned after the transmission. This
implies that receivers may not be able to actively code the feedback signals, but the
feedback statistics can potentially be used at the transmitters and later on at the receivers, after 
the feedback transmissions (therefore the transmitters could use these statistics to encode after receiving feedback). In this work, we focus on the case where feedback statistics
is time-invariant for simplicity, but the schemes described here can be easily
extended into the time-variant case.}. Moreover, the availability of the feedback resource may not be
known ahead of transmission, therefore motivating the assumption of
causal state information at the transmitter. If the receiver has
\emph{a priori} information about the feedback channel statistics, it
can perform active coding, in which case, the intermittent feedback
model reduces to the rate-limited model of \cite{VahidSuh_12}.

We study the
effect of intermittent feedback for the two-user Gaussian IC inspired
by ideas we develop for the linear deterministic IC model
\cite{AvestimehrDiggavi_09}. Our main contribution is the approximate characterization of the
capacity region of the interference channel with intermittent feedback, 
under the Gaussian model. We also derive an exact characterization of the capacity
region under the linear deterministic model, which agrees with the Gaussian result. 
The capacity characterizations under both models depend only on 
the forward channel parameters and the marginal
distributions of $S_1$ and $S_2$; not on their joint distribution. 

Our result shows that feedback can be harnessed to provide
multiplicative gain in Gaussian interference channel capacity even
when it is unreliable and intermittent. The result can be interpreted using the picture given in Figure~\ref{fig:gdof_duplicate}, which is depicted (for convenience) in
terms of symmetric generalized degrees of freedom for the special case of symmetric channel parameters. The given GDoF curves suggest that as the feedback probability increases, the achievable GDoF also increases for all interference regimes for which perfect feedback provides any GDoF gain. One can also observe from the figure that the capacity gain from intermittent feedback, which depends on the portion of time when
the feedback is active, remains unbounded, similar to the perfect feedback case.

\begin{figure}[t]
\centering
\includegraphics[scale=0.58]{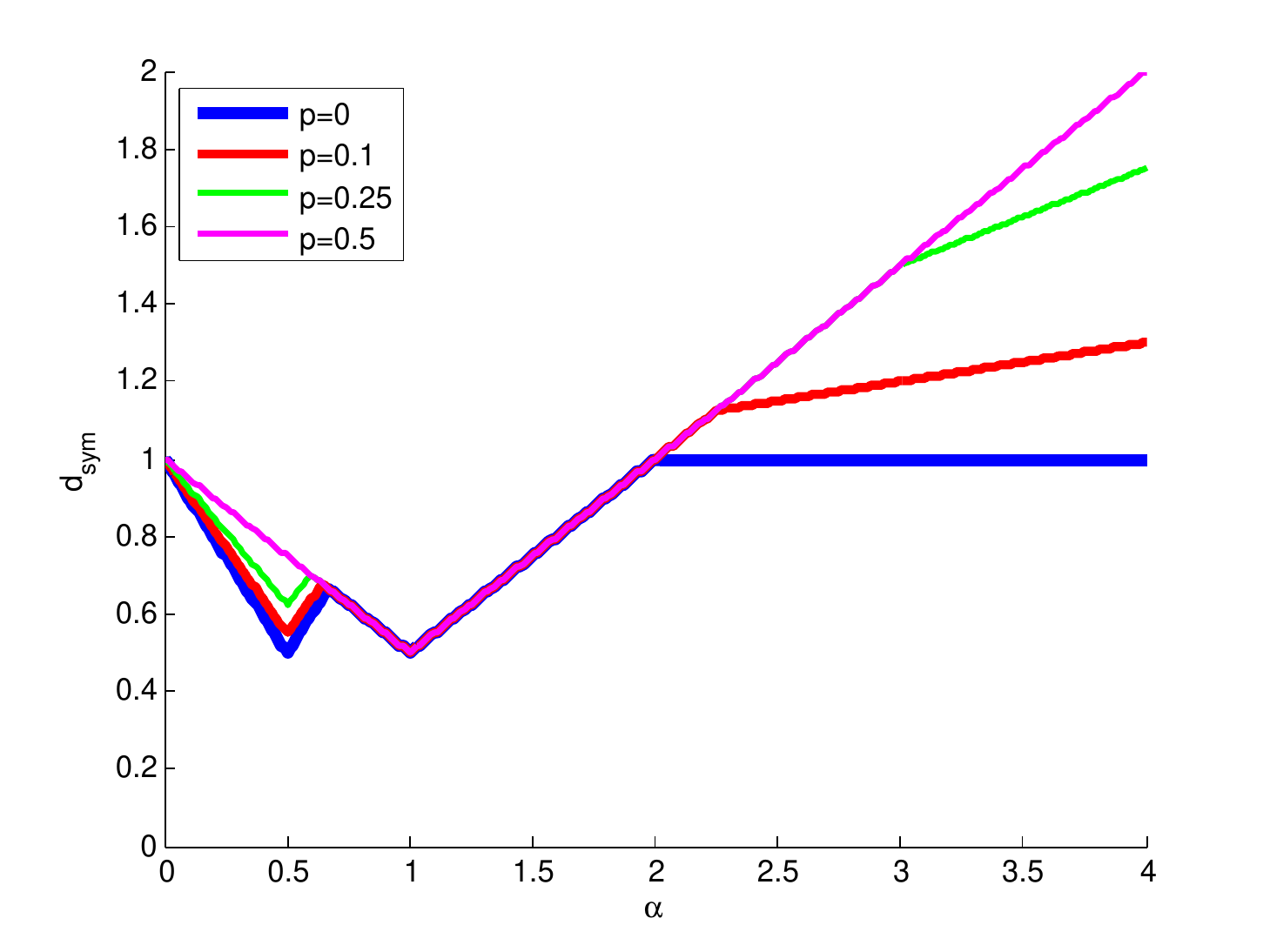}
\caption{Generalized degrees of freedom per user with respect to interference strength $\alpha := \frac{\log \INR}{\log \SNR}$ for symmetric channel parameters.}
\label{fig:gdof_duplicate}
\end{figure}

A consequence of this result is that when the feedback links are active with large enough probabilities, the
sum-capacity of the perfect feedback channel can be achieved to within
a constant gap. Similarly for the linear deterministic case, the
perfect feedback capacity is exactly achieved even when there is only
intermittent feedback, with large enough ``on'' probability. In particular, under the symmetric setting, this
threshold is $1/2$ for each feedback link. This is also reflected in Figure~\ref{fig:gdof_duplicate}, where the ``V-curve'' achievable with perfect feedback is already achievable when the feedback probability is only $1/2$.

Our achievable scheme has three main differences from the previous
schemes developed in \cite{SuhTse_11, VahidSuh_12} and
\cite{LeTandon_12}. First, we use \emph{quantize-map-and-forward} (QMF)\footnote{The QMF scheme of \cite{AvestimehrDiggavi_09} was
  generalized to DMCs in \cite{LimKim_11} (and the scheme was called
  noisy network coding) and to lattices in \cite{OzgurDiggavi_10, OzgurDiggavi_13}. In this
  paper we use the ``short-messaging'' version of QMF
  \cite{KramerHou_11} instead of the ``long-messaging'' version first
  studied in \cite{AvestimehrDiggavi_09} and extended to DMCs in \cite{LimKim_11}. 
For a longer discussion about this and other issues, refer to Section \ref{sec:discussion}.}
\cite{AvestimehrDiggavi_09} at the transmitters to send the
information obtained through feedback, as opposed to (partial or
complete) decode-and-forward, which has been used in \cite{SuhTse_11,
VahidSuh_12, LeTandon_12}. This is because when there is intermittent
feedback, the transmitters might not be able to decode the other
user's (partial) message, but would still need to send useful information about the
interference. A similar situation arises in a relay network, where QMF enables
forwarding of evidence, without requiring decoding
\cite{AvestimehrDiggavi_09}. Second, at the receivers, we perform
forward decoding of blocks instead of backward decoding, which results
in a better delay performance. Third, we do not use structured codes,
\emph{i.e.}, we only perform random coding.

We also develop novel outer bounds that are within a constant of the
achievable rate region for the Gaussian IC and match the achievable region
for the linear deterministic IC. These outer bounds are based on
constructing an enhanced channel and appropriate side-information.
These are illustrated in Section \ref{sec:converse}.

Lastly, we extend these results for packet transmission channels, modeled
through parallel channels which are $M$-symbol extensions of the original
model. This can be considered as a model for OFDM and packet drops over
a best-effort channel.

The rest of this paper is organized as follows. We formally state the
problem and establish the notation in Section~\ref{sec:model}. We
present our main results in Section~\ref{sec:results} and give
interpretations of them. We motivate our coding scheme and explain it
through an example in Section \ref{sec:motivation}. We give the
analysis of the coding scheme in Section \ref{sec:achievability}.  The
outer bound is developed in Section~\ref{sec:converse} and
Section~\ref{sec:discussion} concludes the paper with a brief discussion of
possible extensions of the work.  Many of the detailed proofs are
given in the Appendices.

%% file: Model.tex
\begin{figure} 
\centering
\includegraphics[scale=0.75]{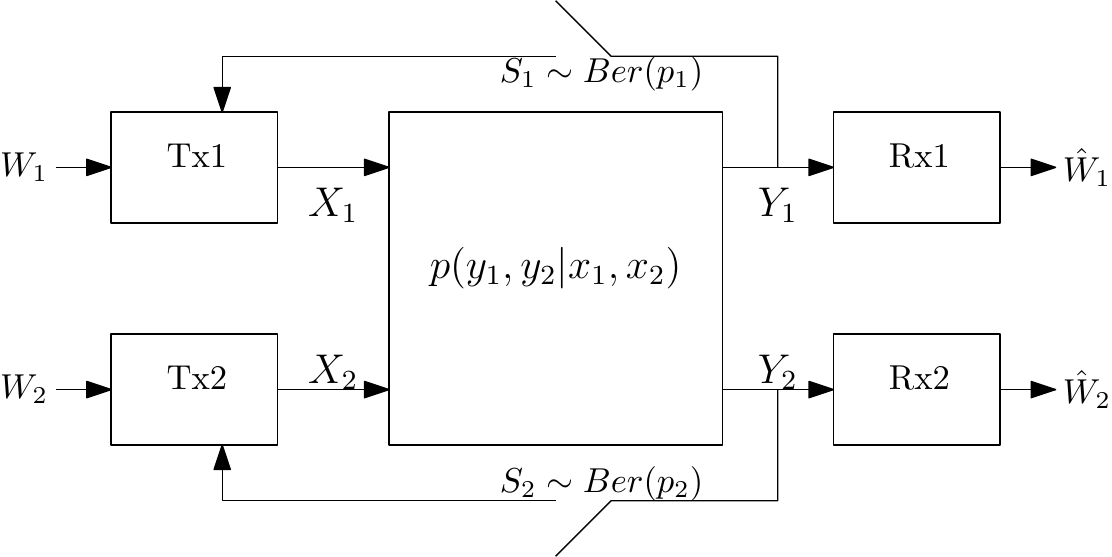}
\centering
\caption{Two-user discrete memoryless interference channel with intermittent feedback} 
\label{fig:dm_ic}
\end{figure}

We consider the 2-user discrete memoryless interference channel (DM-IC) with intermittent feedback, illustrated in Figure~\ref{fig:dm_ic}. We assume Transmitter $i$ (Tx$i$) has a message $W_i$ intended for Receiver $i$ (Rx$i$), $i=1,2$. $W_1 \in \lb 2^{NR_1} \rb$ and $W_2 \in \lb 2^{NR_2} \rb$ are independent and uniformly distributed, where, for $n \in \mbb{N}$, $\lb n \rb := \lbp k \in \mbb{N}: k \leq n\rbp$. The signal transmitted by Tx$i$ at time $t$ is denoted by $X_{i,t} \in \mcal{X}_i$, while the channel output observed at Rx$i$ is denoted by $Y_{i,t} \in \mcal{Y}_i$, for $i=1,2$. For a block length $N$, the conditional probability distribution mapping the input codeword to the output sequence is given by
\begin{align*}
p(Y_1^N, Y_2^N | X_1^N, X_2^N) = \prod_{t=1}^N p \lp Y_{1,t}, Y_{2,t} | X_{1,t}, X_{2,t} \rp
\end{align*}
The feedback state sequence pair $\ul{S}:=\lp S_1^N, S_2^N \rp$ have the joint distribution
\begin{align*}
p\lp S_1^N, S_2^N\rp = \prod_{t=1}^N p\lp S_{1,t}, S_{2,t}\rp.
\end{align*}
and marginally, at time $t$, $S_{i,t} \sim Bernoulli (p_i)$, for $i=1,2$, for all $t$ and $N$. Note that, for any fixed time slot $t$, the random variables $S_{1,t}$ and $S_{2,t}$ are not necessarily independent, that is, the joint distribution $p(S_{1,t}, S_{2,t})$ can be arbitrary. We assume that receivers have access to $\ul{S}$ strictly causally, that is, at time $t$, both receivers know the realization of $\ul{S}^{t-1}$.

At the beginning of time $t$, Tx$i$ observes the channel output received by Rx$i$ at time $t-1$ through an erasure channel, \emph{i.e.}, it receives $\wtild Y_{i,t-1} := S_{i,t-1} Y_{i,t-1}$, for $i=1,2$. Note that this is a \emph{passive} feedback model, in that it does not allow the receiver to perform any processing on the channel output; it simply forwards the received signal $Y_i$ at every time slot, which gets erased with probability $1-p_i$.

For random variables $A$ and $B$, we use the notation $A \eqFunc B$ to denote that $A$ is a deterministic function of $B$\footnote{More formally, $A \eqFunc B$ means that there exists a $\sigma (B)$-measurable function $f$ such that $A=f(B)$ almost surely, where $\sigma (B)$ is the sigma-algebra generated by $B$.}. Then our channel model implies $X_{i,t} \eqFunc \lp W_i, S_i^{t-1}, \wtild Y_i^{t-1} \rp$. 

\begin{figure}
  \centering
  \includegraphics[scale=0.45]{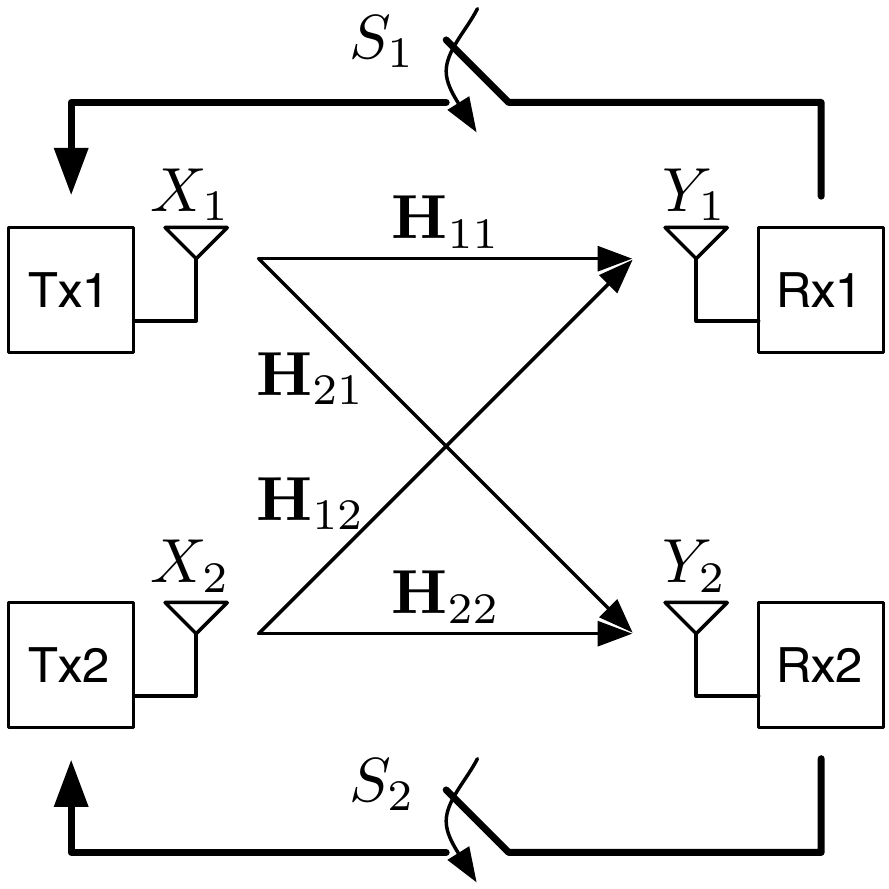}
  \caption{Two-user linear deterministic interference channel with intermittent feedback}
  \label{fig:ldc_ic}
\end{figure}

\begin{figure}
  \centering
  \includegraphics[scale=0.45]{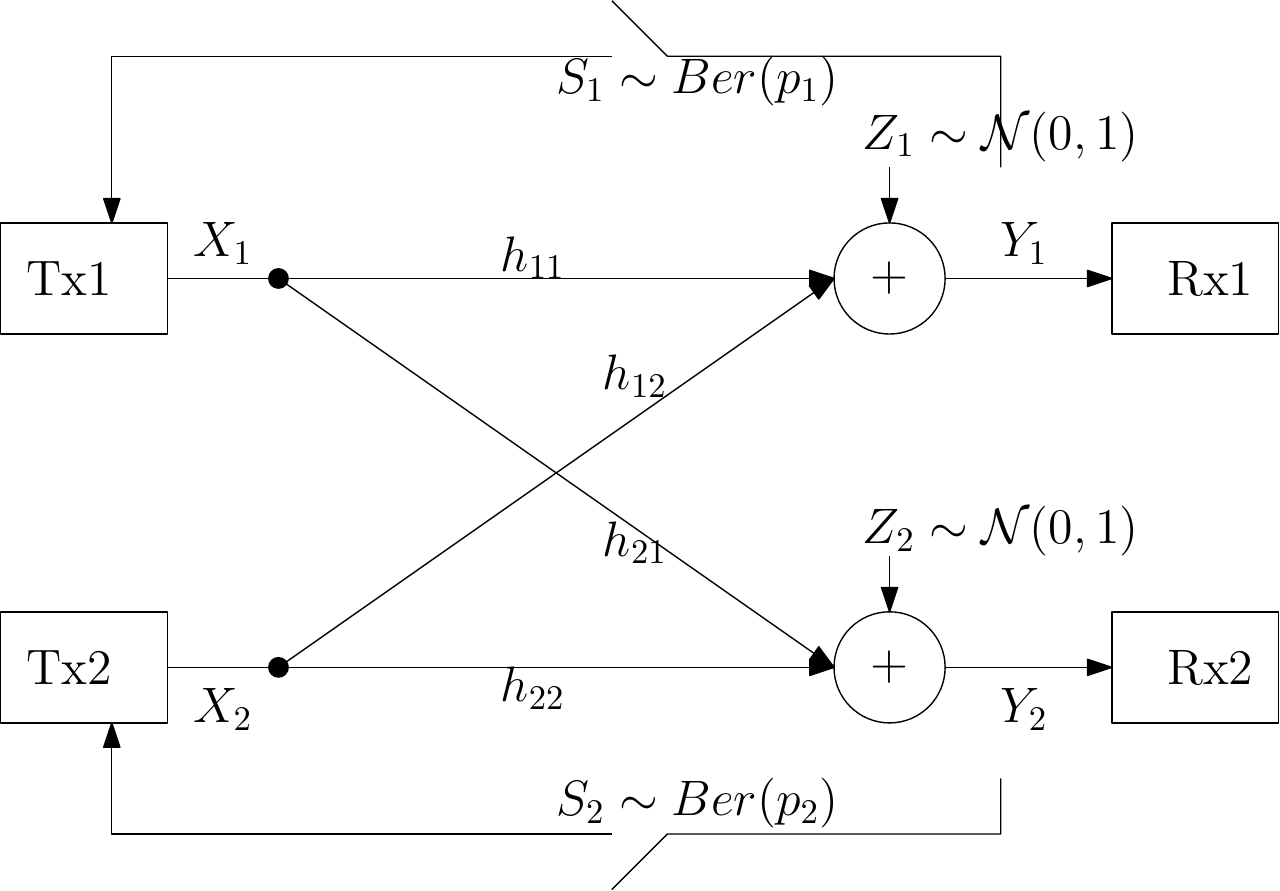}
  \caption{Two-user Gaussian interference channel with intermittent feedback}
  \label{fig:gaussian_ic}
\end{figure}

A rate pair $(R_1, R_2)$ is said to be achievable if there exists a pair of codebooks $(\mcal{C}_1, \mcal{C}_2)$ at Tx1 and Tx2, with rates $R_1$ and $R_2$, respectively, and pairs of encoding and decoding functions such that the average probability of error at any decoder goes to zero as the block length $N$ goes to infinity. The capacity region with feedback probabilities $p_1$ and $p_2$, $\mcal{C} (p_1,p_2)$, is defined as the closure of the set of all achievable rate pairs $(R_1, R_2)$ when $S_1 \sim Bernoulli(p_1)$ and $S_2 \sim Bernoulli(p_2)$. Sum-capacity is defined by
\begin{align*}
C^{\text{sum}} (p_1, p_2) := \sup \lbp R_1+R_2: \lp R_1, R_2 \rp \in \mcal{C} (p_1, p_2) \rbp.
\end{align*}
In this work, we consider two specific channel models (that is, two specific classes of $\lp \mcal{X}_1, \mcal{X}_2, \mcal{Y}_1, \mcal{Y}_2, p\lp y_1,y_2|x_1,x_2\rp \rp$), described in the following subsections.
%


\subsection{Linear Deterministic Model}
This channel model was introduced in \cite{AvestimehrDiggavi_09} and since then proved useful in providing insight into the nature of signal interactions many network information theory problems (see Figure~\ref{fig:ldc_ic}). 

We assume $X_{i,t} \in \mbb{F}_2^{q}$, for $i=1,2$, where $\mbb{F}_2$ is the binary field. The received signal at Rx$i$ is given by
\begin{align*}
Y_{i,t} &= \mb{H}_{ii} X_{i,t} + \mb{H}_{ij} X_{j,t}
\end{align*}
for $\ijj$. The channel matrices are given by $\mb{H}_{ij} := \mb{S}^{q-n_{ij}}$ for $(i,j) \in \{1,2\}^2$, where $q = \max\lbp n_{11},n_{12},n_{21},n_{22}\rbp$, and $\mb{S}\in\mathbb{F}_2^{q\times q}$ is the shift matrix
$
\begin{bmatrix}
\mb{0}^T & 0\\ \mb{I}_{q-1} & \mb{0}
\end{bmatrix}
$, where $\mb{0}$ is the zero vector in $\mbb{F}_2^{q-1}$ and $\mb{I}_{q-1}$ is the identity matrix in $\mbb{F}_2^{(q-1)\times(q-1)}$. We also define, for $\ijj$,
\begin{align*}
V_{i,t} = \mb{H}_{ji} X_{i,t}.
\end{align*}
The capacity region for the linear deterministic model will be denoted by $\mcal{C}_{LDC} (p_1,p_2)$, while its sum-capacity will be denoted by $C_{LDC}^{\text{sum}} (p_1,p_2)$.

\begin{figure*}[!t]
\begin{align}
R_1 &\le \min\lbp \max(n_{11},n_{12}), n_{11}+p_2(n_{21}-n_{11})^+\rbp \label{eq:ldc_R1}\\
R_2 &\le \min\lbp \max(n_{22},n_{21}), n_{22}+p_1(n_{12}-n_{22})^+\rbp \label{eq:ldc_R2}\\
R_1+R_2 &\le \min\Big\{ \max(n_{11},n_{12}) + (n_{22}-n_{12})^+ , \max(n_{22},n_{21}) + (n_{11}-n_{21})^+ \Big\}\label{eq:ldc_R1R2_1}\\
R_1+R_2 &\le \max\lbp n_{12}, (n_{11}-n_{21})^+\rbp + \max\lbp n_{21}, (n_{22}-n_{12})^+\rbp \notag\\
&\quad + p_1\min\lbp n_{12}, (n_{11}-n_{21})^+\rbp + p_2\min\lbp n_{21}, (n_{22}-n_{12})^+\rbp \label{eq:ldc_R1R2_3}\\
2R_1+R_2 &\le \max(n_{11},n_{12}) + \max\lbp n_{21}, (n_{22}-n_{12})^+\rbp + (n_{11}-n_{21})^+ + p_2\min\lbp n_{21}, (n_{22}-n_{12})^+\rbp \label{eq:ldc_2R1R2}\\
R_1+2R_2 &\le \max(n_{22},n_{21}) + \max\lbp n_{12}, (n_{11}-n_{21})^+\rbp + (n_{22}-n_{12})^+ + p_1\min\lbp n_{12}, (n_{11}-n_{21})^+\rbp \label{eq:ldc_R12R2}
\end{align}
\hrulefill
\end{figure*}

\subsection{Gaussian Model}
Under the canonical Gaussian model (see Figure~\ref{fig:gaussian_ic}), the channel outputs are related to the inputs through the equations
\begin{align*}
Y_{1,t} &= h_{11} X_{1,t} + h_{12} X_{2,t} + Z_{1,t} \\
Y_{2,t} &= h_{21} X_{1,t} + h_{22} X_{2,t} + Z_{2,t}
\end{align*}
where $h_{ij} \in \mbb{C}$, for $\lp i,j\rp \in \lbp 1,2\rbp^2$, are channel gains, and $Z_{1,t}, Z_{2,t} \sim \cgauss{0}{1}$ are circularly symmetric complex white Gaussian noise. We assume an average transmit power constraint of $P_i$ at Tx$i$, \emph{i.e.}, for any length-$N$ codeword $X_i^N$ transmitted by Tx$i$, $\frac{1}{N}\sum_{t=1}^N \left|X_{i,t}\right|^2 \leq P_i$, $i=1,2$. We also define
\begin{align*}
\SNR_i &:= |h_{ii}|^2 P_i \\
\INR_i &:= |h_{ij}|^2 P_j 
\end{align*}
and
\begin{align*}
V_i &:= h_{ji}X_i + Z_j, \\
\wtild V_i &:= S_j V_i,
\end{align*}
for $\ijj$. Note that this definition of $V_{i,t}$ is consistent with its definition under linear deterministic model, in the sense that it is what remains out of the channel output when the intended signal is completely cancelled.

The capacity region for the Gaussian model will be denoted by $\mcal{C}_{G} (p_1,p_2)$, while its sum-capacity will be denoted by $C_{G}^{\text{sum}} (p_1,p_2)$. We will also use the notation $C_{G,p}^{\text{sum}} := C_{G}^{\text{sum}} (1,1)$, denoting the sum-capacity under perfect feedback.

Gaussian parallel channel is described by the equations
\begin{align}
\mb{Y}_{1,t} &= h_{11}\mb{X}_{1,t} + h_{12} \mb{X}_{2,t} + \mb{Z}_{1,t} \label{eq:parallel_model_first} \\
\mb{Y}_{2,t} &= h_{21}\mb{X}_{2,t} + h_{22} \mb{X}_{2,t} + \mb{Z}_{2,t} \\
\mb{\wtild Y}_{1,t} &= S_{1,t} \mb{Y}_{1,t} \\
\mb{\wtild Y}_{2,t} &= S_{2,t} \mb{Y}_{2,t} \label{eq:parallel_model_last} 
\end{align}
where $\mb{X}_{i,t}, \mb{Y}_{i,t} \in \mbb{C}^M$, $i=1,2$, are the channel input and output, respectively, at user $i$; $\mb{Z}_{1,t}$ and $\mb{Z}_{2,t}$ are independent and distributed with $\cgauss{\mb{0}}{\mb{I}}$; and $\mb{\wtild Y}_{i,t}, i=1,2$ is the output of the feedback channel of Tx$i$, at time $t$. Note that the channel gains are scalars. It should also be noted that any given time, the same feedback state variable $S_{i,t}$ controls the presence of feedback for all sub-channels, \emph{i.e.}, the feedback is present either for all $M$ channels, or for none of them.

%% file: Result.tex

In this section, we present our results and discuss their consequences for both linear deterministic and Gaussian models.

%
%
%
\subsection{Linear Deterministic Model}
The following theorem captures our main result for the linear deterministic model. 
\begin{theorem} \label{th:ldc}
The capacity region $\mcal{C}_{LDC} (p_1, p_2)$ of the linear deterministic interference channel with intermittent feedback is given by the set of rate pairs $\lp R_1,R_2\rp$ satisfying \eqref{eq:ldc_R1}--\eqref{eq:ldc_R12R2}.
\end{theorem}
\begin{proof}
See Section~\ref{sec:achievability} for achievability, and Section~\ref{sec:converse} for converse.
\end{proof}
{Note that for the special cases of $p_1=p_2=1$ and $p_1=1,p_2=0$, existing results in the literature \cite{SuhTse_11}, \cite{SahaiAggarwal_09} are recovered.} The following corollary shows that it is possible to achieve perfect feedback sum-capacity even when feedback probabilities are less than one.
\begin{corollary} \label{cor:threshold_ldc}
For $n_{12},n_{21}>0$, there exists $p^* < 1$ such that
\begin{align*}
C_{LDC}^{\text{sum}} (p_1,p_2) = C_{LDC}^{\text{sum}} (1,1)
\end{align*}
for all $p_1,p_2 \geq p^*$.
\end{corollary}
\begin{proof}
See Appendix~\ref{sec:ap_threshold}.
\end{proof}
We illustrate Corollary~\ref{cor:threshold_ldc} through an example. Let us assume $n_{12}=n_{21}=m$, $n_{11}=n_{22}=n$, and $p_1=p_2=p$. It is easy to see that if $p_1=p_2=0.5$, the bounds on $R_1+R_2$, $2R_1+R_2$ and $R_1+2R_2$ that involve $p_1$ and $p_2$ become redundant, and the sum-capacity does not increase beyond this point, for all $\lp m,n \rp$.

\begin{figure*}[!t]
\begin{align}
R_i &< \min \lbp \log \lp 1 + \SNR_i + \INR_i \rp,  \log \lp 1 + \SNR_i \rp + p_j \log \lp 1 + \frac{\INR_j}{1+\SNR_i} \rp \rbp \label{eq:g_Ri}\\
R_i+R_j &< \log \lp 1 + \frac{\SNR_i}{1+\INR_j} \rp + \log \lp 1 + \SNR_j + \INR_j \rp \label{eq:g_RiRj_1}\\
R_i+R_j &< \log \lp 1 + \frac{\SNR_i}{1+\INR_j} + \INR_i \rp + \log \lp 1 + \frac{\SNR_i}{1+\INR_j} + \INR_i \rp \notag\\
&\qquad + p_i \log \lp \frac{\lp 1+\INR_i \rp\lp 1+\frac{\SNR_i}{1+\INR_j}\rp}{1+\frac{\SNR_i}{1+\INR_j}+\INR_i} \rp+ p_j \log \lp \frac{\lp 1+\INR_j \rp\lp 1+\frac{\SNR_j}{1+\INR_i}\rp}{1+\frac{\SNR_j}{1+\INR_i}+\INR_j} \rp \label{eq:g_RiRj_2}\\
2R_i + R_j &< \log \lp 1 + \frac{\SNR_i}{1+\INR_j} \rp + \log \lp 1 + \frac{\SNR_j}{1+\INR_i} + \INR_j \rp \notag\\
&\qquad + \log \lp 1 + \SNR_i + \INR_i \rp + p_j \log \lp \frac{\lp 1+\INR_j \rp\lp 1+\frac{\SNR_j}{1+\INR_i}\rp}{1+\frac{\SNR_j}{1+\INR_i}+\INR_j} \rp \label{eq:g_2RiRj}
\end{align}
\hrulefill
\end{figure*}

\subsection{Gaussian Model}
We define, for any set $\mcal{R}$ of rate pairs $\lp R_1,R_2\rp$ and scalar $\delta \in \mbb{R}$,
\begin{align*}
\mcal{R}-\delta &:= \lbp (R_1,R_2): (R_1+\delta,R_2+\delta) \in \mcal{R} \rbp, \\
\mcal{R}+\delta &:= \lbp (R_1,R_2): (R_1-\delta,R_2-\delta) \in \mcal{R} \rbp.
\end{align*}

The following theorem captures our main result for the Gaussian model.
\begin{theorem} \label{th:gaussian}
The capacity region $\mcal{C}_{G} (p_1, p_2)$ of the Gaussian interference channel with intermittent feedback satisfies 
\begin{align}
\mcal{\bar C}(p_1, p_2) - \delta_1 \subseteq \mcal{C}_G(p_1, p_2) \subseteq \mcal{\bar C}(p_1, p_2) + \delta_2 \label{eq:result}
\end{align}
where $\mcal{\bar C}\lp p_1,p_2\rp$ is the set of rate pairs satisfying \eqref{eq:g_Ri}--\eqref{eq:g_2RiRj} for $\ijj$ and $\delta_1 <2\log 3 + 3\lp p_1+p_2\rp$ bits, and $\delta_2 < \log 3 + p_1 + p_2 $ bits.
\end{theorem}
\begin{proof}
Section~\ref{sec:achievability} proves an inner bound region $\mcal{R}_G^i (p_1, p_2)$, Section~\ref{sec:converse} proves an outer bound region $\mcal{R}_G^o (p_1, p_2)$, and Appendix~\ref{sec:ap_gap} shows that $\mcal{\bar C}(p_1, p_2) - \delta_1 \subseteq \mcal{R}_G^i (p_1, p_2) $ and $\mcal{R}_G^o (p_1, p_2) - \delta_2 \subseteq \mcal{\bar C}(p_1, p_2)$.
\end{proof}
\begin{remark}
Theorem~\ref{th:gaussian} uniformly approximates the capacity region under Gaussian model to within a gap of $3\log 3 + 4\lp p_1+p_2\rp$ bits, independent of channel parameters. To our knowledge, this is the first constant-gap capacity region characterization for Gaussian interference channel with non-perfect feedback with arbitrary channel parameters.
\end{remark}
\begin{remark}
As will be seen in the achievability proof, the proposed coding scheme achieves a smaller gap than what is given in Theorem~\ref{th:gaussian}; however, for simplicity in the achievability proof, we lower bound the achievable rate terms with computationally more tractable ones, which articifically contributes to the claimed gap. Moreover, one can optimize over the parameters of the proposed coding scheme, such as power allocation and quantization distortion, to further reduce the gap, but this issue will not be dealt with in this paper.
\end{remark}

Theorem~\ref{th:gaussian} allows us to characterize the symmetric generalized degrees of freedom under symmetric channel parameters, which is a metric often used to compare the capabilities of the interference channel under different settings.

\begin{corollary}[Generalized Degrees of Freedom] \label{cor:gdof}
For symmetric channel parameters ($\SNR_1=\SNR_2=\SNR$, $\INR_1=\INR_2=\INR$, $p_1=p_2=p$), the symmetric generalized degrees of freedom of freedom, defined by
\begin{align*}
d_{\text{sym}} := \lim_{\substack{ \SNR \to \infty \\ \INR = \SNR^{\alpha}}} \frac{C_{\text{sym}}(\SNR, \INR, p)}{\log \SNR},
\end{align*}
where $C_{\text{sym}} (\SNR, \INR, p) := \sup \lbp R: (R, R) \in \mcal{C}_G (p, p) \rbp$, is given by 
\begin{align*}
d_{\text{sym}} = \lbp\begin{array}{ll}
\min\lbp 1-\alpha/2, 1-(1-p)\alpha\rbp, &\alpha \le 1/2\\
\min\lbp 1-\alpha/2, p+(1-p)\alpha\rbp, &1/2 \le \alpha \le 1\\
\min\lbp \alpha/2, (1-p) + p\alpha\rbp, & \alpha \ge 1
\end{array}\right. 
\end{align*}
\end{corollary}
Figure~\ref{fig:gdof_duplicate} plots the available generalized degrees of freedom with respect to interference strength for various values of $p$. As can be observed, as $p$ is increased, gradually better curves are obtained. It should be noted that once $p\geq 0.5$, the ``V-curve'' that is achieved by perfect feedback \cite{SuhTse_11} is already achieved. Next, this observation will be made precise.

The perfect feedback outer bound on the sum-capacity, $C_{G,p}^{\text{sum}}$, is given in Theorem 3 of \cite{SuhTse_11} as follows.
\begin{align*}
C_{G,p}^{\text{sum}} \leq &\sup \limits_{0\leq\rho\leq 1} \min_{(i,j)\in\mathcal{I}} \log \lp 1 + \frac{(1-\rho^2)\SNR_i}{1+(1-\rho^2)\INR_j} \rp \\
&\quad + \log \lp 1 + \SNR_j + \INR_j + 2\rho \sqrt{\SNR_j \cdot \INR_j} \rp
\end{align*}
where $\mathcal{I}=\lbp (1,2),(2,1)\rbp$. The next corollary shows that when $p_1$ and $p_2$ are sufficiently large, the sum-capacity of the perfect feedback Gaussian channel can be achieved with intermittent feedback, to within a constant gap. Hence, this corollary is the Gaussian counterpart of the similar result given in Corollary~\ref{cor:threshold_ldc}, for the linear deterministic channel.
\begin{corollary} \label{cor:threshold}
For $\INR_1,\INR_2>0$, there exists $p^*<1$ such that
\begin{align*}
C_{G,p}^{\text{sum}} - C_{G}^{\text{sum}} (p_1, p_2) \leq \delta_p
\end{align*}
for all $p_1, p_2 \geq p^*$, where $\delta_p$ is a constant independent of channel parameters.
\end{corollary}
\begin{proof}
See Appendix~\ref{sec:ap_threshold}.
\end{proof}
In our intermittent feedback model, erasures are symbol-wise, that is, each symbol can get erased independently of others. However, in a best-effort channel, erasures might occur on \emph{packet-level} instead. In order to study this scenario, we consider the parallel channel model described by the equations \eqref{eq:parallel_model_first}--\eqref{eq:parallel_model_last}, which is simply the $M$-symbol extension of the Gaussian channel, where the channel parameters are the same for each subchannel. Each extended symbol over this channel models a packet. The result in Theorem~\ref{th:gaussian} easily generalizes to parallel channel model, as shown by the following corollary. 

\begin{corollary}[Parallel channel] \label{cor:parallel}
The capacity region $\mcal{C}_G^{(M)}(p_1, p_2)$ of any parallel channel of size $M$ with feedback probabilities $p_1$ and $p_2$ satisfies
\begin{align*}
M\mcal{\bar C}(p_1, p_2) - M\delta_1 \subseteq \mcal{C}_G^{(M)}(p_1, p_2) \subseteq M\mcal{\bar C}(p_1, p_2) + M\delta_2
\end{align*}
where $\mcal{\bar C}(p_1, p_2),\delta_1$ and $\delta_2$ are as defined in Theorem \ref{th:gaussian}.
\end{corollary}
\begin{remark}
Although strictly speaking, the claim in Corollary~\ref{cor:parallel} is more general than that in Theorem~\ref{th:gaussian}, the achievability and converse proofs for the scalar channel directly extend to the parallel channel without any non-trivial modification. Hence, for simplicity, we focus on the scalar case in the paper, and omit a separate proof for the parallel channel.
\end{remark}

\subsection{Discussion of Results}
\subsubsection{Feedback Strategy}
Our result shows that even unreliable feedback provides multiplicative gain in interference channels. The key insight in showing this result is using quantize-map-forward as a feedback strategy at the transmitters. This is in contrast to the schemes proposed for perfect feedback \cite{SuhTse_11} and rate-limited feedback \cite{VahidSuh_12}, which use decode-and-forward to extract the feedback information. When the feedback channel is noisy\footnote{``Noise'' in this context refers to the erasures in the feedback channel.}, such schemes can result in rates arbitrarily far from optimality. In order to see this, consider unfolding the channel over time, as shown in Figure~\ref{fig:unfold}. This transformation effectively turns this channel into a relay network, where it is known that decode-and-forward based relaying schemes can give arbitrarily loose rates. This also motivates using quantize-map-forward as a feedback strategy, which has been shown to approximately achieve the relay network capacity \cite{AvestimehrDiggavi_09}. This observation also suggests that quantize-map-forward might be a promising feedback strategy for the additive white Gaussian noise (AWGN) feedback model of \cite{LeTandon_12} in order to uniformly achieve its capacity region to within a constant gap. 

\begin{figure}[!t]
\centering
\includegraphics[scale=0.48]{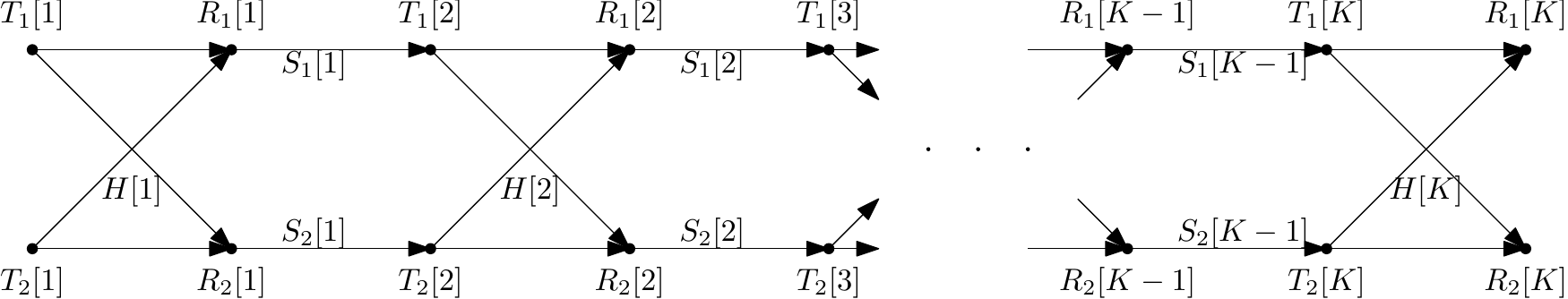}
\caption{The interference network unfolded over a block of $K$ time slots. The node $T_i[t]$ corresponds to the copy of Tx$i$ at time $t$, while $R_i[t]$ corresponds to the copy of Rx$i$ at time $t$. The feedback channel for time $t$ is an erasure channel controlled by $S_1[t]$ and $S_2[t]$, while the forward channel is a Gaussian interference channel with channel matrix $H[t]$.}
\label{fig:unfold}
\end{figure}

It is instructive to compare the achievable rate region for the case of $p_1=p_2=1$ with the outer bound region of the perfect feedback model of \cite{SuhTse_11}. Evaluating the region $\mcal{\bar C}(p_1, p_2) - \delta_1$ with $p_1=p_2=1$, we see that the perfect feedback bound \eqref{eq:g_RiRj_1} becomes redundant, and the achievable region comes within $\lp 3+3\log 3 \rp$ bits of the outer bound region of \cite{SuhTse_11} (see Appendix~\ref{sec:ap_gap} for details). We note that this gap is larger than what is achieved by the decode-and-forward based scheme of \cite{SuhTse_11}. This shows that uniform approximation of capacity region via quantize-map-forward comes at the expense of an additional (but constant) gap\footnote{Although we stated that the quantize-map-forward scheme achieves a smaller gap than what is claimed in Theorem~\ref{th:gaussian}, the actual gap is still expected to be larger than that of the decode-and-forward based scheme for perfect feedback, due to quantization distortion.}. The source of this additional gap is the quantization step at the transmitters, which introduces a distortion in the feedback signal, and eventually incurs a constant rate penalty whose amount depends on the distortion level.

\subsubsection{Perfect Feedback Sum Capacity with Intermittent Feedback}
Corollary~\ref{cor:threshold} shows that for any set of channel parameters, there exists a threshold $p^*$ on the feedback probability above which perfect feedback sum-capacity is achieved to within a constant gap. Although the exact closed-form expression of $p^*$ is not clean, an examination of the symmetric case (see Figure~\ref{fig:gdof}) reveals that in some cases it can be as low as 0.5.

The intuition behind this result lies in the fact that it takes the transmitter forward-channel resources to send the information obtained through feedback. Note that the larger $p$ is, the larger the amount of additional information about the past reception can be obtained through intermittent feedback at the transmitters. If the amount of such information is larger than a threshold, then sending it to the receivers will limit the rate for delivering fresh information. Hence, once this threshold is reached, having more feedback resource is no longer useful. However, this property is not observed for the entire capacity region, since if one of the users transmit at a low rate, then it will have sufficient slackness in rate to forward the entire feedback information.

\begin{figure}
\centering
\includegraphics[scale=0.5]{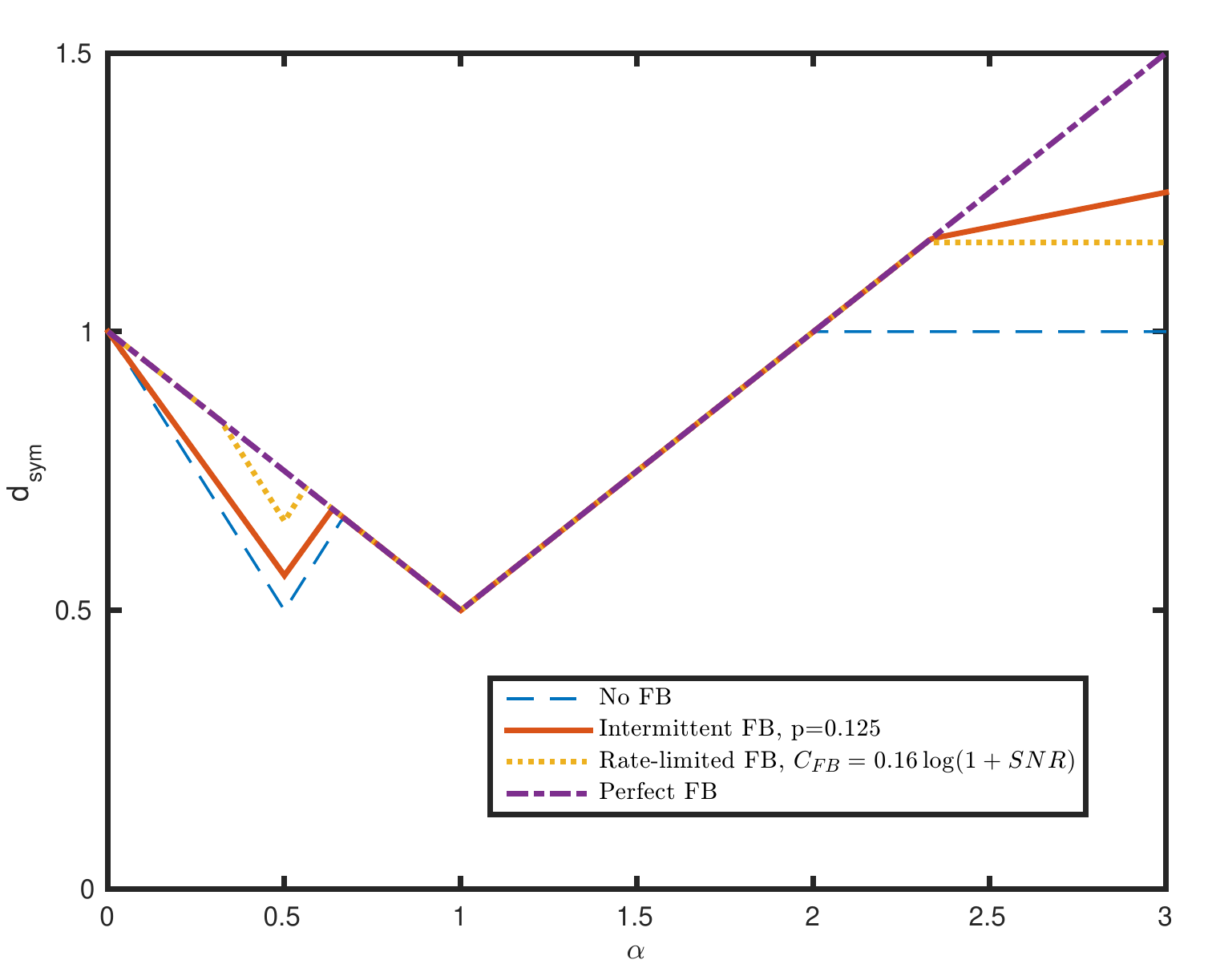}
\caption{Generalized degrees of freedom per user with respect to interference strength $\alpha := \frac{\log \INR}{\log \SNR}$ for symmetric channel parameters, for no feedback, intermittent feedback, rate-limited feedback and perfect feedback.}
\label{fig:gdof}
\end{figure}

\begin{figure*}[!t]
\centering
\includegraphics{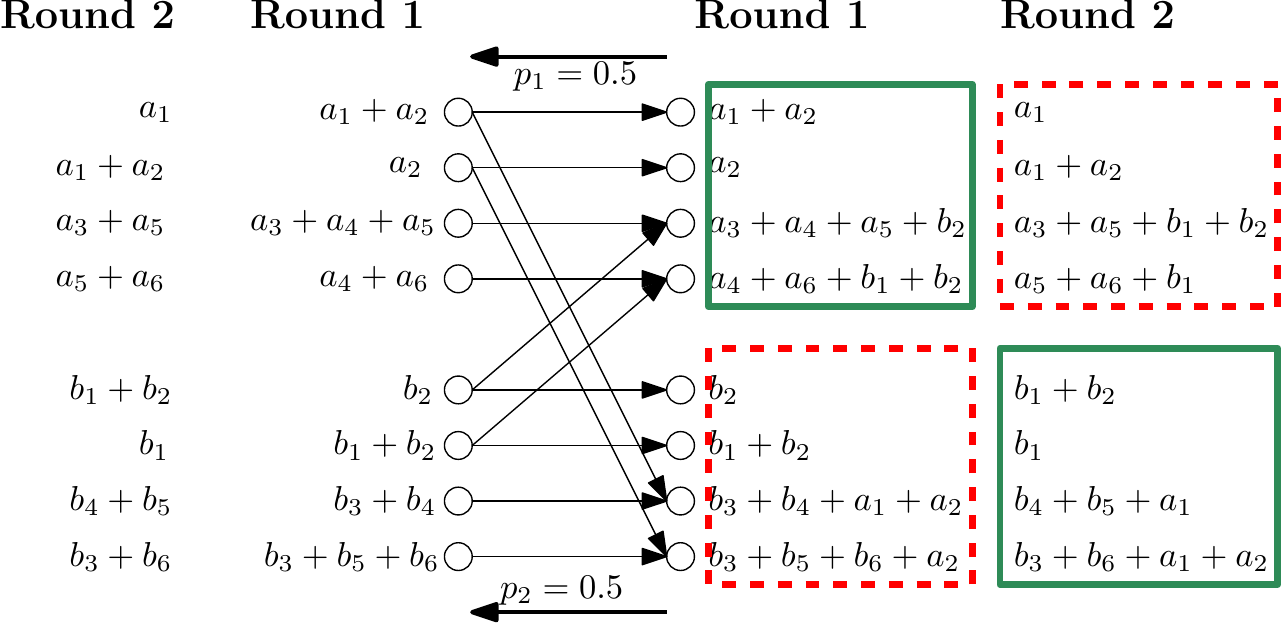}
\centering
\caption{First block of transmissions for the example coding scheme over linear deterministic channel. Receptions enclosed in green/solid rectangles represent the channel outputs that the receivers are able to feed back; whereas those enclosed in red/dashed rectangles represent the channel outputs that gets erased through the feedback channel.} 
\label{fig:ldc_example1}
\centering
\includegraphics{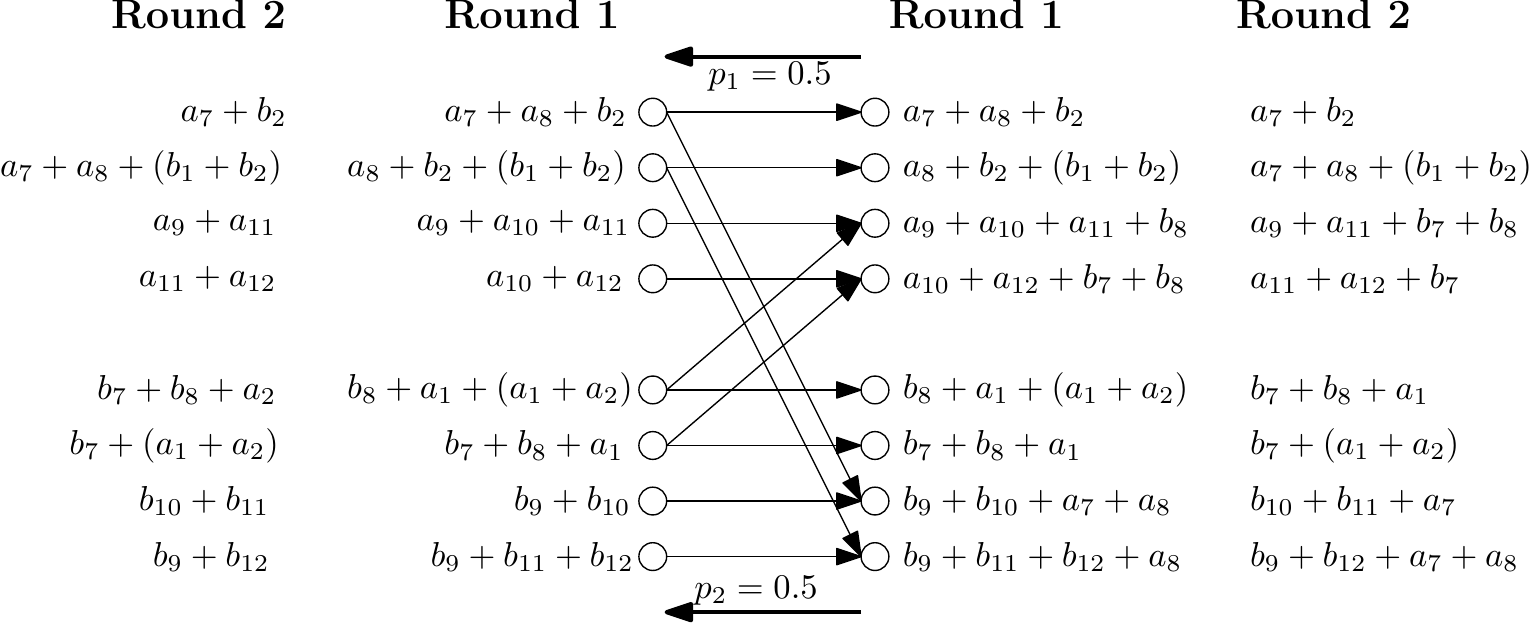}
\centering
\caption{Second block of transmissions for the example coding scheme over linear deterministic channel. The helping information sent by the interfering transmitters ($a_1,a_2$ at Rx1, $b_1,b_2$ at Rx2) are omitted for brevity. Note that these are already known at the receivers from previous block, and hence can be cancelled.} 
\label{fig:ldc_example2}
\end{figure*}

\subsubsection{Comparison with Rate-Limited Feedback}
{Given that both intermittent feedback and rate-limited feedback of \cite{VahidSuh_12} provide degrees-of-freedom gains, one might ask how the two models compare. In order to understand the relative merits of the two feedback models, we revisit the symmetric generalized-degrees-of-freedom curves for the two models for symmetric channel parameters, plotted in Figure~\ref{fig:gdof}. The figure illustrates the fact that there is no direct equivalence between the two models, \emph{i.e.}, there is no amount of rate-limited feedback that can exactly replicate the gain of intermittent feedback uniformly for all interference strengths, and vice versa. For the specific feedback parameters given in Figure~\ref{fig:gdof}, we observe that rate-limited feedback is more useful for the weak interference regime described by $\alpha \leq 2/3$, whereas intermittent feedback is more useful for the strong interference regime, given by $\alpha \geq 2$. The reason is that for weak interference, rate-limited model allows for block processing of the channel output to generate the feedback signal, hence feedback helps to resolve the interference in all time slots, whereas for intermittent feedback, information about interference in some time slots gets unrecoverably lost on the erasure channel. On the other hand, for strong interference, rate-limited feedback imposes a hard limit on the amount of gain that can be obtained from feedback, but the gain is still unbounded for intermittent feedback, since the end-to-end mutual information of the alternative path created by feedback gets larger with increasing interference.}

%% file: Motivation.tex
In this section, we illustrate our coding scheme through an example over the linear deterministic channel. This example is intended to demonstrate how and why the proposed scheme works, and motivate the use of quantize-map-forward as a feedback strategy.

We consider the symmetric channel shown in Figures~\ref{fig:ldc_example1} and ~\ref{fig:ldc_example2}, with $n_{11}=n_{22}=4$, $n_{12}=n_{21}=2$, and $p_1=p_2=0.5$, and focus on the achievable symmetric rate. In this example we will take a block length of $N=2$ for illustration purposes. Although for this particular case, the probability of decoding error is large due to short block length, in general the same coding idea can be applied for a large block length, in which case arbitrarily small error probability can be achieved by taking advantage of the law of large numbers.

We focus on two blocks of transmission. At each block, the users split their messages into common and private parts. The common parts of the messages are decoded by both receivers, whereas the private part is only decoded by the intended receiver, as in Han-Kobayashi scheme for the interference channel without feedback \cite{HanKobayashi_81}. In the first block, Tx$1$ sends linear combinations of its two common information symbols, $a_1,a_2$ on its two common (upper) levels, and linear combinations of its private information symbols, $a_3,a_4,a_5,a_6$, over its private (lower) two levels over a block of two time slots. Tx2 performs similar operations for its common symbols $b_1,b_2$, and its private symbols $b_3,b_4,b_5,b_6$.

Note that at this point, the receivers can decode the symbols sent at their upper two levels by solving the four equations in two unknowns.

After each time slot, the receivers feed back their channel outputs, but the transmitters wait until the end of the block to collect sufficient information from feedback. We consider a particular feedback channel realization $\lp S_1^N, S_2^N\rp=\lp (1,0),(0,1)\rp$ for illustration purposes. After the first block, each transmitter gets from feedback two linear combinations of the interfering symbols of the previous block, by subtracting their own linear combinations from the channel outputs. In the second block, the transmitters perform further linear encoding of these two linear combinations. These additional linear combinations of the interference symbols are superimposed on top of the linear combinations of the fresh common information symbols $a_7, a_8$ (and $b_7,b_8$ for Tx2) of the second block. On the private levels, linear combinations of new symbols $a_9,a_{10},a_{11},a_{12}$ at Tx1 and $b_9,b_{10},b_{11},b_{12}$ at Tx2 are sent, as in the first block.

After the second block of transmission, the receivers collect the four linear equations obtained in the lower two levels of the first block and the four linear equations obtained at the upper two uninterfered levels in the second block. It is easy to check that these eight equations are linearly independent, and hence the receivers can solve for the eight unknowns ($a_3,a_4,a_5,a_6,a_7,a_8, b_1,b_2$ for Tx1, and $b_3,b_4,b_5,b_6,b_7,b_8, a_1,a_2$ for Tx2).

Having decoded the private information (and interference) of the first block and the common information of the second block, the receivers next cancel the additional linear combinations of the previously decoded common information received at the lower two levels of the second block due to feedback. This means that Rx1 cancels the $a_1$ and $a_2$ symbols in the lower two levels, and Rx2 cancels the $b_1$ and $b_2$ symbols.

Since the transmitters can also cancel this information from the received feedback (because it is a function of their own symbols), the state of each terminal reduces to that in the end of the first block. Therefore, in each of the following blocks, the operation in the second block can be repeated, each time letting the receivers decode the private information of the previous block and the common information of the new block. 

One caveat is that, the feedback channel realization will not be the same at each block. To address this point, we first note that the only decoding error event is when the channel realization is such that the resulting linear system in any of the receivers is not full rank. For the particular code in the example, it is easy to check that the probability of this event is zero for any feedback channel realization as long as $S_i^N \neq (0,0)$ for $i=1,2$. In general, for any $\epsilon>0$, in order to achieve a symmetric rate $C_{\text{sym}}-\epsilon$, Tx$i$ needs to receive feedback for at least $N(p_i - \epsilon)$ time slots at each block. This condition is ensured by law of large numbers by letting $N \to \infty$, and arbitrarily small error probability can be achieved\footnote{Note that this does not prove the existence of a sequence of codes that allows arbitrarily small error probability for an arbitrary block length. The intention in this section is to give an illustration of the coding scheme; the precise achievability proof will be presented in Section~\ref{sec:achievability}.}. 

To find the symmetric rate achieved by this scheme, we assume the scheme is run for $B$ blocks. At the end, each receiver will have resolved $6B-4$ information bits in $2B$ time slots. Letting $B \to \infty$ gives a symmetric rate of 3 bits/time slot. Note that without feedback, a symmetric rate of at most 2 bits/time slot can be achieved. At the other extreme, it is also easy to verify from the results in \cite{SuhTse_11} that symmetric capacity under perfect feedback is also 3 bits/time slot, which is in agreement with Figure~\ref{fig:gdof} and Corollary~\ref{cor:threshold_ldc}. 

This example also serves to demonstrate why we perform quantize-map-forward instead of decode-and-forward as a feedback strategy. In general, to achieve the symmetric capacity, Tx2 needs to send linear combinations of $N$ information symbols on its common levels, while Tx1 receives $2Np_1$ of these linear combinations on the average. Hence, if $p_1 < 0.5$, Tx1 will not be able to decode the interference of the previous block. Instead, Tx1 performs a linear mapping of the received feedback information, which turns out to achieve the symmetric capacity.   

Finally, we point out that decoding in this scheme is \emph{sequential}, \emph{i.e.}, the receiver decodes the blocks in the same order they are encoded\footnote{An alternate scheme based on backward decoding was presented in \cite{KarakusWang_13}, for the case of linear deterministic channel.}. This is in contrast to earlier feedback coding schemes proposed for interference channel, which perform backward coding. The obvious advantage of using sequential decoding is better delay performance, since the receiver does not need to wait for the end of the entire transmission to start to decode.

%% file: Achievability.tex
In this section, we describe the coding scheme in detail and derive an inner bound $\mcal{R}^i_G \lp p_1,p_2\rp$ on the rate region.

\subsection{Overview of the Achievable Strategy}

The main idea of the coding scheme is the same as the one presented
for the example in Section~\ref{sec:motivation}. However, it substantially generalizes 
the example scheme in order to account for possible channel noise, different 
interference regimes and an arbitrary target rate point in the achievable region.

The scheme consists of transmission over $B$ blocks, each of length
$N$. At the beginning of block $b$, upon reception of feedback,
transmitters first remove their own contribution from the feedback
signal and obtain a function of the interference and noise realization
of block $b-1$. This signal is then quantized and mapped to a random
codeword, which will be called the helping information. Finally, a new
common codeword, which is to be decoded by both receivers, and a private codeword,
to be decoded by only the intended receiver, are superimposed to the helping
information, and transmitted.

The decoding operation depends on the desired rate point (see
Figure \ref{fig:decoding_policy}). To achieve the rate points for which the
common component of the message is large, the receiver simply performs
a variation of Han-Kobayashi decoding \cite{HanKobayashi_81}, \emph{i.e.}, it decodes the
intended information jointly with the common part of the
interference. Note that this does not make use of the helping
information.

To achieve the remaining rate points, the helping information is
used. For weak interference, at block $b$, we assume that the receiver
has already decoded the intended common information of block
$b-1$. After receiving the transmission of block $b$, the receivers
jointly decode the intended private information and the interference
of block $b-1$ jointly with the common information of block $b$, while
using the helping information sent at block $b$ as side
information. For strong interference, the roles of intended common
information and the interfering common information get switched.

Next, we present a detailed description of the coding scheme and proof
of achievability.

\subsection{Codebook Generation}

Fix $p(x_{ie})p(x_{ic})p(x_{ip})$ for $i=1,2$\footnote{Although the scheme loses beamforming gain by generating independent codebooks at the two users, this only results in a constant rate penalty.}, and $p(u_i|\wtild v_j)$
that achieves $\E{d(U_i, \wtild V_j)}\leq D_i$ for
$\ijj$, where $d:\mcal{U}\times\mcal{V} \to \mbb{R}$ is the distortion measure,
where $\mcal{U}$ and $\mcal{V}$ are the alphabets of $U_i$ and $\wtild V_j$, 
respectively. Generate $2^{Nr_i}$
quantization codewords $U_i^N$ i.i.d. $\sim p(u_i) = \sum_{\wtild v_j}
p(u_i|\wtild v_j)p(\wtild v_j)$, for $\ijj$. For
$i=1,2$, generate $2^{Nr_i}$ codewords $X^N_{ie}$ i.i.d. $\sim
p(x_{ie})$. Further generate, for $i=1,2$, $2^{NR_{ic}}$ codewords
$X^N_{ic}$ i.i.d. $\sim p(x_{ic})$ and $2^{NR_{ip}}$ codewords
$X_{ip}^N$ i.i.d. $\sim p(x_{ip})$. For $i=1,2$, define symbol-by-symbol 
mapping functions $x_i: \mcal{X}_{if}\times\mcal{X}_{ip} \to \mcal{X}_i$ and
 $x_{if}: \mcal{X}_{ie}\times\mcal{X}_{ic} \to \mcal{X}_{if}$, where $\mcal{X}_{ie}$,
$\mcal{X}_{ic}$, $\mcal{X}_{ip}$, and $\mcal{X}_{if}$ are the alphabets for the symbols
$X_{ie}$, $X_{ic}$, $X_{ip}$, and $X_{if}$, respectively.

\subsection{Encoding}

Encoding is performed over blocks (indexed by $b$) of length $N$. See
Figure \ref{fig:encoder} for a system diagram. At the beginning of block
$b$, Tx$i$ receives the punctured feedback signal $\wtild Y^N_i (b-1) = S_i^N(b-1) Y_i^N(b-1)$
containing information about the channel output in block $b-1$,
where the multiplication is element-wise. Upon
reception of $\wtild Y^N_i$, Tx$i$ first removes its own
contribution from the feedback signal to obtain $\wtild V_j^N (b-1)=S_i^N (b-1) V_j^N (b-1)$. 
For linear deterministic model, this is done by
\begin{align*}
\wtild V^N_j (b-1) = \wtild Y_i^N (b-1) - S^N_i (b-1) \mb{H}_{ii} X^N_i (b-1),
\end{align*}
whereas for Gaussian model, it can be obtained by
\begin{align*}
\wtild V^N_j (b-1) &= \wtild Y_i^N (b-1) - S^N_i (b-1) h_{ii} X^N_i (b-1)
\end{align*}
for $\ijj$. 

The interference signal $\wtild V_j^N (b-1)$ is then quantized by finding an
index $Q_i (b)$ such that
\begin{align*}
\intypset{\wtild V^N_j (b-1), U^N_{i}(Q_i(b))},
\end{align*}
where $\mcal{T}_\epsilon^{(N)}$ denotes the $\epsilon$-typical set
with respect to the distribution $p(\wtild v_j)p(u_i|\wtild v_j)$, and
$p(\wtild v_j)$ is induced by the channel and the input
distributions. If such an index $Q_i (b)$ has been found, the codeword
$X^N_{ie} (Q_i (b))$ that has the same index is chosen to be sent for
block $b$. If there are multiple such indices, the smallest one is
chosen. If no such index is found, the quantization index 1 is chosen.

\begin{figure}
\centering
\includegraphics[scale=0.65]{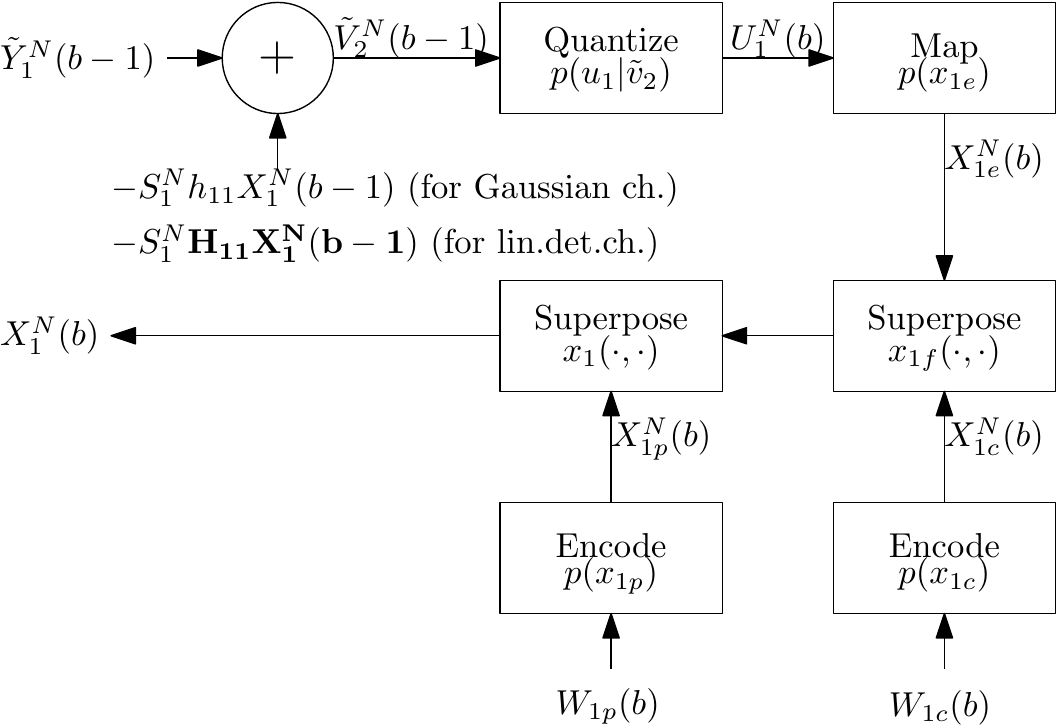}
\centering
\caption{Encoder diagram at Tx1}
\label{fig:encoder}
\end{figure}

Next, the message $W_i (b) \in \lb 2^{NR_i}
\rb$ to be sent at block $b$ is split into common and private
components $(W_{ic} (b), W_{ip} (b)) \in \lb 2^{NR_{ic}} \rb \times
\lb 2^{NR_{ip}} \rb$. Depending on the desired message indices $\lp
W_{ic} (b), W_{ip} (b) \rp$, a common codeword $X^N_{ic} (W_{ic}
(b))$, and a private codeword $X^N_{ip} (W_{ip} (b))$ is chosen from
the respective codebooks.

Finally, the using the symbol-wise maps $x_{if}\lp \cdot,\cdot\rp$ and $x_{i}\lp \cdot,\cdot\rp$, we obtain the codewords
\begin{align*}
X_{if}^N (b) &= x_{if}\lp X^N_{ie} (b), X^N_{ic} (b)\rp \\
X_i^N (b) &= x_i \lp X_{if}^N (b), X_{ip}^N (b) \rp
\end{align*}
where the functions are applied to vectors element-wise. $X_i^N (b)$ is sent at Tx$i$ over $N$ channel uses.

\subsection{Decoding}
The message indices for common and private messages, and the
quantization indices of Tx$i$ at block $b$ will be denoted by $m_i
(b)$, $n_i (b)$, and $q_i (b)$, respectively. When there are two
quantization indices to be decoded from the same user, the second one will be denoted
with $q_i ' (b)$. 

In order to describe the decoding process, we need to introduce some
notation. Define the following sequence of sets:
\begin{align*}
&\mcal{B}_i^{(N)} ( (q_j, m_j )(b-1)) := \Big \{ q_i(b) : \big( \ul{S}^N (b-1),\\
& \quad X_{jf}^{N} ((q_j, m_j )(b-1)), (U_i^N, X_{ie}^N) (q_i (b)) \big) \in \mcal{T}_\epsilon^{(N)} \Big \}.
\end{align*}
for $(i,j)=(1,2), (2,1)$. Loosely, $\mcal{B}_i^{(N)}$ is the set of
quantization indices of Tx$i$ that are jointly typical with the interference
of the previous round. If any of the indices $(q_j, m_j )$ is known,
we will suppress the dependence to that index, \emph{e.g.}, if both
are known, we simply denote
\begin{align*}
  \mcal{B}_i^{(N)} (b) &:= \Big \{ q_i (b): \big( \underline{S}^N (b-1),  X_{jf}^{N} (b-1),\\
  &\qquad\qquad\qquad\quad (U_i^N, X_{ie}^N) (q_i (b)) \big)
  \in \mcal{T}_\epsilon^{(N)} \Big \}
\end{align*}
where $X_{jf}^{N} (b-1)$ refers to the codeword corresponding to the
known message indices.

We assume that the set $\mcal{B}_i^{(N)} (b)$ has cardinality $2^{NK_i (b)}$. Specifically,
\begin{align*}
  K_i (b) = \frac{\log \left| \lbp q_i (b) : \intypset{\wtild
      V^N_j (b-1), U^N_{i}(q_i(b))} \rbp \right|}{N} 
\end{align*}
Note that due to random codebook generation, $K_i (b), i=1,2$, are
random variables.
The following lemma shows that $K_i (b)$ is almost surely bounded for sufficiently large $N$.

\begin{lemma} \label{lem:listsize} For any $\epsilon>0$, there exists
  a block length $N$, and a quantization scheme such that $K_i (b) <
  \kappa_i + \delta(\epsilon)$, where
\begin{align*}
\kappa_i := I(\wtild V_j ; U_i|S_i) - I(X_{jf}; U_i|S_i)
\end{align*}
for $\ijj$, and $\delta(\epsilon)$ is such that
$\delta(\epsilon) \to 0$ as $\epsilon \to 0$.
\end{lemma}
\begin{proof}
See Appendix~\ref{sec:ap_listsize}.
\end{proof}

Lemma~\ref{lem:listsize} suggests that for each interference codeword, there is a constant number of plausible quantization codewords, for sufficiently large block length (to see that $\kappa_i$ is a constant independent of channel parameters, refer to Appendix~\ref{sec:ap_evaluation}). This means that the cost of jointly decoding the quantization indices together with the actual messages is a constant reduction in the achievable rate, which will be a useful observation in deriving the constant-gap result.

We also define $C_i = \kappa_i + 2\kappa_j$, for $\ijj$. The
reason for this particular definition will become clear in the error
analysis. Intuitively, $C_i$ represents the rate cost associated with
performing quantization to forward the feedback information, which introduces
distortion. However, as we will show later in the proof, the upper bound given in
Lemma~\ref{lem:listsize} can be evaluated as a constant independent of
channel parameters.

Given an input distribution, Rx1 is said to be in weak interference if $I(X_2;Y_1|X_1) \leq I(X_1;Y_1|X_2)$, and in strong interference otherwise. These regimes are defined similarly for Rx2.

Decoding operation depends on the interference regime and the desired
operating point $(R_1,
R_2)$. In order to describe the relevant regimes of operating points, we define
\begin{align}
I_{wi} &:= I( X_{if}; Y_i|X_{1e}, X_{2e}) - C_i, \label{eq:I_def_w} \\
I_{si} &:= I( X_{jf}; Y_i|X_{1e}, X_{2e}) - C_i, \label{eq:I_def_s}
\end{align}
for $\ijj$. In what follows, for clarity, we will focus only on Rx1. The
operations performed at Rx2 are similar. 

\begin{figure}
\centering
\includegraphics[scale=0.62]{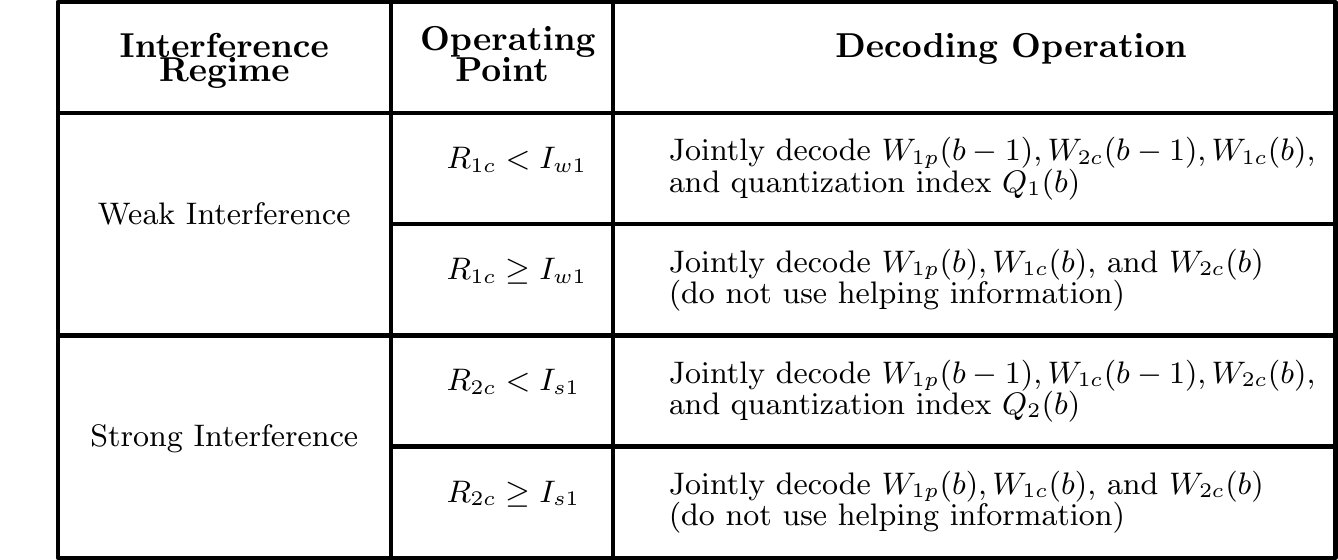}
\centering
\caption{A high-level summary of the decoding policy at Rx1 (Details are omitted).}
\label{fig:decoding_policy}
\end{figure}

\subsubsection{Weak Interference $\lp I(X_2;Y_1|X_1) \leq I(X_1;Y_1|X_2)\rp$}

If, for the desired operating point, $R_{1c} > I_{w1}$, where $I_{w1}$ is as defined in
\eqref{eq:I_def_w}, the helping information is
not used, and a slight modification of Han-Kobayashi scheme is
employed. Otherwise, the helping information is used to decode the
information of block $b-1$. We describe the decoding for the two cases
below.

$\mathbf{R_{1c} \geq I_{w1}}:$ At block $b$, we assume that $X^N_{1e}
(b)$ and $X^N_{1} (b-1)$ are known. The decoder attempts to find
unique indices $\lp m_1 (b), n_1 (b), m_2 (b) \rp \in \lb 2^{NR_{1c}}
\rb \times \lb 2^{NR_{1p}} \rb \times \lb 2^{NR_{2c}} \rb$, and some
$q_2 (b) \in \lb 2^{Nr_2} \rb$ such that
\begin{align}
\lp
\begin{array}{l}
\ul{S}^N (b-1), X^N_{1f} (b-1), X^N_{1e} (b), X^N_{2e} (q_2 (b)), \\
X_{1f}^N(m_1 (b)), X_{1}^N(m_1 (b), n_1 (b)), \\
X^N_{2f} (q_2 (b), m_2 (b)), Y_1^N (b)
\end{array}
\rp \in \mcal{T}_{\epsilon}^{(N)}
\label{eq:weak_hk_dec}
\end{align}
where the known message indices are suppressed. If the receiver can find a
unique collection of such indices, it declares them as the decoded
message indices $\lp \what W_{1c} (b), \what W_{1p}(b), \what W_{2c} (b)\rp$;
otherwise it declares an error.

After decoding, given the knowledge of $X_{1}^N (b)$, Rx1 reconstructs
$X^N_{1e} (b+1)$ by imitating the steps taken by Tx1 at the beginning
of block $b+1$, thereby maintaining the assumption that $X^N_{1e} (b)$
is known at the beginning of block $b$. Further, note that $X^N_{2e} (b)$
is not uniquely decoded, hence in block $b+1$, it will still be jointly (but still, non-uniquely)
decoded with the variables of that block. We resort to non-unique decoding
of this codeword since unique decoding imposes an additional
rate constraint on the helping information, thereby limiting the amount of 
rate enhancement it can provide.

$\mathbf{R_{1c} < I_{w1}}:$ At block $b$, it is assumed that $X_{1f}^N (b-1)$ and $X_1^N (b-2)$ are known at Rx1.

To decode, Rx1 attempts to find unique indices $\lp m_1 (b), n_1
(b-1), m_2 (b-1) \rp \in \lb 2^{NR_{1c}} \rb \times \lb 2^{NR_{1p}}
\rb \times \lb 2^{NR_{2c}} \rb$ and some triple $( q_2 (b-1), q_2 (b),
q_1 (b) ) \in \lb 2^{Nr_2} \rb \times \lb 2^{Nr_2} \rb \times \lb
2^{Nr_1} \rb$ such that
\begin{align}
\lp
\begin{array}{l}  
\underline{S}^N (b-1), X_{1f}^N (b-2), X_{1f}^N (b-1), \\
X_1^N (n_1 (b-1)), X_{2e}^N (q_2 (b-1)),Y_1^N (b-1),  \\
\lp U_{1}^N, X_{1e}^N\rp (q_1 (b)), X_{2e}^N (q_2 (b)), \\
X_{2c}^N (m_2 (b-1)), X_{1c}^N (m_1(b)), Y_1^N (b)
\end{array}
\rp \in \mcal{T}_{\epsilon}^{(N)}
\label{eq:weak_fb_dec}
\end{align}

If a unique collection of such indices exists, then these are declared
as the decoded message indices $\lp \what W_{1c} (b), \what W_{1p} (b-1),\right.$ $\left.\what
W_{2c} (b-1)\rp$. Otherwise, an error is declared.

In \eqref{eq:weak_fb_dec}, the dependence of $ X_1^N (b-1)$ to the indices
$q_1(b-1)$ and $m_1(b-1)$ is suppressed, since these indices
correspond to messages that have already been decoded.

In words, the decoder jointly decodes the private information and the
interference of block $b-1$ jointly with the helping information and
common information from block $b$.

Note that non-unique decoding is performed for $X_{1e}^N (b)$, but we
have assumed that $X_{1f}^N (b-1)$ $\lp \text{and thus, } X_{1e}^N (b-1) \rp$ is
uniquely known at the beginning of block $b$. In order to maintain
this assumption for the next block, $X_{1e}^N (b)$ is reconstructed at
Rx1. To achieve this, given the knowledge of $X_{1}^N (b-1)$, and the
quantization codebook, Rx1 imitates the operations performed by Tx1 at
the beginning of block $b$.

\subsubsection{Strong Interference $\lp I(X_2;Y_1|X_1) > I(X_1;Y_1|X_2)\rp$}

As in the weak interference case, decoding depends on the operating
point. For $R_{2c} < I_{s1}$, where $I_{s1}$ is as defined in \eqref{eq:I_def_s},
helping information is used, otherwise, it is not
used.

$\mathbf{R_{2c} \geq I_{s1}}:$ The operations performed are identical
to those for the case of $R_{1c} \geq I_{w1}$ under weak interference.

$\mathbf{R_{2c} < I_{s1}}:$ We assume $X_1^N (b-2)$, $X_{1e}^N (b-1)$,
and $X_{2c}^N (b-1)$ are known at Rx1 at block $b$.

To decode, Rx1 attempts to find unique indices $\lp m_1 (b-1), n_1
(b-1), m_2 (b) \rp \in \lb 2^{NR_{1c}} \rb \times \lb 2^{NR_{1p}} \rb
\times \lb 2^{NR_{2c}} \rb$ and some $(q_2 (b-1), q_2 (b), q_1 (b) )
\in \lb 2^{Nr_2} \rb \times \lb 2^{Nr_2} \rb \times \lb 2^{Nr_1} \rb$
such that
\begin{align}
\lp
\begin{array}{l}
\ul{S}^N (b-1), X_{1f}^N (b-2), X_{1e}^N (b-1),   \\
X_{2e}^N (q_2 (b-1)),
X_{1c}^N (m_1 (b-1)), \\
 X_{1p}^N ( n_1 (b-1) ), X_{2c}^N (m_2 (b)), X_{1e}^N (q_1 (b)), \\
 (U_{2}^N, X_{2e}^N) (q_2 (b)), 
Y_1^N (b-1), Y_1^N (b) 
\end{array}
\rp \in \mcal{T}_{\epsilon}^{(N)}
 \label{eq:strong_fb_dec}
\end{align} 
If a unique collection of such indices exists, they are declared as
the decoded message indices $\lp \what W_{1c} (b-1), \what W_{1p} (b-1), \what
W_{2c} (b) \rp$. Otherwise, an error is declared. Using the information of
$X^N_{1} (b-1)$, Rx1 can now uniquely reconstruct $X_{1e}^N (b-1)$ by
following the steps taken by Tx1 at the beginning of block $b$.

\subsection{Error Analysis}

Without loss of generality, we only consider the error events
occurring at Tx1 and Rx1. All arguments here will be applicable to the
other Tx-Rx pair. We define the following decoding error events at
Rx1, for block $b$ and block length $N$:
\begin{align*}
 & D_{FB,w} (b,N) = \lbp \what W_{1c} (b) = W_{1c} (b), \right.\\
  & \quad\left.\what W_{1p} (b-1) = W_{1p} (b-1), \what W_{2c} (b-1) = W_{2c} (b-1) \rbp ^c \\
 & D_{FB,s} (b, N) = \lbp \what W_{1c} (b-1) = W_{1c} (b-1),\right.\\
&\quad\left.  \what W_{1p} (b-1) = W_{1p} (b-1), \what W_{2c} (b) = W_{2c} (b) \rbp ^c \\
&  D_{NFB} (b, N) = \lbp \what W_{1} (b) = W_{1} (b), \what W_{2c} (b) =
  W_{2c} (b) \rbp ^c
\end{align*}
The overall decoding error events at Rx1 is given by
\begin{align*}
D_{FB,w}(N)= &\bigcup_{b=1}^B D_{FB,w} (b),\; D_{FB,s}(N) = \bigcup_{b=1}^B D_{FB,w} (b) \\
&D_{NFB}(N) = \bigcup_{b=1}^B D_{NFB} (b)
\end{align*}
We first prove that in order to find the rate achieved after transmission of $B$ blocks, it is sufficient to focus on the error events at an arbitrary block $b$. Without loss of generality, consider the error event $D_{FB,w}(N)$. Assume that, after $B$ blocks of transmission, the effective rate achieved by Tx$i$ is $\bar{R}_i$ (Note that at the end of block $B$, some of the information pertaining to block $B$ is still undecoded), which can be lower bounded by $\bar R_i \geq \frac{B-2}{B} R_i$, by ignoring the partial information decoded in the first block and the last one. We can also upper bound the overall probability of error by
\begin{align*}
\Prob{D_{FB,w}} &\leq \sum_{b=2}^B \Prob{D_{FB,w} (b,N)|\lbp D^c_{FB,w} (b',N) \rbp_{b'=2}^{b-1}} \\
&\leq B\Prob{D_{FB,w} (b,N)|\lbp D^c_{FB,w} (b',N) \rbp_{b'=2}^{b-1}} \\
&=: B\Prob{\mcal{D}_{FB,w} (b,N)}
\end{align*}
for an arbitrary block $b$, where the second line follows by the fact that the encoding and decoding processes are identical in each block, and we made a definition in the last line for brevity\footnote{The event $\mcal{D}_{FB,w}$ is defined in the filtered probability space formed by the conditioning.}. Setting $B=N=N'$, we see that for any $N'$, an error probability less than $N' \Prob{\mcal{D}_{FB,w} (b,N')}$ can be achieved with rate $\frac{N'-1}{N'}R_i$. Therefore, in order to show that rate $R_i$ is achievable, it is sufficient to show that $N \Prob{\mcal{D}_{FB,w} (b,N)} \to 0$ as $N \to \infty$. Using the same arguments, one can show the same result for $D_{FB,s}(N)$ and $D_{NFB}(N)$, and define $\mcal{D}_{FB,s}(b,N)$ and $\mcal{D}_{NFB} (b,N)$ similarly.

Now we analyze the weak and strong interference regimes separately.

\subsubsection{Weak Interference}

The following lemmas characterize the rate constraints for reliable
communication with Rx1 for feedback and non-feedback strategies,
respectively, under weak interference.

\begin{lemma} \label{lem:weak_fb}
For weak interference at Rx1, $N\Prob{\mcal{D}_{FB,w} (b,N)} \to 0$ as $N \to \infty$ if
\begin{align}
  R_{1c} &< I(X_{1f};Y_1|\ul{S}, X_{1e},X_{2e}) - C_1 \label{eq:weak_fb_first} \\
  R_{1p} &< I(X_1;Y_1|\ul{S}, X_{1f}, X_{2f}) - C_1 \label{eq:weak_fb_second} \\
  R_{2c} &< I(X_{2f};Y_1 | \ul{S}, X_{2e}, X_1) - C_1 \\
  R_{1p} + R_{2c} &< \min \Big \{ I(X_1, X_{2f}; Y_1, U_1 |\ul{S},  X_{1f}, X_{2e})-2C_1, \notag\\
  & \qquad I(X_1, X_{2f}; Y_1 | \ul{S}, X_{1c}, X_{2e})-C_1 \Big \} \label{eq:weak_fb_nontrivial}\\
  R_{1} + R_{2c} &< I(X_1, X_{2f} ;Y_1 | \ul{S}, X_{1e}, X_{2e}) -C_1
  \label{eq:weak_fb_last}
\end{align}
\end{lemma}

\begin{proof}
See Appendix~\ref{sec:ap_achievability}.
\end{proof}


\begin{lemma} \label{lem:weak_hk}
For weak interference at Rx1, $N\Prob{\mcal{D}_{NFB} (b,N)} \to 0$ as $N \to \infty$ if
\begin{align}
  R_{1c} &> I(X_{1f};Y_1 | \ul{S}, X_{1e}, X_{2e}) - C_1 \label{eq:weak_hk_first} \\
  R_{1p} &< I( X_1; Y_1| \ul{S}, X_{1f}, X_{2f}) - \kappa_2 \label{eq:weak_hk_second} \\
  R_{2c} &< I( X_{2f} ; Y_1 | \ul{S},  X_ {2e},  X_1) - \kappa_2 \\
  R_{1} &< I( X_1;  Y_1| \ul{S}, X_{2f},  X_{1e}) - \kappa_2 \\
  R_{1}+R_{2c} &< I( X_1, X_{2f}; Y_1| \ul{S}, X_{1e}, X_{2e}) - C_1-\kappa_2
  \label{eq:weak_hk_last}
\end{align}
\end{lemma}
\begin{proof}
See Appendix~\ref{sec:ap_achievability}.
\end{proof}


\subsubsection{Strong Interference}

The following lemmas give the rate constraints for the feedback and
non-feedback modes under strong interference at Rx$i$.

\begin{lemma} \label{lem:strong_fb}
For strong interference at Rx1, $N\Prob{\mcal{D}_{FB,s} (b,N)} \to 0$ as $N \to \infty$ if
\begin{align}
R_{2c} &< I(X_{2f};Y_1| \ul{S}, X_{1e}, X_{2e}) - C_1  \label{eq:strong_fb_first} \\
R_{1p} &< I(X_1;Y_1|\ul{S}, X_{1f},X_{2f}) - C_1 \label{eq:strong_fb_second} \\
R_{1} &< \min \lbp I(X_{1};Y_1, U_2|\ul{S}, X_{1e},X_{2f}), \right.\\
&\qquad \left. I(X_{1}, X_{2e};Y_1|\ul{S}, X_{1e},X_{2c} ) \rbp - C_1  \label{eq:strong_fb_nontrivial}\\
R_{1}+R_{2c} &< I(X_1, X_{2f}; Y_1|\ul{S}, X_{1e},  X_{2e}) - C_1 \label{eq:strong_fb_last}
\end{align}
\end{lemma}

\begin{proof}
See Appendix~\ref{sec:ap_achievability}.
\end{proof}

\begin{lemma} \label{lem:strong_hk}
For strong interference at Rx1, $N\Prob{\mcal{D}_{NFB} (b,N)} \to 0$ as $N \to \infty$ if
\begin{align}
R_{2c} &> I(X_{2f};Y_1|\ul{S}, X_{1e}, X_{2e}) - C_1 \label{eq:strong_hk_first} \\
R_{1p} &< I(X_1;Y_1|\ul{S}, X_{1f},X_{2f}) - \kappa_2 \label{eq:strong_hk_second} \\
R_{1p}+R_{2c} &< I( X_1, X_{2f};  Y_1|\ul{S},  X_{1f}, X_{2e}) - \kappa_2 \\
R_{1}+R_{2c} &< I(X_1, X_{2f}; Y_1|\ul{S}, X_{1e}, X_{2e}) - C_1 - \kappa_2 \label{eq:strong_hk_last}
\end{align}
\end{lemma}
\begin{proof}
See Appendix~\ref{sec:ap_achievability}.
\end{proof}

\subsection{Rate Region Evaluation} 

In this subsection, we first explicitly derive the set of achievable $\lp R_1,R_2\rp$ pairs for linear deterministic and Gaussian models, from the results of the previous subsection.

We first find the conditions for decodability at Rx1 under weak interference. Recall that feedback mode is used at Rx$1$ only if $\eqref{eq:weak_fb_first}$
is satisfied; otherwise Han-Kobayashi decoding is performed. If we
define $\ul{R} := \lp R_{1c}, R_{2c}, R_{1p} \rp$, and
\begin{align*}
\mcal{R}_{FB}^w &:= \lbp \underline{R}: \text {\eqref{eq:weak_fb_second}-\eqref{eq:weak_fb_last} is satisfied}  \rbp, \\
\mcal{R}_{NFB}^w &:= \lbp \underline{R}: \text { \eqref{eq:weak_hk_second}-\eqref{eq:weak_hk_last} is satisfied}  \rbp, \\
\mcal{R}_{d}^w &:= \lbp \underline{R}: \text { \eqref{eq:weak_fb_first} is satisfied} \rbp,
\end{align*}
then the set of rate points $\mcal{R}^w$ that ensure decodability at
Rx1 under weak interference contains
\begin{align*}
\mcal{R}^w &= \lp \mcal{R}_{FB}^w \cap \mcal{R}_{d}^w \rp \cup \lp \mcal{R}_{NFB}^w \cap \mcal{R}^{w,c}_{d}   \rp \\
&\supseteq \lp \mcal{R}_{NFB}^w \cap \mcal{R}_{FB}^w \cap \mcal{R}_{d}^w \rp \cup \lp \mcal{R}_{NFB}^w \cap \mcal{R}_{FB}^w \cap \mcal{R}^{w,c}_{d}   \rp \\
&= \mcal{R}_{NFB}^w\cap \mcal{R}_{FB}^w
\end{align*}
where $\mcal{R}^{w,c}_{d}$ is the complement of the set
$\mcal{R}^{w}_{d}$. Therefore, the rate constraints for decodability
at Rx1 for the described strategy for weak interference are given by
\eqref{eq:weak_fb_second}-\eqref{eq:weak_fb_last} and \eqref{eq:weak_hk_second}-\eqref{eq:weak_hk_last},
for all joint distributions $\prod_{i=1}^2
p(x_{ie})p(x_{ic})p(x_{ip})$, symbol-wise mappings 
$x_{if}(x_{ie}, x_{ic})$, $x_{i}(x_{if}, x_{ip})$,
and $p(u_i|\tilde v_j)$,
$\ijj$, consistent with the distortion constraints.

One can perform the same line of arguments as in the case of weak
interference to show that the rate constraints for decodability at Rx1
for strong interference are given by \eqref{eq:strong_fb_second}-\eqref{eq:strong_fb_last}
and \eqref{eq:strong_hk_second}-\eqref{eq:strong_hk_last}, for all joint distributions
$\prod_{i=1}^2 p(x_{ie})p(x_{ic})p(x_{ip})$, symbol-wise mappings 
$x_{if}(x_{ie}, x_{ic})$, $x_{i}(x_{if}, x_{ip})$, and $p(u_i|\tilde v_j)$,
$\ijj$, consistent with the distortion constraints.

Next, we consider linear deterministic and Gaussian models separately, and derive the achievable rate regions explicitly for both cases.

\subsubsection{Rate Region for Linear Deterministic Model}
To obtain the achievable rate region, we first evaluate the mutual information terms with specific input distributions. In particular, we choose the distributions and mappings
\begin{align}
X_{ie} &\sim Unif \lb \mbb{F}_2^{n_{ji}} \rb \label{eq:ldc_inputdist_first}\\
X_{ic} &\sim Unif \lb \mbb{F}_2^{n_{ji}} \rb \\
X_{ip} &\sim Unif \lb \mbb{F}_2^{(n_{ii}-n_{ji})^+} \rb \\
U_i &= \wtild V_j \\
x_{if} &: \mbb{F}_2^{n_{ji}} \times \mbb{F}_2^{n_{ji}} \to \mbb{F}_2^{n_{ji}},\notag\\
x_i &: \mbb{F}_2^{n_{ji}} \times \mbb{F}_2^{(n_{ii}-n_{ji})^+} \to \mbb{F}_2^{\max \lp n_{ii}, n_{ji}\rp}, \notag\\
x_i &= \lb X_{if} \;\; X_{ip} \rb^T, \;\;x_{if}(a,b) = a+b  \label{eq:ldc_inputdist_last} 
\end{align}
for $\ijj$, where $Unif \lb \mcal{A}\rb$ denotes uniform distribution over the set $\mcal{A}$. Evaluating the mutual information terms of the previous subsection with this set of distributions, and applying Fourier-Motzkin elimination (see Appendix~\ref{sec:ap_evaluation} for details), we obtain the rate region given in \eqref{eq:ldc_R1}--\eqref{eq:ldc_R12R2}.

\subsubsection{Rate Region for Gaussian Model}
Now we evaluate the rate constraints obtained in the previous section,
and obtain the final achievable rate region. Assuming available power
$P_i$ at Tx$i$, we assign the following input distributions, for
$\ijj$:
\begin{align}
X_{ie} &\sim \cgauss{0}{\frac{1}{2} P_i} \label{eq:g_inputdist_first}\\
X_{ic} &\sim \cgauss{0}{\frac{1}{2}(1-P_{ip}) P_i} \\
X_{ip} &\sim \cgauss{0}{\frac{1}{2}\min \lp \frac{1}{|h_{ji}|^2 P_i}, 1 \rp P_i} \\
U_i | \wtild V_j &\sim \cgauss{\wtild V_j}{D_i} \\
x_{if} &: \mbb{C} \times\mbb{C} \to \mbb{C}, \;\; x_{i} : \mbb{C} \times\mbb{C} \to \mbb{C}, \notag\\
x_{if}&(a,b) = a+b , \;\; x_{i}(a,b) = a+b \label{eq:g_inputdist_last}
\end{align}
where $D_i>0$ are the distortion parameters. Using these input distributions, and applying Fourier-Motzkin elimination (See Appendix~\ref{sec:ap_evaluation} for details), we can show that the rate region \eqref{eq:ib_Ri}--\eqref{eq:ib_2RiRj}, given in Appendix~\ref{sec:ap_evaluation}, is achievable.

%% file: Converse.tex
We now prove an outer bound region that exactly matches the region given in \eqref{eq:ldc_R1}--\eqref{eq:ldc_R12R2}, and is within a constant gap of the region in \eqref{eq:g_Ri}--\eqref{eq:g_2RiRj}.

The main idea between the novel bounds on $R_1$ and $R_2$ is based on a genie argument, where the receivers are provided with side-information about the messages. The bounds on $R_1+R_2$, $2R_1+R_2$ and $R_1+2R_2$ are proven through a channel enhancement technique, resembling the one used for the multiple-access channel in \cite{KhistiLapidoth_13}.

\subsection{Bounds on $R_1$ and $R_2$}
Since any outer bound for perfect feedback is also an outer bound for intermittent feedback, we have the perfect feedback bound
\begin{align}
R_i \leq \max \lp n_{ii}, n_{ij}\rp
\end{align}
for linear deterministic model, and the bound
\begin{align}
R_i &\leq \sup_{0\leq\rho\leq1} \log \lp 1 + \SNR_i + \INR_i + 2\rho\sqrt{\SNR_i \cdot \INR_i}\rp \label{eq:ob_p_Ri}
\end{align}
for Gaussian model, for $\ijj$, which are both proved in \cite{SuhTse_11}. Next, we prove a novel bound for both models.

Without loss of generality, we focus on the bound on $R_1$. In order to prove the novel bound on $R_1$, the main idea is to provide $\lp W_2, \wtild V_1^N \rp$ as side-information to Rx1. The intuition behind this particular choice is revealed when we consider the interference regime and operating point in which this bound is active. First, due to the structure of the capacity region, this bound is relevant only when the message (\emph{i.e.}, the rate) of the interfering user is small enough. Hence, for that regime, $W_2$ does not carry too much information, and thus providing this to Rx1 still results in a tight outer bound. Second, note that this bound is only active in the strong interference regime, where feedback from Rx2 to Tx2 creates an alternative path for the transmission of $W_1$. Therefore, by forwarding this information, Tx2 indeed provides the information contained in $\wtild V_1^N$ to Rx1.

Based on this idea, we prove the bound
\begin{align}
R_i \leq n_{ii} + p_j\lp n_{ji}-n_{ii}\rp^+  \label{eq:ob_ldc_Ri}
\end{align}
for linear deterministic model in Appendix~\ref{sec:ap_ob_ldc}, and the bound
\begin{align}
R_i \leq \log \lp 1 + \SNR_i \rp + p_j \log \lp1 + \frac{\INR_j}{1+\SNR_i} \rp  \label{eq:ob_g_Ri}
\end{align}
for the Gaussian model in Appendix~\ref{sec:ap_ob_g}, for $\ijj$.

\subsection{Bounds on $R_1+R_2$, $2R_1+R_2$ and $R_1+2R_2$}
We have the perfect feedback outer bounds
\begin{align}
R_i+R_j \leq \max \lp n_{ii}, n_{ij}\rp + \lp n_{jj}-n_{ji}\rp^+
\end{align}
for linear deterministic model, and 
\begin{align}
&R_i+R_j < \sup_{0\leq\rho\leq1} \log \lp 1 + \frac{(1-\rho^2)\SNR_i}{1+(1-\rho^2)\INR_j} \rp  \notag\\
&\qquad+ \log\lp 1 + \SNR_j + \INR_j + 2\rho\sqrt{\SNR_j\cdot\INR_j} \rp \label{eq:ob_p_RiRj}
\end{align}
for Gaussian model, for $\ijj$.

Next, we prove novel outer bounds on the capacity region. {The novelty in these bounds is in the fact that it combines the existing genie-aided bounding techniques for interference channel with the channel enhancement technique of \cite{KhistiLapidoth_13}.} In order to prove these bounds, we first define a notion of enhanced channel. Considering our achievable scheme, feedback can be interpreted as a mechanism for the receivers to separate the interference and the intended signal, to the extent allowed by the erasure probability in the feedback channel. In the weak interference regime, this allows the receiver to cancel the interference. In the strong interference regime, through the alternate path created by the interfering user, it allows the reception of additional information about the intended message. Therefore we consider an enhanced channel where the receivers observe the interference and the intended signal individually whenever the feedback is available, and their sum otherwise. In addition to this enhancement, we provide Rx$i$ with the side-information of $V_i^N$ as well, as was done in \cite{EtkinTse_07}. To make this more precise, we consider the two models separately.

\subsubsection{Linear Deterministic Model}
We define the enhanced linear deterministic channel with intermittent feedback by the following equations
\begin{align*}
\breve Y_i = \left\{ 
\begin{array}{ll}
Y_i, & \text{ if $S_i = 0$}\\
\lp \mb{H}_{ii}X_i, V_j\rp, & \text{ if $S_i = 1$}
\end{array}
\right.
\end{align*}
for $\ijj$, where $\breve Y_i$ is the channel output of the enhanced channel at Rx$i$, $Y_i$ is the channel output of the original channel, and $X_i$ and $V_j$ are as defined for the original channel. The output of the feedback channel is given by $\wtild Y_i = S_i Y_i$, \emph{i.e.}, the same as the original channel. Note that any scheme that achieves arbitrarily small error probability in the original channel can also achieve arbitrarily small error probability for the enhanced channel, using the fact that $Y_i =  \mb{H}_{ii}X_i + V_j$. This means that the capacity region of the original channel is a subset of that of the enhanced channel, and we can derive an outer bound for the enhanced channel instead. 

\begin{figure}
\centering
\includegraphics{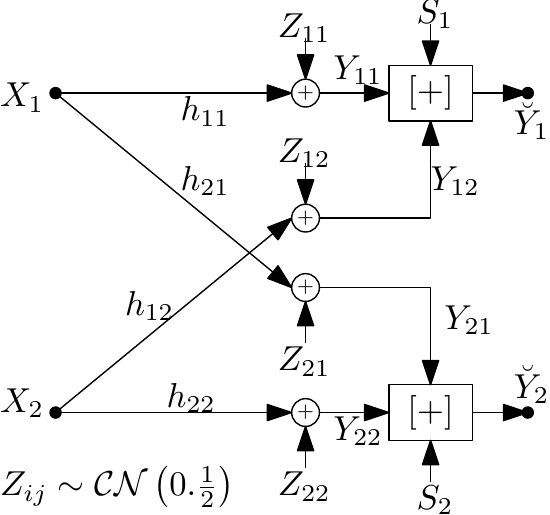}
\caption{The enhanced channel for Gaussian model. The block $\lb+\rb$ is a conditional adder, which outputs the sum of the other two inputs if $S_i=0$, and outputs the two inputs separately otherwise.}
\label{fig:enhanced}
\end{figure}

It is easy to see that this enhancement is equivalent to providing the Rx$i$ with $\wtild V_2^N$, since for time slots where $S_i=1$, Rx$i$ can use this information to individually obtain the interference and the intended symbol.

Using the channel enhancement technique, we arrive at the following outer bounds on the capacity region of the linear deterministic interference channel with intermittent feedback, which are explicitly proved in Appendix~\ref{sec:ap_ob_ldc}.
\begin{align}
&R_1+R_2 \leq \max\lbp n_{12}, \lp n_{11}-n_{21}\rp^+\rbp \notag\\
&+ \max\lbp n_{21}, \lp n_{22}-n_{12}\rp^+\rbp + p_1\min\lbp n_{12}, \lp n_{11}-n_{21}\rp^+\rbp  \notag\\
& + p_2\min\lbp n_{21}, \lp n_{22}-n_{12}\rp^+\rbp\label{eq:ob_ldc_R1R2}\\
&2R_1+R_2 \leq \max\lp n_{11}, n_{12}\rp + \max\lbp n_{21}, \lp n_{22}-n_{12}\rp^+\rbp \notag\\
&+\lp n_{11}-n_{21}\rp^+ + p_2\min\lbp n_{21}, \lp n_{22}-n_{12}\rp^+\rbp \label{eq:ob_ldc_2R1R2}\\
&R_1+2R_2 \leq  \max\lp n_{22}, n_{21}\rp + \max\lbp n_{12}, \lp n_{11}-n_{21}\rp^+\rbp \notag\\
&+\lp n_{22}-n_{12}\rp^+ + p_1\min\lbp n_{12}, \lp n_{11}-n_{21}\rp^+\rbp \label{eq:ob_ldc_R12R2}
\end{align}

\subsubsection{Gaussian Model}
Next, we extend the enhanced channel idea to the Gaussian model. In this case, while splitting the interference and the intended signal, we also split the noise evenly between these two variables (see Figure~\ref{fig:enhanced}). Specifically, we consider the channel defined by the equations
\begin{align*}
\breve Y_i = \left\{ 
\begin{array}{ll}
\bar Y_i, & \text{ if $S_i = 0$}\\
\lp Y_{ii}, Y_{ij}\rp, & \text{ if $S_i = 1$}
\end{array}
\right.
\end{align*}
for $\ijj$, where $\breve Y_i$ is the output of the enhanced channel, and
\begin{align*}
Y_{ii} &= h_{ii} X_i + Z_{ii} \\
Y_{ij} &= h_{ij} X_j + Z_{ij} \\
\bar Y_i &= Y_{ii} + Y_{ij} = h_{ii} X_i +h_{ij} X_j + \bar Z_i
\end{align*}
with $Z_{ij}, Z_{ii}$ are independent and distributed with $\cgauss{0}{\frac{1}{2}}$, and we define $\bar Z_i = Z_{ii}+Z_{ij}$. The output of the feedback channel at Tx$i$ is given by $S_i \bar Y_i = S_i \cdot \lp Y_{ii} + Y_{ij}\rp$, \emph{i.e.}, the same as the original channel. It is worth noting that unlike the linear deterministic case, this enhancement is not equivalent to providing Rx$i$ with $\wtild V_j^N$, since giving this side-information allows the receiver to completely cancel the noise for some time slots, resulting in an infinitely loose bound.

Let $\mcal{C}_e(p_1,p_2)$ denote the capacity region of the enhanced channel.

The next lemma shows that the capacity region of the enhanced channel indeed dominates the original one.
\begin{lemma} \label{lem:enhanced}
For all $0\leq p_1,p_2 \leq 1$,
\begin{align*}
\mcal{C}_G(p_1,p_2) \subseteq \mcal{C}_e(p_1,p_2)
\end{align*}
\end{lemma}
\begin{proof}
The proof has two steps. First, we consider an intermediate channel, with capacity region $\mcal{C}_i(p_1,p_2)$, and the channel output at Rx$i$ is given by
\begin{align*}
Y_i = h_{ii}X_i + h_{ij} X_j + \bar Z_i
\end{align*}
for $\ijj$, where $\bar Z_i = Z_{ii}+Z_{ij}$ is the sum of two independent $\cgauss{0}{\frac{1}{2}}$ random variables as in the enhanced channel. Since $\bar Z_i$ and $Z_i$ (the noise in the original channel) have the same probability distribution and are both i.i.d. processes across time and across users, the joint distribution of the channel $p(y_1,y_2|x_1, x_2)$ is identical for both channels, and hence they have the same feedback capacity region, \emph{i.e.}, $\mcal{C}_i(p_1,p_2)=\mcal{C}_G(p_1,p_2)$.

Next, comparing the intermediate channel and the enhanced channel, we note that any rate pair $\lp R_1,R_2\rp$ achievable in the intermediate channel is also achievable for the enhanced channel using the same pair of codes, using the fact that $\bar Y_i = Y_{ii} + Y_{ij}$. Therefore, $\mcal{C}_i(p_1,p_2) \subseteq \mcal{C}_e(p_1,p_2)$, which completes the proof.
\end{proof}

\begin{remark}
We note that a similar channel enhancement technique has been applied before by Khisti and Lapidoth \cite{KhistiLapidoth_13}, for Gaussian multiple-access channel with intermittent feedback. In that work, the variances of the random variables $Z_{ii}$ and $Z_{ij}$ are not fixed, but are arbitrary, subject to the constraint that they sum to one. Although one can optimize over the noise variances in order to obtain the tightest bound, this only results in a small and constant improvement. Hence, for simplicity, we stick to the fixed variance of $\frac{1}{2}$ for the noise variables of the enhanced channel.
\end{remark}

Using Lemma~\ref{lem:enhanced}, we can instead prove outer bounds for the enhanced channel. In Appendix~\ref{sec:ap_ob_g}, we prove the following bounds.
\begin{align}
&R_1+R_2 \leq \log \lp 1+ \INR_1+ \frac{\SNR_1+2\sqrt{\SNR_1 \cdot \INR_1}}{1+\INR_2} \rp\notag \\
&\quad+ \log \lp 1+ \INR_2+ \frac{\SNR_2+2\sqrt{\SNR_2 \cdot \INR_2}}{1+\INR_1} \rp \notag\\
&\quad + p_1 \log \lp \frac{\lp 1+2\INR_1 \rp \lp 1+\frac{\SNR_1}{\INR_2+\frac{1}{2}}\rp}{1+ \INR_1+ \frac{\SNR_1+2\sqrt{\SNR_1 \cdot \INR_1}}{1+\INR_2}} \rp  \notag\\
&\quad+ p_2 \log \lp \frac{\lp 1+2\INR_2 \rp \lp 1+\frac{\SNR_2}{\INR_1+\frac{1}{2}}\rp}{1+ \INR_2+ \frac{\SNR_2+2\sqrt{\SNR_2 \cdot \INR_2}}{1+\INR_1}} \rp \label{eq:ob_g_R1R2}\\
&2R_1+R_2 \leq \log\lp 1 + \SNR_1 + \INR_1 + 2\sqrt{\SNR_1\cdot\INR_1} \rp \notag\\
&\quad+ \log\lp1 + \frac{\SNR_1}{\frac{1}{2}+\INR_2}\rp   \notag\\
&\quad + \log \lp 1+ \INR_2+ \frac{\SNR_2+2\sqrt{\SNR_2 \cdot \INR_2}}{1+\INR_1} \rp\notag\\
&\quad+ p_2 \log \lp \frac{\lp 1+2\INR_2 \rp \lp 1+\frac{\SNR_2}{\INR_1+\frac{1}{2}}\rp}{1+ \INR_2+ \frac{\SNR_2+2\sqrt{\SNR_2 \cdot \INR_2}}{1+\INR_1}} \rp \label{eq:ob_g_2R1R2}\\
&R_1+2R_2 \leq \log\lp 1 + \SNR_2 + \INR_2 + 2\sqrt{\SNR_2\cdot\INR_2} \rp\notag\\
&\quad + \log\lp1 + \frac{\SNR_2}{\frac{1}{2}+\INR_1}\rp \notag\\
&\quad + \log \lp 1+ \INR_1+ \frac{\SNR_1+2\sqrt{\SNR_1 \cdot \INR_1}}{1+\INR_2} \rp \notag\\
&\quad+ p_1 \log \lp \frac{\lp 1+2\INR_1 \rp \lp 1+\frac{\SNR_1}{\INR_2+\frac{1}{2}}\rp}{1+ \INR_1+ \frac{\SNR_1+2\sqrt{\SNR_1 \cdot \INR_1}}{1+\INR_2}} \rp \label{eq:ob_g_R12R2}
\end{align}

%% file: Discussion.tex
We considered the interference channel with intermittent feedback, and derived
an approximate characterization of the capacity region under Gaussian model,
as well as an exact characterization for the linear deterministic case. The result
shows that even intermittent feedback provides multiplicative gain in capacity in
interference channels. The achievability result was based on quantize-map-forward
relaying at the transmitters, and the outer bound result was based on a channel
enhancement technique.

In this paper, we considered short messaging, \emph{i.e.}, a new message is sent
at every block of transmission. An alternate approach one could try is long messaging, 
where the transmitters send codewords describing the same message at every block,
and the receivers jointly decode all blocks to recover the message. The clear advantage
of short messaging approach is better delay performance, since each message is decoded
immediately after the transmission of the corresponding block, instead of waiting for
the end of the entire transmission. However, combined with forward decoding, the rate region
achievable by this strategy cannot approximate the entire capacity region by itself, as can 
be seen from the results of Section~\ref{sec:achievability}; we need to take the union with
Han-Kobayashi rate region to approximate the entire capacity region. This is because while decoding
block $b$, part of the message of block $b+1$ is jointly decoded by treating the interference
of block $b+1$ as noise, which limits the rate in certain operating points. Hence, long-messaging
approach would remove the need for taking union with Han-Kobayashi region and simplify the
proof, since all blocks are jointly decoded. Such an approach has been taken in \cite{Zaidi_14}
to derive an inner bound on the capacity region of interference channels with generalized
feedback, which overlaps with the capacity region \eqref{eq:ldc_R1}--\eqref{eq:ldc_R12R2} for the special case of linear deterministic IC with intermittent feedback.

The extension to parallel channels is carried out for the special case of identical subchannels in this work. An important generalization can be the case where the channel gains of the subchannels are not necessarily the same. The main obstacle in generalizing our achievable scheme to this case is that it distinguishes the cases of weak and strong interference, although such a separation is not possible for vector channels. Again, long-messaging can be a strategy to circumvent this issue, since it removes the need for making such a distinction between weak and strong interference regimes \cite{Zaidi_14}, albeit at the cost of a much larger delay.

Another important extension could be to the additive white Gaussian noise (AWGN) feedback model of \cite{LeTandon_12}. Since this model assumes passive feedback as well, our quantize-map-forward based scheme can be directly applied to to this channel model. The results of this paper indicate that quantize-map-forward, as a feedback strategy, might be a promising candidate as an approximately-capacity-achieving scheme for AWGN feedback model. However, this investigation is not the focus of this paper, and is left as future work.

%% file: AP_Listsize.tex
Choose $\epsilon>0$. We suppress the dependence of variables on block index $b$ for simplicity. Consider, for $\ijj$,
\begin{align*}
\E{2^{NK_i}} &= \E{\sum_{q_i = 1}^{2^{Nr_i}} \ind{q_i: \intypset{ X^N_{jf}, U_i^N (q_i) }} } \\
&= \sum_{q_i = 1}^{2^{Nr_i}} \Prob{ \intypset{ X^N_{jf}, U_i^N (q_i) } } \\
&= 2^{Nr_i} \Prob{ \intypset{ X^N_{jf}, U_i^N (1) } }
\end{align*}
Since $U_i^N (1) $ is generated independently from $X^N_{jf}$, by packing lemma \cite{El-GamalKim_11}, there exists $\delta(\epsilon)$ with $\delta(\epsilon) \to 0$ such that
\begin{align*}
\E{2^{NK_i}} \leq 2^{Nr_i} 2^{-N \lb I(X_{jf}; U_i) - \delta(\epsilon)/3 \rb}
\end{align*}
for all $N$. 

Next consider the variance of $2^{NK_i}$.
\begin{align*}
&var (2^{NK_i}) = var \lp \sum_{q_i = 1}^{2^{Nr_i}} \ind{q_i: \intypset{ X^N_{jf}, U_i^N (q_i) }} \rp \\ &\quad\overset{\aaaa}{=}  \sum_{q_i = 1}^{2^{Nr_i}} var \lp \ind{q_i: \intypset{ X^N_{jf}, U_i^N (q_i) }} \rp \\
&\quad=  \sum_{q_i = 1}^{2^{Nr_i}} \E{ \ind{q_i: \intypset{ X^N_{jf}, U_i^N (q_i) }} } \\
&\qquad\qquad\qquad\cdot \lp 1 - \E{ \ind{q_i: \intypset{ X^N_{jf}, U_i^N (q_i) }} } \rp \\
&\quad=  \sum_{q_i = 1}^{2^{Nr_i}} \Prob { \intypset{ X^N_{jf}, U_i^N (q_i) } }\\
&\qquad\qquad\qquad \cdot \lp 1 - \Prob{ \intypset{ X^N_{jf}, U_i^N (q_i)  }} \rp \\
&\quad\overset{\bbbb}{=}  2^{Nr_i} p_N (1 - p_N)\leq 2^{Nr_i} p_N
\end{align*}
where (a) is due to independence of the indicator variables, and we have defined $p_N:= \Prob{ \intypset{ X^N_{jf}, U_i^N (1) } }$ in (b). Hence, there exists $N_1$ such that for all $N>N_1$,
\begin{align*}
var (2^{NK_i}) \leq 2^{Nr_i} 2^{-N \left[ I(X_{jf}; U_i) - \delta(\epsilon) \right]} 
\end{align*}
for some $\delta(\epsilon)$ with $\delta(\epsilon) \to 0$ as $\epsilon \to \infty$.

Define $\eta := 2^{Nr_i} 2^{-N \left[ I(X_{jf}; U_i) - \delta(\epsilon)/3 \right]} (2^{N\delta(\epsilon)/3} - 1)$, and the sequence of events
\begin{align*}
\mcal{E}_n := \lbp \left| 2^{ (n+N_1) K_i} - \E{2^{ (n+N_1) K_i}}  \right| > \eta \rbp
\end{align*}
indexed by $n\geq 1$.

Borel-Cantelli lemma \cite{Durrett_10} states that if $\sum_{n=1}^\infty \Prob {\mcal{E}_n}<\infty$, then $\Prob{\mcal{E}_n \text{ infinitely often}}=0$. Then consider
\begin{align*}
\sum_{n=1}^\infty \Prob {\mcal{E}_n} &=\sum_{n=1}^\infty \Prob {\left| 2^{ (n+N_1) K_i} - \E{2^{ (n+N_1) K_i }}  \right| > \eta} \\
& \overset{\aaaa}{\leq} \sum_{n=1}^\infty \frac{var \lp 2^{ (n+N_1) K_i } \rp}{\eta^2} \\
&\leq \sum_{n=1}^\infty \frac{1}{2^{n \lb r_i - I(X_{jf}; U_i) + \delta(\epsilon)/3 \rb} (2^{n\delta(\epsilon)/3} - 1)^2} \\
&\overset{\bbbb}{<} \infty
\end{align*}
where (a) follows by Chebyshev's inequality, and (b) is because exponentially decaying series converge, and $r_i > I(\wtild V_j; U_i) \geq I(X_{jf}; U_i)$ where the first inequality is by covering lemma \cite{El-GamalKim_11}, and the second is by data processing inequality (recall that $X_{jf} - \wtild V_j - U_i$ is a Markov chain). Therefore, with probability one, there exists a finite integer $N_2 \geq N_1$ such that for all $N \geq N_2$,
\begin{align*}
2^{NK_i} < 2^{Nr_i} 2^{-N \left[ I(X_{jf}, U_i) - 2\delta(\epsilon)/3 \right]}
\end{align*}
Choosing $r_i = I(\wtild V_j; U_i) + \delta(\epsilon)/3$, taking the logarithm of both sides, and dividing by $N$, we get the desired result.

%% file: AP_Achievability.tex
\subsection{Notation} \label{subsec:notation}
We will often suppress the dependence on block index $b$ and block length $N$ for brevity.

For any given set of message indices $\lp m_1, n_1, m_2\rp$, define the following events, with a little abuse of notation
\begin{align*}
&T(m_1, n_1, m_2) := \lbp \exists \lp q_1, q_2, q_2 ' \rp \text{ s.t. \eqref{eq:weak_fb_dec} holds for the} \right. \\ & \left. \text{indices $(m_1, n_1, m_2, q_1, q_2, q_2 ')$} \rbp, \\
&T(m_1, n_1, m_2, q_1, q_2, q_2 ') := \lbp \text{\eqref{eq:weak_fb_dec} holds for the indices} \right. \\  &\left. \text{$(m_1, n_1, m_2, q_1, q_2, q_2 ')$} \rbp. 
\end{align*}
We also define the following quantization error event at Tx$i$
\begin{align*}
E_i &= \lbp \wtild V_j^N \in \mcal{T}^{(N)}_{\epsilon '}, \; \nintypset{\wtild V_j^N, U_i^N (q_i)} \; \forall q_i \rbp \\
 &\quad \cup \lbp \wtild V_j^N \notin \mcal{T}^{(N)}_{\epsilon '} \rbp
\end{align*}
for $\ijj$, and $E := E_1 \cup E_2$.

Without loss of generality, we assume that the correct message and quantization indices correspond to the index 1, \emph{i.e.} $\lp W_{1c}, W_{1p}, W_{2c}, Q_1, Q_2 \rp = \lp 1,1,1,1,1 \rp$ for all blocks. We introduce the notation 
\begin{align*}
\mcal{\bar B}_i (b) := \mcal{B}_i (b) \setminus \lbp1\rbp.
\end{align*}
An arbitrary element of the set $\mcal{\bar B}_i (b)$ will be denoted with $\bar q_i$, or $\bar q_i '$.
In this analysis, we focus on an arbitrary block $b$, but we will also need to refer to variables from block $b-1$. The variables associated with block $b-1$ will be represented with a caron notation when in single letter form. For example, while $\check X_{2e}$ is the single letter form for $X_{2e}^N (b-1)$, $X_{2e}$ is the single letter form for $X_{2e}^N (b)$. The feedback state pair $\ul{S}=\lp S_1,S_2\rp$ is assumed to be conditioned upon in all the mutual information terms (since the receivers have access to this information causally), but will be omitted for brevity.

\subsection{Claims}
In this subsection, we will prove two simple claims that will be useful in bounding the probability of decoding error.
\begin{claim} \label{cl:generalized_ub}
Let $A_k$, $k=1,2,...$ be a sequence of i.i.d. events. Let $\mcal{S} \subset \mbb{N}$ be a random subset of natural numbers (not necessarily independent from the events $A_k$) such that $\left| \mcal{S}\right| \leq M$ a.s. for some real number $M$, and $\Prob{A_k | \mcal{S}} = \Prob{A_m | \mcal{S}}$ a.s. for all $(k,m)$ pairs. Then
\begin{align*}
\Prob{\bigcup_{k \in \mcal{S}} A_k} \leq M \Prob{A_j}
\end{align*}
for an arbitrary $j$.
\end{claim}
\begin{proof}
\begin{align*}
\Prob{\bigcup_{k \in \mcal{S}} A_k} &= \E{ \E{ \inds{\cup_{k\in \mcal{S}} A_k} | \mcal{S} } } \leq \E{ \E{ \left. \sum_{k \in \mcal{S}} \inds{A_k} \right| \mcal{S} } } \\
&\overset{\aaaa}{=} \E{ \E{\left| \mcal{S}\right| \inds{A_1} | \mcal{S} } } = \E{ \left|\mcal{S}\right| \E{ \inds{A_1} | \mcal{S} } } \\
&\leq \E{ M \E{ \inds{A_1} | \mcal{S} } } = M \E{ \E{ \inds{A_1} | \mcal{S} } } = M \Prob{A_1}
\end{align*}
where (a) follows by the fact that $\Prob{A_k | \mcal{S}}$ is the same for all $k$.
\end{proof}
\begin{claim} \label{cl:generalized_pl}
Let $\lp X^N,Y^N,Z^N \rp$ be distributed i.i.d. according to $p(x,y,z)$, and $\lp \wtild X^N, \wtild Y^N, \wtild Z^N \rp$ be distributed i.i.d. according to $p(x)p(y)p(z)$. Then there exists $\delta (\epsilon)$ with $\lim_{\epsilon \to 0} \delta (\epsilon) = 0$ such that
\begin{align*}
\Prob{\intypset{\wtild X^N, \wtild Y^N, \wtild Z^N}} \leq 2^{-N \lb I(X;Y) + I(Z;X,Y) - \delta(\epsilon) \rb}
\end{align*}
\end{claim}
\begin{proof}
\begin{align*}
&\Prob{\intypset{\wtild X^N, \wtild Y^N, \wtild Z^N}} \\
&\qquad \leq 2^{-N \lb D\lp P_{X,Y,Z}||P_X P_Y P_Z\rp-\delta(\epsilon)\rb} \\
&\qquad = 2^{-N \lb I(X;Y) + I(Z;X,Y) - \delta(\epsilon) \rb}
\end{align*}
where $D\lp P||Q\rp$ is the relative entropy between probability distributions $P$ and $Q$.
\end{proof}

\subsection{Proof of Lemma~\ref{lem:weak_fb}}
We will show that there exists a sequence of codes such that $\Prob{\mcal{D}_{FB,w}} \to 0$ \emph{exponentially}, if the given rate constraints are satisfied, which implies the claimed result. The probability of the decoding error event $\mcal{D}_{FB,w}$ can be bounded by
\begin{align}
\Prob{\mcal{D}_{FB,w}} &= \Prob{E} \Prob{\mcal{D}_{FB,w} | E} + \Prob{E^c} \Prob{\mcal{D}_{FB,w} | E^c} \notag \\
&\leq \Prob{E} + \Prob{\mcal{D}_{FB,w} | E^c} \notag \\
&\leq \Prob{E_1} + \Prob{E_2} + \Prob{\mcal{D}_{FB,w} | E^c} \label{eq:error1}
\end{align}
If we choose the rates of the quantization codebooks such that $r_i > I(\wtild V_j; U_i)$, for $\ijj$, by covering lemma \cite{El-GamalKim_11}, $\Prob{E_1}, \Prob{E_2} \to 0$. Therefore, it is sufficient to show that $\Prob{\mcal{D}_{FB,w} | E^c}$ vanishes if the conditions in the lemma are satisfied.

The decoding error event $\mcal{D}_{FB,w}$ can also be expressed as the following union of events.
\begin{align}
\mcal{D}_{FB,w} &= \bigcup_{m_1 \neq 1} T(m_1, 1, 1) \cup \bigcup_{n_1 \neq 1} T(1, n_1, 1) \notag\\
&\quad \cup \bigcup_{m_2 \neq 1} T(1, 1, m_2) \cup \bigcup_{\substack{m_1 \neq 1 \\ n_1 \neq 1}} T(m_1, n_1, 1) \notag\\
&\quad \cup \bigcup_{\substack{m_1 \neq 1 \\ m_2 \neq 1}} T(m_1, 1, m_2) \cup \bigcup_{\substack{n_1 \neq 1 \\ m_2 \neq 1}} T(1, n_1, m_2) \notag\\
&\quad \cup \bigcup_{\substack{m_1 \neq 1 \\ n_1 \neq 1 \\ m_2 \neq 1}} T(m_1, n_1, m_2) \cup T^c (1,1,1)  \label{eq:error2}
\end{align}
Using the union bound on \eqref{eq:error2}, probability of decoding error conditioned on quantization success can be bounded by
\begin{align}
&\Prob{\mcal{D}_{FB,w}|E^c} = \notag\\
&\quad 2^{NR_{1c}} \Prob{T(m_1, 1, 1)|E^c} + 2^{NR_{1p}} \Prob{T(1, n_1, 1)|E^c} \notag \\
&\quad + 2^{NR_{2c}} \Prob{T(1, 1, m_2)|E^c} + 2^{NR_{1}} \Prob{T(m_1, n_1, 1)|E^c} \notag \\
&\quad + 2^{N(R_{1c}+R_{2c})} \Prob{T(m_1, 1, m_2)|E^c} \notag\\
&\quad + 2^{N(R_{1p}+R_{2c})} \Prob{T(1, n_1, m_2)|E^c} \notag \\
&\quad + 2^{N(R_{1}+R_{2c})} \Prob{T(m_1, n_1, m_2)|E^c} \notag\\
&\quad + \Prob{T^c (1,1,1)|E^c} \label{eq:dec_error} 
\end{align}
Note that conditioned on successful quantization, the relevant random variables are distributed i.i.d. over time according to the joint distribution
\begin{align}
&p(\check x_{1p}, \check x_{2e}, \check x_{2c}, u_1, x_{1e}, x_{1c}, x_{2e}, \check y_1, y_1) = \notag\\
&\quad p(\check x_{1p}) p(\check x_{2e}) p(\check x_{2c}) p( x_{1e}) p( x_{1c}) p(x_{2e}) \notag \\
&\quad \cdot p(u_1| \check y_1,\check x_{1p} ) p(\check y_1 | \check x_{2e}, \check x_{2c}, \check x_{1p}) p(y_1|x_{1e}, x_{1c}, x_{2e}) \label{eq:joint_dist}
\end{align}
Next, we bound the error terms one by one. In what follows, joint typicality is sought with respect to the joint distribution in $(\ref{eq:joint_dist})$. The first term is bounded by $\Prob{T^c (1,1,1)|E^c}<\epsilon$ by law of large numbers. 
\begin{figure}
\centering
\includegraphics[scale=0.9]{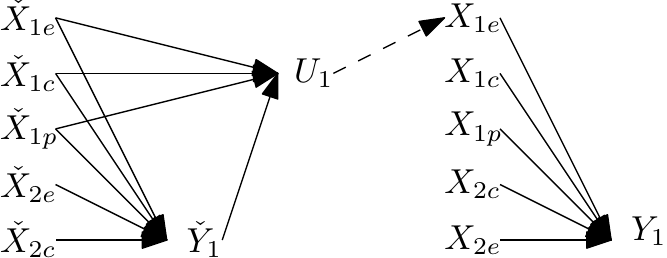}
\centering
\caption{Markov network showing the dependence of the relevant variables. The variables connected with the dashed arrow are independent in the single-letter form, although in multi-letter form they are not.}
\label{fig:dependence}
\end{figure}

The second error term in \eqref{eq:dec_error} can be bounded as follows.
\begin{align*}
&\Prob{T(m_1, 1, 1)|E^c} = \mbb{P} \lp \left. \bigcup_{q_1, q_2, q_2 '} T(m_1, 1, 1, q_1, q_2, q_2 ') \right| E^c \rp \\
&= \mbb{P} \lp \bigcup_{\substack{q_1 \neq 1, q_2 \neq 1 \\ q_2 ' \neq 1}}  T(m_1, 1, 1, q_1, q_2, q_2 ')\right.\\
&\quad \cup \bigcup_{\substack{q_2 \neq 1 \\ q_2 ' \neq 1}} T(m_1, 1, 1, 1, q_2, q_2 ') \cup \bigcup_{\substack{q_1 \neq 1 \\ q_2 ' \neq 1}} T(m_1, 1, 1, q_1, 1, q_2 ') \\
&\quad \cup \bigcup_{\substack{q_1 \neq 1 \\ q_2  \neq 1}} T(m_1, 1, 1, q_1, q_2, 1) \cup \bigcup_{q_1 \neq 1} T(m_1, 1, 1, q_1, 1, 1) \\
&\quad \cup \bigcup_{ q_2  \neq 1} T(m_1, 1, 1, 1, q_2, 1) \cup \bigcup_{q_2 ' \neq 1} T(m_1, 1, 1, 1, 1, q_2 ') \\
&\quad \cup T(m_1, 1, 1, 1, 1, 1)  \left. \left| E^c \vphantom{\bigcup_{\substack{q_1 \neq 1, q_2 \neq 1 \\ q_2 ' \neq 1}}} \right. \rp \\
&\overset{\aaaa}{=} \mbb{P} \left( \bigcup_{\substack{q_1 \in \mcal{\bar B}_1 (1,q_2) \\ q_2 \in \mcal{\bar B}_2 (b-1), q_2 ' \in \mcal{\bar B}_2 (b)}}  T(m_1, 1, 1, q_1, q_2, q_2 ') \right. \\
&\quad \cup \bigcup_{\substack{q_2 \in \mcal{\bar B}_2 \\ q_2 ' \in \mcal{\bar B}_2 (b)}} T(m_1, 1, 1, 1, q_2, q_2 ') \\
&\quad \cup \bigcup_{\substack{q_1 \in \mcal{\bar B}_1 (b) \\ q_2 ' \in \mcal{\bar B}_2 (b)}} T(m_1, 1, 1, q_1, 1, q_2 ') \\
&\quad \cup \bigcup_{\substack{q_1 \in \mcal{\bar B}_1 (1,q_2) \\ q_2 \in \mcal{\bar B}_2 (b-1)}} T(m_1, 1, 1, q_1, q_2, 1) \\
&\quad \cup \bigcup_{q_1 \in \mcal{\bar B}_1 (b)} T(m_1, 1, 1, q_1, 1, 1) \\
&\quad \cup \bigcup_{ q_2 \in \mcal{\bar B}_2 (b-1)} T(m_1, 1, 1, 1, q_2, 1) \\
&\quad \cup \bigcup_{q_2 ' \in \mcal{\bar B}_2 (b)} T(m_1, 1, 1, 1, 1, q_2 ')\\
&\quad \cup T(m_1, 1, 1, 1, 1, 1) \left. \left| E^c  \vphantom{\bigcup_{\substack{q_1 \neq 1, q_2 \neq 1 \\ q_2 ' \neq 1}}} \right. \rp \\
&\overset{\bbbb}{\leq} \Prob{ \left. \bigcup_{\substack{q_1 \in \mcal{\bar B}_1 (1,q_2) \\ q_2 \in \mcal{\bar B}_2 (b-1), q_2 ' \in \mcal{\bar B}_2 (b)}} T(m_1, 1, 1, q_1, q_2, q_2 ') \right| E^c} \\
&\quad + \Prob{ \left. \bigcup_{\substack{q_2 \in \mcal{\bar B}_2 (b-1) \\ q_2 ' \in \mcal{\bar B}_2 (b)}} T(m_1, 1, 1, 1, q_2, q_2 ') \right| E^c } \\
&\quad + \Prob{ \left. \bigcup_{\substack{q_1 \in \mcal{\bar B}_1 (b) \\ q_2 ' \in \mcal{\bar B}_2 (b)}} T(m_1, 1, 1, q_1, 1, q_2 ') \right| E^c } \\
&\quad + \Prob{ \left. \bigcup_{\substack{q_1 \in \mcal{\bar B}_1 (1,q_2(b-1)) \\ q_2 \in \mcal{\bar B}_2 (b-1)}} T(m_1, 1, 1, q_1, q_2, 1) \right| E^c }  \\
&\quad +\Prob{ \left. \bigcup_{q_1 \in \mcal{\bar B}_1 (b)} T(m_1, 1, 1, q_1, 1, 1) \right| E^c}\\
&\quad +\Prob{ \left. \bigcup_{ q_2 \in \mcal{\bar B}_2 (b-1)} T(m_1, 1, 1, 1, q_2, 1) \right| E^c }\\
&\quad + \Prob{ \left. \bigcup_{q_2 ' \in \mcal{\bar B}_2 (b)} T(m_1, 1, 1, 1, 1, q_2 ') \right| E^c} \\
&\quad + \Prob{ T(m_1, 1, 1, 1, 1, 1) |E^c } \\
&\overset{\cccc}{=} 2^{N(\kappa_1+2\kappa_2)} \Prob{ T(m_1, 1, 1,  q_1, q_2, q_2 ')|E^c } \\
&\quad + 2^{2N\kappa_2} \Prob{ T(m_1, 1, 1, 1, q_2, q_2 ') |E^c } \\
&\quad + 2^{N(\kappa_1+\kappa_2)} \Prob{ T(m_1, 1, 1, q_1, 1, q_2 ') |E^c } \\
&\quad + 2^{N(\kappa_1+\kappa_2)} \Prob{ T(m_1, 1, 1, q_1, q_2, 1)|E^c  } \\
&\quad +2^{N\kappa_1} \Prob{ T(m_1, 1, 1, q_1, 1, 1)  |E^c} \\
&\quad + 2^{N\kappa_2} \Prob{ T(m_1, 1, 1, 1, q_2, 1) |E^c }\\
&\quad +2^{N\kappa_2} \Prob{ T(m_1, 1, 1, 1, 1, q_2 ')  |E^c} \\
&\quad + \Prob{ T(m_1, 1, 1, 1, 1, 1) |E^c } \\
&\overset{\dddd}{\leq} 2^{N(\kappa_1+2\kappa_2)} \\
&\quad \cdot 2^{-N\lb I (U_1;\check X_{2e})+ I (\check X_{2e}, U_1, X_{1f}, X_{2e}; Y_1, \check Y_1, \check X_{2c}| \check X_{1}) - \delta(\epsilon) \rb} \\
&\quad + 2^{2N\kappa_2} 2^{-N \lb I(X_{1f}, X_{2e}, \check X_{2e};\check Y_1, Y_1 | U_1, X_{1e}, \check X_{2c}, \check X_{1}) - \delta(\epsilon)\rb} \\
&\quad + 2^{N(\kappa_1+\kappa_2)} \\
&\quad \cdot  2^{-N\lb I (U_1;\check X_{2e})+ I (\check X_{2e}, U_1, X_{1f}; Y_1, \check Y_1, \check X_{2c}| \check X_{1}, X_{2e}) - \delta(\epsilon) \rb} \\ 
&\quad + 2^{N(\kappa_1+\kappa_2)} 2^{-N\lb I(X_{1f}, X_{2e}, U_1; Y_1, \check Y_1, \check X_{2f}| \check X_1)-\delta(\epsilon)\rb} \\
&\quad +2^{N\kappa_1} 2^{-N\lb I(X_{1f}, U_1; Y_1, \check Y_1, \check X_{2f}| \check X_1, X_{2e})-\delta(\epsilon)\rb}\\
&\quad + 2^{N\kappa_2} 2^{-N\lb I(X_{1f}, \check X_{2e}; U_1, Y_1, \check Y_1|X_{1e}, X_{2e}, \check X_1, \check X_{2c})-\delta(\epsilon)\rb}  \\
&\quad +2^{N\kappa_2} 2^{-N\lb I(X_{1f}, X_{2e}; Y_1| X_{1e}, U_1, \check Y_1, \check X_1, \check X_{2f})-\delta(\epsilon)\rb} \\
&\quad + 2^{-N\lb I(X_{1f};Y_1| X_{1e}, U_1, \check Y_1, \check X_1, \check X_{2f}, X_{2e})-\delta(\epsilon)\rb} \\
&\overset{\eeee}{=} 2^{N \lp  \kappa_1 + 2\kappa_2 \rp}\\
&\quad \cdot 2^{-N\lb I(U_1;\check X_{2e}) + I (X_{1f}, X_{2e}; Y_1)+I (\check X_{2e}; \check Y_1|\check X_1, \check X_{2c}) - \delta(\epsilon) \rb}\\
&\quad + 2^{-N \lb I(X_{1f}, X_{2e}, \check X_{2e};\check Y_1, Y_1 | U_1, X_{1e}, \check X_{2c}, \check X_{1}) -2\kappa_2 - \delta(\epsilon)\rb} \\
&\quad +2^{N \lp  \kappa_1 + \kappa_2 \rp} \\
&\quad \cdot 2^{-N \lb I(U_1;\check X_{2e}) + I (X_{1f}; Y_1|X_{2e}) +I (\check X_{2e}; \check Y_1|\check X_1, \check X_{2c}) - \delta(\epsilon) \rb} \\
&\quad + 2^{-N\lb I(X_{1f}, X_{2e}; Y_1)+ I (U_1; \check Y_1 | \check X_1, \check X_{2f})- \kappa_1 - \kappa_2-\delta(\epsilon)\rb} \\
&\quad +2^{-N\lb I(X_{1f}; Y_1|X_{2e})+ I (U_1; \check Y_1 | \check X_1, \check X_{2f})-\kappa_1-\delta(\epsilon)\rb}\\
&\quad +2^{-N\lb I(X_{1f}, \check X_{2e}; U_1, Y_1, \check Y_1|X_{1e}, X_{2e}, \check X_1, \check X_{2c})-\kappa_2-\delta(\epsilon)\rb}  \\
&\quad +2^{-N\lb I(X_{1f}, X_{2e}; Y_1|X_{1e})-\kappa_2-\delta(\epsilon)\rb} \\
&\quad + 2^{-N\lb I(X_{1f};Y_1| X_{1e}, X_{2e})-\delta(\epsilon)\rb} \\
&\overset{\ffff}{\leq} 8 \cdot 2^{-N\lb I(X_{1f};Y_1| X_{1e}, X_{2e})-C_1-\delta(\epsilon)\rb}
\end{align*}
where
\begin{itemize}
\item (a) is since $T(m_1, 1, 1, q_1, q_2, q_2 ')$ is empty set for $q_1 \notin \mcal{B}_1 (b)$, $q_2 \notin \mcal{B}_2 (b-1)$, or $q_2 ' \notin \mcal{B}_2 ((q_2, m_2)(b))$, since for random variables $(X,Y,Z) \sim p(x,y,z)$, $\intypset{X,Y,Z}$ implies $\intypset{X,Y}$,
\item (b) follows by union bound,
\item (c) follows by Claim \ref{cl:generalized_ub}, where the upper bound on the size of the $\mcal{\bar B}_i$ sets for sufficiently large $N$ is given by Lemma \ref{lem:listsize},
\item (d) follows by packing lemma and Claim~\ref{cl:generalized_pl},
\item (e) is by manipulating the mutual information terms using the dependence structure of the involved variables (see Figure~\ref{fig:dependence}),
\item (f) is by upper bounding each of the eight terms with the same bound, using chain rule and non-negativity of mutual information.
\end{itemize}

Next, we bound the term $\Prob{T(1, n_1, m_2)|E^c}$. We apply steps (a)-(d), which are also applicable here, to obtain the following.
\begin{align*}
&\Prob{T(1, n_1, m_2)|E^c} \\
&\quad \leq 2^{N(\kappa_1+2\kappa_2)} \\
&\quad \cdot 2^{-N\lb I (U_1;\check X_{2f})+ I (\check X_1, \check X_{2f}, X_{2e}, U_1, X_{1e}; Y_1, \check Y_1, X_{1c}| \check X_{1f}) - \delta(\epsilon) \rb}\\
&\quad + 2^{2N\kappa_2} 2^{-N \lb I(\check X_{1}, \check X_{2f}, X_{2e};\check Y_1, Y_1, U_1| X_{1f}, \check X_{1f}) - \delta(\epsilon)\rb} \\
&\quad + 2^{N(\kappa_1+\kappa_2)} \\
&\quad \cdot 2^{-N\lb I (U_1;\check X_{2f})+ I (\check X_{2f}, \check X_{1}, U_1, X_{1e}; Y_1, \check Y_1| \check X_{1f}, X_{1c}, X_{2e}) - \delta(\epsilon) \rb} \\
&\quad + 2^{N(\kappa_1+\kappa_2)} \\
&\quad \cdot 2^{-N\lb I (\check X_{2f}, \check X_{1}, U_1, X_{1e}, X_{2e}; Y_1, \check Y_1| \check X_{1f}, X_{1c}, \check X_{2e})-\delta(\epsilon)\rb}\\
&\quad +2^{N\kappa_1} 2^{-N\lb I(\check X_{1}, \check X_{2f}, U_1, X_{1e}; Y_1, \check Y_1, \check X_{2e}| \check X_{1f}, X_{1c}, X_{2e})-\delta(\epsilon)\rb}\\
&\quad + 2^{N\kappa_2} 2^{-N\lb I(\check X_1, \check X_{2f}; U_1, Y_1, \check Y_1|X_{1f}, X_{2e}, \check X_{1f})-\delta(\epsilon)\rb}  \\
&\quad +2^{N\kappa_2} 2^{-N\lb I(\check X_1, \check X_{2f}, X_{2e}; Y_1, \check Y_1, U_1 | \check X_{1f}, X_{1f}, \check X_{2e})-\delta(\epsilon)\rb} \\
&\quad + 2^{-N\lb I(\check X_{1}, \check X_{2f};\check Y_1, U_1| Y_1, \check X_{1f}, \check X_{2e}, X_{1e}, X_{2e})-\delta(\epsilon)\rb} \\
&\overset{\eeee}{\leq} 4 \cdot 2^{-N\lb I (\check X_{2f}, \check X_{1}, U_1, X_{1e}; Y_1, \check Y_1| \check X_{1f}, \check X_{2e}, X_{1c}, X_{2e})-\delta(\epsilon)\rb} \\
&\quad + 4 \cdot 2^{-N\lb I (\check X_{2f}, \check X_{1}; Y_1, \check Y_1, U_1| \check X_{1f}, \check X_{2e}, X_{1f}, X_{2e})-\delta(\epsilon)\rb} \\
&\overset{\ffff}{\leq} 4 \cdot 2^{-N\lb I (X_{2f}, X_{1}; Y_1| X_{1c}, X_{2e})-\delta(\epsilon)\rb} \\
&\quad + 4 \cdot 2^{-N\lb I (\check X_{2f}, \check X_{1}; \check Y_1, U_1| \check X_{1f}, \check X_{2e})-\delta(\epsilon)\rb}
\end{align*}
where (e) is by upper bounding the first, third, fourth and fifth terms with the first term in (j), and the rest of the terms with the second; (f) is by rearranging the mutual information terms using chain rule and the fact that the distribution of variables are the same for each block.

In order to bound the term $\Prob{T(1, n_1, 1)|E^c}$ in \eqref{eq:dec_error}, we note that the the joint distribution \eqref{eq:joint_dist} has a similar structure with respect to $X_{1c}$ and $\check X_{1p}$, with the following mapping between random variables
\begin{align*}
\check X_{1p} &\leftrightarrow X_{1c}, \\
\lp \check Y_1, U_1 \rp &\leftrightarrow Y_1, \\
\lp \check X_{2e}, \check X_{2c} \rp &\leftrightarrow \lp X_{1e}, X_{2e} \rp
\end{align*}
Therefore, one can perform the steps (a)-(f) for the third error term as well, by switching the variables as above, to obtain the following bound
\begin{align*}
\Prob{T(1, n_1, 1)|E^c} &\leq 4 \cdot 2^{-N\lb I(\check X_{1};\check Y_1, U_1| \check X_{1f}, \check X_{2f})-C_1-\delta(\epsilon)\rb} \\
&\quad + 4 \cdot 2^{-N\lb I(\check X_{1};\check Y_1, U_1| \check X_{1c}, \check X_{2f})-C_1-\delta(\epsilon)\rb} \\
&\leq 8 \cdot 2^{-N\lb I(\check X_{1};\check Y_1| \check X_{1f}, \check X_{2f})-C_1-\delta(\epsilon)\rb}\\
&=8 \cdot 2^{-N\lb I(X_{1};Y_1| X_{1f}, X_{2f})-C_1-\delta(\epsilon)\rb}
\end{align*}
We have dropped the $U_1$ variable from the mutual information term for the sake of simplicity in evaluating the rate region, since its contribution is small. In the final step, we used the fact that the distribution of variables is the same for each block. We can obtain the following bounds for each error term in a similar way, by exploiting the structure of the joint distribution as done above and noting that the steps (a)-(f) are applicable with an appropriate mapping between the variables.
\begin{align*}
&\Prob{T(1, 1, m_2)|E^c} \\
&\quad \leq 8 \cdot 2^{-N\lb I(\check X_{2f};\check Y_1, U_1| \check X_{1}, \check X_{2e})-C_1-\delta(\epsilon)\rb} \\
&\quad \leq 8 \cdot 2^{-N\lb I(\check X_{2f};\check Y_1| \check X_{1}, \check X_{2e})-C_1-\delta(\epsilon)\rb} \\
&\quad = 8 \cdot 2^{-N\lb I(X_{2f};Y_1| X_{1}, X_{2e})-C_1-\delta(\epsilon)\rb} \\
&\Prob{T(m_1, n_1, 1)|E^c} \\
&\quad \leq 8 \cdot 2^{-N\lb I(\check X_{1}, X_{1f};\check Y_1, Y_1, U_1| \check X_{1f}, \check X_{2f}, X_{1e}, X_{2e})-C_1-\delta(\epsilon)\rb} \\
&\quad \leq 8 \cdot 2^{-N\lb I(\check X_{1}, X_{1f};\check Y_1, Y_1| \check X_{1f}, \check X_{2f}, X_{1e}, X_{2e})-C_1-\delta(\epsilon)\rb} \\
&\quad = 8 \cdot 2^{-N\lb I(\check X_{1};\check Y_1| \check X_{1f}, \check X_{2f})+I(X_{1f};Y_1| X_{1e}, X_{2e})-C_1-\delta(\epsilon)\rb} \\
&\quad = 8 \cdot 2^{-N\lb I(X_{1};Y_1| X_{1f}, X_{2f})+I(X_{1f};Y_1| X_{1e}, X_{2e})-C_1-\delta(\epsilon)\rb} \\
&\Prob{T(m_1, 1, m_2)|E^c} \\
&\quad \leq 8 \cdot 2^{-N\lb I(\check X_{2f}, X_{1f};\check Y_1, Y_1, U_1| \check X_{1}, \check X_{2e}, X_{1e}, X_{2e})-C_1-\delta(\epsilon)\rb} \\
&\quad \leq 8 \cdot 2^{-N\lb I(\check X_{2f}, X_{1f};\check Y_1, Y_1| \check X_{1}, \check X_{2e}, X_{1e}, X_{2e})-C_1-\delta(\epsilon)\rb} \\
&\quad = 8 \cdot 2^{-N\lb I(\check X_{2f};\check Y_1| \check X_{1}, \check X_{2e})+I(X_{1f};Y_1| X_{1e}, X_{2e})-C_1-\delta(\epsilon)\rb} \\
&\quad = 8 \cdot 2^{-N\lb I(X_{2f};Y_1| X_{1}, X_{2e})+I(X_{1f};Y_1| X_{1e}, X_{2e})-C_1-\delta(\epsilon)\rb} \\
&\Prob{T(m_1, n_1, m_2)|E^c} \\
&\quad \leq 8 \cdot 2^{-N\lb I(\check{X_1}, \check X_{2f}, X_{1f};\check Y_1, Y_1, U_1| \check X_{1f}, \check X_{2e}, X_{1e}, X_{2e})-C_1-\delta(\epsilon)\rb} \\
&\quad \leq 8 \cdot 2^{-N\lb I(\check X_1, \check X_{2f}, X_{1f};\check Y_1, Y_1| \check X_{1f}, \check X_{2e}, X_{1e}, X_{2e})-C_1-\delta(\epsilon)\rb} \\
&\quad = 8 \cdot 2^{-N\lb I(\check X_1, \check X_{2f};\check Y_1| \check X_{1f}, \check X_{2e})+I(X_{1f};Y_1|X_{1e}, X_{2e})-C_1-\delta(\epsilon)\rb} \\
&\quad = 8 \cdot 2^{-N\lb I( X_1, X_{2f}; Y_1| X_{1f}, X_{2e})+I(X_{1f};Y_1|X_{1e}, X_{2e})-C_1-\delta(\epsilon)\rb} \\
&\quad = 8 \cdot 2^{-N\lb I( X_1, X_{2f}; Y_1| X_{1e}, X_{2e})-C_1-\delta(\epsilon)\rb}
\end{align*}
Using these bounds in \eqref{eq:dec_error}, it is easy to see that if the following are satisfied, then $\Prob{\mcal{D}_{FB,w}|E^c} \to 0$ as $N \to \infty$ (Note that the bounds on $R_{1p}+R_{1c}$ and $R_{1c}+R_{2c}$ are redundant, as they can be expressed as a sum of other bounds),
\begin{align}
R_{1p} &< I(X_1;Y_1|X_{1f}, X_{2f}) - C_1 \label{eq:weak_fb_rc_first} \\
R_{1c} &< I(X_{1f};Y_1|X_{1e}, X_{2e}) - C_1 \\
R_{1p}+R_{2c} &< \min \lbp I (X_{2f}, X_{1}; Y_1| X_{1c}, X_{2e}), \right. \notag\\
&\left. I (X_{2f}, X_{1}; Y_1, U_1| X_{1f}, X_{2e}) \rbp - C_1 \\
R_1+R_{2c} &< I(X_1,X_{2f};Y_1|X_{1e}, X_{2e}) - C_1 \label{eq:weak_fb_rc_last} 
\end{align}
The rate constraint on $R_{1p}+R_{2c}$ provided in the lemma is slightly stricter, which allows us to show the redundancy of some of the bounds obtained later.

\subsection{Proof of Lemma~\ref{lem:strong_fb}}
We will show that there exists a sequence of codes such that $\Prob{\mcal{D}_{FB,s}} \to 0$ \emph{exponentially}, if the given rate constraints are satisfied, which implies the claimed result. 

Similar to the case of weak interference, choosing the quantization rates such that $r_i > I(\wtild V_j; U_i)$, for $\ijj$, probability of decoding error can be bounded by
\begin{align}
&\Prob{\mcal{D}_{FB,s}|E^c} \notag \\
&\quad = \Prob{T^c (1,1,1)|E^c} + 2^{NR_{1p}} \Prob{T(1, n_1, 1)|E^c} \notag \\
&\quad + 2^{NR_{2c}} \Prob{T(1, 1, m_2)|E^c} + 2^{NR_{1}} \Prob{T(m_1, n_1, 1)|E^c} \notag \\
&\quad + 2^{N(R_{1p}+R_{2c})} \Prob{T(1, n_1, m_2)|E^c} \notag\\
&\quad + 2^{N(R_{1}+R_{2c})} \Prob{T(m_1, n_1, m_2)|E^c} \label{eq:strong_dec_error} 
\end{align}
Note that conditioned on successful quantization, the relevant random variables are distributed i.i.d. over time according to the joint distribution
\begin{align}
&p(\check x_{1c}, \check x_{1p}, \check x_{2e}, x_{2c}, u_2, x_{2e}, x_{1e}, \check y_1, y_1) \notag\\
&\quad = p(\check x_{1c}) p(\check x_{1p}) p(\check x_{2e}) p( x_{1e}) p( x_{1c}) p(x_{2e}) \notag \\
&\quad \cdot p(u_2| \check x_{1c},\check x_{1p} ) p(\check y_1 | \check x_{1c}, \check x_{1c}, \check x_{2e}) p(y_1|x_{1e}, x_{2e}, x_{2c}).  \label{eq:sjoint_dist}
\end{align}
Next, we bound the error terms one by one. In what follows, joint typicality is sought with respect to the joint distribution in \eqref{eq:sjoint_dist}. The first term is bounded by $\Prob{T^c (1,1,1)|E^c}<\epsilon$ by law of large numbers. 

Now we take the third term, which is bounded as follows.
\begin{align*}
&\Prob{T(1, n_1, 1)|E^c} = \Prob{ \bigcup_{q_1, q_2, q_2 '} T(1, n_1, 1, q_1, q_2, q_2 ') |E^c} \\
&= \mbb{P} \lp \bigcup_{\substack{q_1 \neq 1, q_2 \neq 1 \\ q_2 ' \neq 1}} T(1, n_1, 1, q_1, q_2, q_2 ') \right. \\
&\quad \cup \bigcup_{\substack{q_2 \neq 1 \\ q_2 ' \neq 1}} T(1, n_1, 1, 1, q_2, q_2 ') \cup \bigcup_{\substack{q_1 \neq 1 \\ q_2 ' \neq 1}} T(1, n_1, 1, q_1, 1, q_2 ') \\
&\quad \cup \bigcup_{\substack{q_1 \neq 1 \\ q_2  \neq 1}} T(1, n_1, 1, q_1, q_2, 1) \cup \bigcup_{q_1 \neq 1} T(1, n_1, 1, q_1, 1, 1) \\
&\quad \cup \bigcup_{ q_2  \neq 1} T(1, n_1, 1, 1, q_2, 1) \cup \bigcup_{q_2 ' \neq 1} T(1, n_1, 1, 1, 1, q_2 ') \\
&\quad \cup T(1, n_1, 1, 1, 1, 1)  \left. \left| E^c \vphantom{\bigcup_{\substack{q_1 \neq 1, q_2 \neq 1 \\ q_2 ' \neq 1}}} \right. \rp \\
&\overset{\aaaa}{=} \mbb{P} \lp \bigcup_{\substack{q_1 \in \mcal{\bar B}_1 (q_2), q_2 \in \mcal{\bar B}_2 (b-1) \\ q_2 ' \in \mcal{\bar B}_2 (b)}} T(1, n_1, 1, q_1, q_2, q_2 ') \right. \\
&\quad \cup \bigcup_{\substack{q_2 \in \mcal{\bar B}_2 (b-1) \\ q_2 ' \in \mcal{\bar B}_2 (b)}} T(1, n_1, 1, 1, q_2, q_2 ') \\
&\quad \cup \bigcup_{\substack{q_1 \in \mcal{\bar B}_1 (q_2) \\ q_2 ' \in \mcal{\bar B}_2 (b)}} T(1, n_1, 1, q_1, 1, q_2 ') \\
&\quad \cup \bigcup_{\substack{q_1 \in \mcal{\bar B}_1 (q_2) \\ q_2 \in \mcal{\bar B}_2 (b-1)}} T(1, n_1, 1, q_1, q_2, 1) \\
&\quad \cup \bigcup_{q_1 \in \mcal{\bar B}_1 (q_2)} T(1, n_1, 1, q_1, 1, 1) \\
&\quad \cup \bigcup_{ q_2 \in \mcal{\bar B}_2 (b-1)} T(1, n_1, 1, 1, q_2, 1) \\
&\quad \cup \bigcup_{q_2 ' \in \mcal{\bar B}_2 ((b)} T(1, n_1, 1, 1, 1, q_2 ') \\
&\quad \cup T(1, n_1, 1, 1, 1, 1) \left. \left| E^c \vphantom{\bigcup_{\substack{q_1 \neq 1, q_2 \neq 1 \\ q_2 ' \neq 1}}} \right. \right) \\
&\overset{\bbbb}{\leq} \Prob{ \left. \bigcup_{\substack{q_1 \in \mcal{\bar B}_1 (q_2), q_2 \in \mcal{\bar B}_2 (b-1) \\ q_2 ' \in \mcal{\bar B}_2 (b)}} T(1, n_1, 1, q_1, q_2, q_2 ') \right| E^c} \\
&\quad + \Prob{ \left. \bigcup_{\substack{q_2 \in \mcal{\bar B}_2 (b-1) \\ q_2 ' \in \mcal{\bar B}_2 (b)}} T(1, n_1, 1, 1, q_2, q_2 ') \right| E^c } \\
&\quad + \Prob{ \left.\bigcup_{\substack{q_1 \in \mcal{\bar B}_1 (q_2) \\ q_2 ' \in \mcal{\bar B}_2 (b)}} T(1, n_1, 1, q_1, 1, q_2 ') \right| E^c } \\
&\quad + \Prob{ \left.\bigcup_{\substack{q_1 \in \mcal{\bar B}_1 (q_2) \\ q_2 \in \mcal{\bar B}_2 (b-1)}} T(1, n_1, 1, q_1, q_2, 1) \right| E^c }  \\
&\quad +\Prob{ \left.\bigcup_{q_1 \in \mcal{\bar B}_1 (q_2)} T(1, n_1, 1, q_1, 1, 1) \right| E^c} \\
&\quad +\Prob{ \left.\bigcup_{ q_2 \in \mcal{\bar B}_2 (b-1)} T(1, n_1, 1, 1, q_2, 1) \right| E^c }\\
&\quad + \Prob{ \left.\bigcup_{q_2 ' \in \mcal{\bar B}_2 (b)} T(1, n_1, 1, 1, 1, q_2 ') \right| E^c} \\
&\quad + \Prob{ T(1, n_1, 1, 1, 1, 1) | E^c } \\
&\overset{\cccc}{\leq} 2^{N(\kappa_1+2\kappa_2)} \Prob{ T(1, n_1, 1, q_1, q_2, q_2 ') |E^c} \\
&\quad + 2^{2N\kappa_2} \Prob{ T(1, n_1, 1, 1, q_2, q_2 ') |E^c } \\
&\quad + 2^{N(\kappa_1+\kappa_2)} \Prob{ T(1, n_1, 1, q_1, 1, q_2 ') |E^c } \\
&\quad + 2^{N(\kappa_1+\kappa_2)} \Prob{ T(1, n_1, 1, q_1, q_2, 1) |E^c }\\ 
&\quad +2^{N\kappa_1} \Prob{ T(1, n_1, 1, \bar q_1, 1, 1) |E^c } \\
&\quad + 2^{N\kappa_2} \Prob{ T(1, n_1, 1, 1, \bar q_2, 1) |E^c }\\
&\quad +2^{N\kappa_2} \Prob{ T(1, n_1, 1, 1, 1, \bar q_2 ') |E^c } \\
&\quad + \Prob{ T(1, n_1, 1, 1, 1, 1) |E^c } \\
&\overset{\dddd}{\leq} 2^{N(\kappa_1+2\kappa_2)} 2^{-N I (U_2;\check X_{1p}|\check X_{1e})}\\
&\quad \cdot 2^{-N\lb I (\check X_{1p}, \check X_{2e}, U_2, X_{1e}, X_{2e}; Y_1, \check Y_1, X_{2c}, \check X_{1c}|\check X_{1e}, \check X_{2c}) - \delta(\epsilon) \rb}\\
&\quad + 2^{2N\kappa_2} 2^{-N I (U_2;\check X_{1p})}\\
&\quad \cdot 2^{-N \lb I (\check X_{1p}, \check X_{2e}, U_2, X_{2e}; Y_1, \check Y_1, X_{1e}, X_{2c}, \check X_{1c}|\check X_{1e},\check X_{2c}) - \delta(\epsilon)\rb} \\
&\quad + 2^{N(\kappa_1+\kappa_2)} 2^{-N I (U_2;\check X_{1p})}\\
&\quad \cdot 2^{-N\lb I (\check X_{1p}, U_2, X_{1e}, X_{2e}; Y_1, \check Y_1, \check X_{2e}, X_{2c}, \check X_{1c}|\check X_{1e},\check X_{2c}) - \delta(\epsilon) \rb} \\
&\quad + 2^{N(\kappa_1+\kappa_2)}  \\
&\quad \cdot 2^{-N\lb I (\check X_{1p}, \check X_{2e}, X_{1e}; Y_1, \check Y_1, U_2, X_{2f}, \check X_{1c}|\check X_{1e}, \check X_{2c})-\delta(\epsilon)\rb} \\
&\quad +2^{N\kappa_1} \\
&\quad \cdot 2^{-N\lb I (\check X_{1p}, X_{1e}; Y_1, \check Y_1, \check X_{2f}, U_2, X_{2f}, \check X_{1c}|\check X_{1e}, \check X_{2c})-\delta(\epsilon)\rb}\\
&\quad + 2^{N\kappa_2} 2^{-N I (U_2;\check X_{1p}|\check X_{1e}) }\\
&\quad \cdot 2^{-N\lb I (\check X_1, U_2, X_{2e}; Y_1, \check Y_1, \check X_{2e}, X_{2c}, X_{1e}, \check X_{1c}|\check X_{1e}, \check X_{2c})-\delta(\epsilon)\rb}  \\
&\quad +2^{N\kappa_2} \\
&\quad \cdot 2^{-N\lb I (\check X_{1p}, \check X_{2e}; Y_1, \check Y_1, X_{2c}, U_2, X_{1e}, X_{2e}, \check X_{1c}|\check X_{1e}, \check X_{2c})-\delta(\epsilon)\rb} \\
&\quad + 2^{-N\lb I (\check X_{1p}; Y_1, \check Y_1, \check X_{2c}, \check X_{2e}, U_2, X_{1e}, X_{2e}, \check X_{1c}|\check X_{1e}, \check X_{2c})-\delta(\epsilon)\rb} \\
&\overset{\eeee}{\leq} 2^{N(\kappa_1+2\kappa_2)} \\
&\quad \cdot 2^{-N\lb I (\check X_{1p}, \check X_{2e}, X_{1e}, X_{2e};\check Y_1, Y_1, U_2|\check X_{1f}, X_{2c}, \check X_{2c}) - \delta(\epsilon) \rb}\\
&\quad + 2^{2N\kappa_2} \\
&\quad \cdot 2^{-N \lb I (\check X_{1p}, \check X_{2e}, U_2, X_{2e}; Y_1, \check Y_1|\check X_{1f}, X_{1e}, X_{2c}, \check X_{2c}) - \delta(\epsilon)\rb} \\
&\quad + 2^{N(\kappa_1+\kappa_2)} \\
&\quad \cdot 2^{-N\lb I (\check X_{1p}, X_{1e}, X_{2e};\check Y_1, Y_1, U_2|\check X_{1e}, \check X_{1c}, \check X_{2f}, X_{2c})- \delta(\epsilon) \rb} \\
&\quad + 2^{N(\kappa_1+\kappa_2)} \\
&\quad \cdot  2^{-N\lb I (\check X_{1p}, \check X_{2e}, X_{1e}; Y_1, \check Y_1, U_2|\check X_{1f}, \check X_{2c}, X_{2f})-\delta(\epsilon)\rb} \\
&\quad +2^{N\kappa_1} 2^{-N\lb I (\check X_{1p}, X_{1e}; Y_1, \check Y_1, U_2|\check X_{1f}, \check X_{2f}, X_{2f})-\delta(\epsilon)\rb}\\
&\quad + 2^{N\kappa_2} 2^{-N\lb I (\check X_{1p}, X_{2e};\check Y_1, Y_1, U_2|\check X_{1f}, \check X_{2f}, X_{1e}, X_{2c}) -\delta(\epsilon)\rb}  \\
&\quad +2^{N\kappa_2} 2^{-N\lb I (\check X_{1p}, \check X_{2e}; Y_1, \check Y_1, U_2|\check X_{1f}, \check X_{2c}, X_{1e}, X_{2f})-\delta(\epsilon)\rb} \\
&\quad + 2^{-N\lb I (\check X_{1p}; Y_1, \check Y_1, U_2|\check X_{1f}, \check X_{2f}, X_{1e}, X_{2e})-\delta(\epsilon)\rb} \\
&\overset{\ffff}{\leq} 4 \cdot 2^{-N\lb I(\check X_{1};\check Y_1, U_2| \check X_{1f}, \check X_{2f}) -C_1-\delta(\epsilon)\rb}\\
&\quad + 4 \cdot 2^{-N \lb I(\check X_1, X_{2e};\check Y_1, Y_1|\check X_{1e}, \check X_{2f}, X_{1e}, X_{2c})-C_1-\delta(\epsilon)\rb} \\
&\overset{\gggg}{=} 4 \cdot 2^{-N\lb I(X_{1};Y_1, U_2| X_{1f}, X_{2f})-C_1 -\delta(\epsilon)\rb}\\
&\quad + 4 \cdot 2^{-N\lb I(X_1, X_{2e};Y_1|X_{1f}, X_{2c})-C_1-\delta(\epsilon)\rb} \\ 
&\overset{\hhhh}{\leq} 8 \cdot 2^{-N\lb I(X_{1};Y_1| X_{1f}, X_{2f})-C_1-\delta(\epsilon)\rb}
\end{align*}
where
\begin{itemize}
\item (a) is since $T(1, n_1, 1, q_1, q_2, q_2 ')$ is empty set for $q_1 \notin \mcal{B}_1 (b)$, $q_2 \notin \mcal{B}_2 (b-1)$, or $q_2 ' \notin \mcal{B}_2 ((q_2, m_2)(b))$, since for random variables $(X,Y,Z) \sim p(x,y,z)$, $\intypset{X,Y,Z}$ implies $\intypset{X,Y}$,
\item (b) follows by union bound,
\item (c) follows by Claim \ref{cl:generalized_ub}, where the upper bound on the number of terms is given by Lemma \ref{lem:listsize},
\item (d) is by packing lemma, Claim~\ref{cl:generalized_pl}, and the fact that $X_{1e}^N (b-1)$ is already known at the decoder,
\item (e) is by rearranging mutual information terms using chain rule and independence (see Figure~\ref{fig:dependence}),
\item (f) follows by upper bounding four of the terms with the first expression, the remaining terms with the second expression, and using the definition of $C_1$,
\item (g) is because the distribution of variables is the same for all blocks,
\item (h) is by upper bounding the two terms with the same expression.
\end{itemize}

Once again, we use the structure of the joint distribution \eqref{eq:sjoint_dist} to show that a similar bounding can be performed for other error terms as follows.
\begin{align*}
&\Prob{T(1, 1, m_2)|E^c} \\
&\quad \leq 8 \cdot 2^{-N\lb I(X_{2f}; Y_1| X_{1e}, X_{2e})-C_1-\delta(\epsilon)\rb} \\
&\Prob{T(m_1, n_1, 1)|E^c} \\
&\quad \leq 4 \cdot 2^{-N\lb I(X_{1};Y_1, U_2| X_{1e}, X_{2f}) -C_1-\delta(\epsilon)\rb }\\
&\qquad + 4\cdot 2^{-N\lb I(X_1, X_{2e};Y_1|X_{1e}, X_{2c})-C_1-\delta(\epsilon)\rb} \\ 
&\Prob{T(1, n_1, m_2)|E^c} \\
&\quad \leq 8 \cdot 2^{-N\lb I(\check X_{1}, X_{2f};\check Y_1, Y_1, U_2| \check X_{1f}, \check X_{2e}, X_{1e}, X_{2e})-C_1-\delta(\epsilon)\rb} \\
&= 8 \cdot 2^{-N\lb I(X_{1}, X_{2f};Y_1, U_2| X_{1f}, X_{2e})-C_1-\delta(\epsilon)\rb} \\
&\leq 8 \cdot 2^{-N\lb I(X_{1}, X_{2f};Y_1| X_{1f}, X_{2e})-C_1-\delta(\epsilon)\rb} \\
&\Prob{T(m_1, n_1, m_2)|E^c} \\
&\quad \leq 8 \cdot 2^{-N\lb I(\check{X}_1, X_{2f};\check Y_1, Y_1, U_2| \check X_{1e}, \check X_{2e}, X_{1e}, X_{2e})-C_1-\delta(\epsilon)\rb} \\
&= 8 \cdot 2^{-N\lb I(X_1, X_{2f};Y_1, U_2| X_{1e}, X_{2e})-C_1-\delta(\epsilon)\rb} \\
&\leq 8 \cdot 2^{-N\lb I(X_1, X_{2f};Y_1| X_{1e}, X_{2e})-C_1-\delta(\epsilon)\rb}
\end{align*}
Using these bounds in \eqref{eq:strong_dec_error}, we see that if the conditions in the lemma are satisfied, $\Prob{\mcal{D}_{FB,s}|E^c} \to 0$ as $N \to \infty$.

\subsection{Proof of Lemmas~\ref{lem:weak_hk} and \ref{lem:strong_hk}}
Extending the notation defined in the first subsection, we define
\begin{align*}
&T(m_1, n_1, m_2, q_2) := \\
&\quad \lbp \text{\eqref{eq:weak_hk_dec} holds for the indices $(m_1, n_1, m_2, q_2)$} \rbp
\end{align*}
We will show that there exists a code such that $\Prob{\mcal{D}_{NFB}} \to 0$ \emph{exponentially}, if the given rate constraints are satisfied, which implies the claimed result. The decoding error event $\mcal{D}_{NFB}$ can be expressed as follows.
\begin{align*}
\mcal{D}_{NFB} &= \lp \bigcap_{q_2} T^c (1,1,1,q_2) \rp \cup \bigcup_{n_1 \neq 1} T(1, n_1, 1, 1) \\
&\quad \cup \bigcup_{m_2 \neq 1} T(1, 1, m_2, 1) \cup \bigcup_{\substack{m_1 \neq 1 \\ n_1 \neq 1}} T(m_1, n_1, 1, 1) \\
&\quad \cup \bigcup_{\substack{n_1 \neq 1 \\ m_2 \neq 1}} T(1, n_1, m_2,1) \cup \bigcup_{\substack{m_2 \neq 1 \\ q_2 \neq 1}} T(1, 1, m_2,q_2) \\
&\quad \cup \bigcup_{\substack{m_1 \neq 1 \\ n_1 \neq 1 \\ m_2 \neq 1}} T(m_1, n_1, m_2, 1) \cup \bigcup_{\substack{n_1 \neq 1 \\ m_2 \neq 1 \\ q_2 \neq 1}} T(1, n_1, m_2, q_2) \\
&\quad \cup \bigcup_{\substack{m_1 \neq 1, n_1 \neq 1 \\ m_2 \neq 1, q_2 \neq 1}} T(m_1, n_1, m_2,q_2)
\end{align*}
Similar to the previous proofs, choosing $r_i > I(\wtild V_j;U_i)$ ensures quantization success with high probability. Then since
\begin{align*}
\Prob{\mcal{D}_{NFB}} \leq \Prob{E_1} + \Prob{E_2} + \Prob{\mcal{D}_{NFB}|E^c},
\end{align*}
it is sufficient to show that $\Prob{\mcal{D}_{NFB}|E^c} \to 0$. Using union bound, packing lemma, Lemma \ref{lem:listsize}, and Claim \ref{cl:generalized_ub}, we can upper bound the probability of decoding error conditioned on quantization success by
\begin{align*}
&\Prob{\mcal{D}_{NFB}|E^c} \\
&\quad \leq \epsilon_N + 2^{NR_{1p}} 2^{-N \lb I(X_1;Y_1|X_{1f}, X_{2f})-\delta(\epsilon) \rb} \\
&\qquad + 2^{NR_{2c}} 2^{-N \lb I(X_{2f};Y_1|X_{1}, X_{2e})-\delta(\epsilon) \rb} \\ 
&\qquad + 2^{NR_{1}} 2^{-N \lb I(X_1;Y_1|X_{1e}, X_{2f})-\delta(\epsilon) \rb} \\
&\qquad + 2^{N(R_{1p}+R_{2c})} 2^{-N \lb I(X_1, X_{2f};Y_1|X_{1f}, X_{2e})-\delta(\epsilon) \rb} \\ 
&\qquad + 2^{N(R_{2c}+C_2 ')} 2^{-N \lb I(X_{2f};Y_1|X_{1})-\delta(\epsilon) \rb} \\
&\qquad + 2^{N(R_1+R_{2c})} 2^{-N \lb I(X_1, X_{2f};Y_1|X_{1e}, X_{2e})-\delta(\epsilon) \rb} \\
&\qquad + 2^{N(R_{1p}+R_{2c}+C_2 ')} 2^{-N \lb I(X_1, X_{2f};Y_1|X_{1f})-\delta(\epsilon) \rb} \\
&\qquad + 2^{N(R_{1}+R_{2c}+C_2 ')} 2^{-N \lb I(X_1, X_{2f};Y_1|X_{1e})-\delta(\epsilon) \rb}.
\end{align*}
where $\epsilon_N \to 0$ as $N \to \infty$. Note that the conditions in both lemmas are sufficient to ensure $\Prob{\mcal{D}_{NFB}|E^c} \to 0$ as $N \to \infty$.

%% file: AP_Evaluation.tex
In this section, we consider the set of rate conditions derived in Section~\ref{sec:achievability} for decodability (\emph{i.e.}, \eqref{eq:weak_fb_second}--\eqref{eq:weak_fb_last}, \eqref{eq:weak_hk_second}--\eqref{eq:weak_hk_last} for weak interference; \eqref{eq:strong_fb_second}--\eqref{eq:strong_fb_last}, \eqref{eq:strong_hk_second}--\eqref{eq:strong_hk_last} for strong interference), and obtain an explicit rate region for both linear deterministic and Gaussian models.

\subsection{Rate Region for Linear Deterministic Model}
Under the input distribution given by \eqref{eq:ldc_inputdist_first}--\eqref{eq:ldc_inputdist_last}, the set of rate constraints for decodability at Rx1 are evaluated as follows:
\begin{align*}
R_{1p} &\leq H\lp Y_1 | V_1, V_2\rp= \lp n_{11}-n_{21}\rp^+ \\
R_{2c} &\leq H\lp Y_1 | X_1\rp = n_{12} \\
R_{1p}+R_{2c} &\leq \min \lbp H\lp Y_1, \wtild V_2 | V_1 \rp, H\lp Y_1\rp \rbp\\
&= \min \lbp p_1 \lp n_{11}-n_{21}\rp^+ \right. \\
&\quad+ (1-p_1)\max\lbp n_{12}, \lp n_{11}-n_{21}\rp^+\rbp+ p_1n_{12}, \\
&\quad \left. \max\lp n_{11}, n_{12}\rp \vphantom{\max\lbp n_{12}, \lp n_{11}-n_{21}\rp^+\rbp}\rbp \\
R_1+R_{2c} &\leq H\lp Y_1\rp =\max\lp n_{11},n_{12}\rp
\end{align*}
for weak interference ($n_{12}\leq n_{11}$), and
\begin{align*}
R_{1p} &\leq H\lp Y_1 | V_1, V_2\rp= \lp n_{11}-n_{21}\rp^+ \\
R_{2c} &\leq H\lp Y_1 | X_1\rp = n_{12} \\
R_{1} &\leq \min \lbp H\lp Y_1, \wtild V_1 | V_2 \rp, H\lp Y_1\rp \rbp\\
&= \min \lbp p_2\lp n_{11}-n_{21}\rp^+ + (1-p_2)n_{11} \right.\\
&\quad\left.+p_2 n_{21}, \max\lp n_{11}, n_{12}\rp \vphantom{\lp n_{11}-n_{21}\rp^+ } \rbp \\
R_1+R_{2c} &\leq H\lp Y_1\rp=\max\lp n_{11},n_{12}\rp
\end{align*}
for strong interference ($n_{12}> n_{11}$). Note that the set of conditions given above can be summarized into the following five inequalities, valid for any interference regime.
\begin{align*}
R_{1p} &\leq  \lp n_{11}-n_{21}\rp^+ \\
R_1 &\leq n_{11} + p_2 \lp n_{21}-n_{11}\rp^+ \\
R_{2c} &\leq n_{12} \\
R_{1p}+R_{2c} &\leq \max \lbp n_{12}, \lp n_{11}-n_{21}\rp^+\rbp \\
&\quad+ p_1 \min \lbp n_{12}, \lp n_{11}-n_{21}\rp^+\rbp \\
R_1+R_{2c} &\leq \max\lp n_{11}, n_{12}\rp
\end{align*}
Combining these inequalities with their Rx2 counterparts, and applying Fourier-Motzkin elimination,
we arrive at the set of inequalities given in \eqref{eq:ldc_R1}--\eqref{eq:ldc_R12R2}.

\begin{figure*}[!t]
\begin{align}
R_{ip} &< \msf{A}_i := \log \lp 3 + \frac{\SNR_i}{1+\INR_j} \rp - \log 3 - C_i \label{eq:g_rc_Ai} \\
R_{jc} &< \msf{B}_i :=  \log \lp 2 + \INR_i \rp - \log 3 - C_i \label{eq:g_rc_Bi}\\
R_{i} &< \msf{C}_i := \log \lp 3 + \SNR_i+\INR_i \rp  - \log 3 - C_i \label{eq:g_rc_Ci}\\
R_i &< \msf{D}_i :=\log \lp 3 + \SNR_i \rp + \ind{\SNR_i\leq\INR_i} p_j\lb \log \lp 1+\frac{\INR_j}{ 3 +\SNR_i } \rp -\log \frac{5}{3}\rb - \log 3 - C_i\label{eq:g_rc_Di} \\
R_{ip}+R_{jc} &< \msf{E}_i := \log \lp 2 + \SNR_i + \INR_i + \frac{\SNR_i}{1+\INR_j} \rp -\log3 - C_i \label{eq:g_rc_Ei}\\
R_{ip}+R_{jc} &< \msf{F}_i := \log \lp 2 + \INR_i + \frac{\SNR_i}{1+\INR_j} \rp + \ind{\SNR_i\geq\INR_i} p_i \lb \log \lp \frac{\lp 2+\INR_i\rp\lp 3+\frac{\SNR_i}{1+\INR_j}\rp}{2+\frac{\SNR_i}{1+\INR_j}+\INR_i} \rp - \log 6 \rb\notag\\
&\quad  -\log 3 - C_i -\ind{\SNR_i\geq\INR_i}C_i\label{eq:g_rc_Fi} \\
R_{i}+R_{jc} &< \msf{G}_i := \log \lp 2 + \SNR_i+\INR_i\rp -\log 3-C_i - \kappa_j \label{eq:g_rc_Gi}\\
&\quad C_i := p_i+2p_j, \; \kappa_j = p_j \notag
\end{align}
\hrulefill
\end{figure*}

\subsection{Rate Region for Gaussian Model}
We consider the input distributions \eqref{eq:g_inputdist_first}--\eqref{eq:g_inputdist_last}, and the set of input-output relationships given by
\begin{align*}
Y_i &= h_{ii}X_i + h_{ij}X_j + Z_i \\
U_i &= S_i \lp h_{ij}X_j + Z_i \rp + Q_i
\end{align*}
for $\ijj$, where $Q_i \sim \cgauss{0}{D_i}$. Choosing $D_1=D_2=\frac{3}{2}$, and using standard techniques, it is straightforward to evaluate the rate inequalities derived in Section~\ref{sec:achievability}, and show that the set of rate triples $\lp R_{1p}, R_{1c}, R_{2c}\rp$ defined by \eqref{eq:g_rc_Ai}--\eqref{eq:g_rc_Gi}, for $(i,j)=(1,2)$ are contained in the set defined by \eqref{eq:weak_fb_second}--\eqref{eq:weak_fb_last}, \eqref{eq:weak_hk_second}--\eqref{eq:weak_hk_last} for weak interference, and \eqref{eq:strong_fb_second}--\eqref{eq:strong_fb_last}, \eqref{eq:strong_hk_second}--\eqref{eq:strong_hk_last} for strong interference. In \eqref{eq:g_rc_Ai}--\eqref{eq:g_rc_Gi}, we used indicator functions to unify the constraints for weak and strong interference. 

In order to find the set of achievable $\lp R_1,R_2\rp$ points, we first note that $\msf{E}_i\geq\msf{G}_i$ and $\msf{C}_i\geq\msf{G}_i$, and hence the bounds $\msf{E}_i$ and $\msf{C}_i$ are redundant. Considering the remaining bounds for $\ijj$, noting that $\msf{F}_i \leq \msf{A}_i+\msf{B}_i$, and applying Fourier-Motzkin elimination, we find that the set of $\lp R_1,R_2\rp$ points that satisfy the following are achievable.
\begin{align}
R_i &< \min\lbp \msf{A}_i+\msf{B}_j, \msf{D}_i \rbp \label{eq:ib_Ri}\\
R_i+R_j &< \min\lbp \msf{A}_i+\msf{G}_j, \msf{F}_i+\msf{F}_j \rbp \label{eq:ib_RiRj}\\
2R_i + R_j &< \msf{A}_i + \msf{F}_j + \msf{G}_i \label{eq:ib_2RiRj}
\end{align}
for $\ijj$.

%% file: AP_OuterBoundLDC.tex
In this section, we prove the outer bounds for the linear deterministic channel, based on the ideas presented in Section~\ref{sec:converse}. We first prove four claims that will be useful in the main proof.

\subsection{Proof of the Bound \eqref{eq:ob_ldc_Ri}}
By symmetry, we only focus on the bound on $R_1$. By Fano's inequality,
\begin{align*}
&N\lp R_1 - \epsilon_N \rp \leq I\lp W_1; Y_1^N \ul{S}^N\rp = I\lp W_1; Y_1^N, \wtild V_1^N, W_2, \ul{S}^N\rp \\
&\quad = I\lp W_1; Y_1^N, \wtild V_1^N| W_2, \ul{S}^N\rp \overset{\aaaa}{=} H\lp Y_1^N, \wtild V_1^N| W_2, \ul{S}^N\rp \\
&\quad = H\lp Y_1^N| \wtild V_1^N, W_2, \ul{S}^N\rp + H\lp \wtild V_1^N|W_2, \ul{S}^N\rp \\
&\quad \overset{\bbbb}{=} H\lp Y_1^N| \wtild V_1^N, W_2, X_2^N, \ul{S}^N\rp + H\lp\wtild V_1^N|W_2, \ul{S}^N\rp \\
&\quad \leq H\lp Y_1^N| \wtild V_1^N, X_2^N, \ul{S}^N\rp + H\lp\wtild V_1^N|\ul{S}^N\rp \\
&\quad \overset{\cccc}{\leq} n_{11} + p_2\lp n_{21}-n_{11}\rp^+
\end{align*}
where (a) follows by the fact that channel is deterministic and hence all variables are completely determined given $\lp W_1, W_2, \ul{S}^N \rp$; (b) follows by Claim~\ref{cl:ldc_claim1}, and (c) follows by Claim~\ref{cl:ldc_claim5}.

\subsection{Proof of the Bound \eqref{eq:ob_ldc_R1R2}}
By Fano's inequality,
\begin{align*}
&N\lp R_1+R_2-\epsilon_N\rp \leq I(W_1;Y_1^N, \ul{S}^N) + I(W_2;Y_2^N, \ul{S}^N)  \\
&\quad = I(W_1;Y_1^N|  \ul{S}^N) + I(W_2;Y_2^N| \ul{S}^N)  \\
&\quad \leq I(W_1;Y_1^N, V_1^N, \wtild V_2^N| \ul{S}^N) + I(W_2;Y_2^N, V_2^N, \wtild V_1^N| \ul{S}^N)  \\
&\quad = H\lp Y_1^N, V_1^N, \wtild V_2^N| \ul{S}^N \rp + H\lp Y_2^N, V_2^N, \wtild V_1^N| \ul{S}^N \rp \\
&\qquad - H\lp Y_1^N, V_1^N, \wtild V_2^N| W_1, \ul{S}^N \rp \\
&\qquad - H\lp Y_2^N, V_2^N, \wtild V_1^N| W_2, \ul{S}^N \rp\\
& \overset{\aaaa}{=} H\lp Y_1^N | V_1^N, \wtild V_2^N, \ul{S}^N \rp + H\lp Y_2^N| V_2^N, \wtild V_1^N, \ul{S}^N \rp \\
&\quad + H\lp V_1^N, \wtild V_2^N | \ul{S}^N \rp + H\lp V_2^N, \wtild V_1^N | \ul{S}^N \rp \\
&\quad - H\lp V_2^N, \wtild V_1^N| W_1, \ul{S}^N \rp - H\lp V_1^N, \wtild V_2^N| W_2, \ul{S}^N \rp\\
&= H\lp Y_1^N | V_1^N, \wtild V_2^N, \ul{S}^N \rp + H\lp Y_2^N| V_2^N, \wtild V_1^N, \ul{S}^N \rp \\
&\quad + I\lp W_2; V_1^N, \wtild V_2^N | \ul{S}^N \rp + I\lp W_1; V_2^N, \wtild V_1^N | \ul{S}^N \rp \\
&\overset{\bbbb}{\leq} N\max\lbp n_{12}, \lp n_{11}-n_{21}\rp^+ \rbp \\
&\quad + N\max\lbp n_{21}, \lp n_{22}-n_{12}\rp^+ \rbp \\
&\quad + Np_1\min\lbp n_{12}, \lp n_{11}-n_{21}\rp^+ \rbp \\
&\quad + Np_2\min\lbp n_{21}, \lp n_{22}-n_{12}\rp^+ \rbp
\end{align*}
where (a) follows by Claim~\ref{cl:ldc_claim2}, and (b) follows by Claims \ref{cl:ldc_claim4} and \ref{cl:ldc_claim5}.

\subsection{Proof of the Bounds \eqref{eq:ob_ldc_2R1R2} and \eqref{eq:ob_ldc_R12R2}}
By symmetry, it is sufficient to prove \eqref{eq:ob_ldc_2R1R2}. To prove this bound, we consider two copies of Rx1, where one of the copies are enhanced as decribed in Section~\ref{sec:converse}, while the other one is provided with the output of the original channel. The only copy of Rx2 receives the enhanced channel output as well. We would like to prove a sum rate bound for this three-receiver channel. By Fano's inequality,
\begin{align*}
&N\lp 2R_1+R_2 - \epsilon_N \rp \\
& \leq I\lp W_1; Y_1^N, \ul{S}^N\rp + I\lp W_2; Y_2^N, \ul{S}^N\rp + I\lp W_1; Y_1^N, \ul{S}^N\rp\\
& = I\lp W_1; Y_1^N| \ul{S}^N\rp + I\lp W_2; Y_2^N| \ul{S}^N\rp + I\lp W_1; Y_1^N| \ul{S}^N\rp\\
& \leq I\lp W_1; Y_1^N| \ul{S}^N\rp + I\lp W_2; Y_2^N, V_2^N, \wtild V_1^N | \ul{S}^N\rp \\
&\quad + I\lp W_1; Y_1^N, V_1^N| \ul{S}^N, W_2\rp\\
&\overset{\aaaa}{=} H\lp Y_1^N | \ul{S}^N\rp - H\lp V_2^N, \wtild V_1^N |\ul{S}^N, W_1\rp \\
&\quad + H\lp Y_2^N, V_2^N, \wtild V_1^N| \ul{S}^N\rp - H\lp Y_2^N, V_2^N, \wtild V_1^N | \ul{S}^N, W_2 \rp \\
&\quad + H\lp Y_1^N, V_1^N| \ul{S}^N, W_2\rp \\
&\overset{\bbbb}{=} H\lp Y_1^N | \ul{S}^N\rp - H\lp V_2^N, \wtild V_1^N |\ul{S}^N, W_1 \rp \\
&\quad + H\lp V_2^N, \wtild V_1^N | \ul{S}^N \rp + H\lp Y_2^N | V_2^N, \wtild V_1^N\rp \\
&\quad - H\lp V_1^N | \ul{S}^N, W_2 \rp + H\lp Y_1^N, V_1^N | \ul{S}^N, W_2 \rp \\
&= H\lp Y_1^N | \ul{S}^N\rp + I\lp W_1; V_2^N, \wtild V_1^N |\ul{S}^N \rp \\
&\quad H\lp Y_2^N | V_2^N, \wtild V_1^N\rp + H\lp Y_1^N | \ul{S}^N, W_2, V_1^N \rp \\
&\overset{\cccc}{\leq} \max\lp n_{11}, n_{12}\rp + \max\lbp n_{21}, \lp n_{22}-n_{12}\rp^+ \rbp \\
&\quad + \lp n_{11}-n_{21}\rp^+ + p_2 \min\lbp n_{21}, \lp n_{22}-n_{12}\rp^+ \rbp
\end{align*}
where (a) follows by Claim ~\ref{cl:ldc_claim2} , (b) follows by Claim~\ref{cl:ldc_claim3}, (c) follows by Claims~\ref{cl:ldc_claim4},~\ref{cl:ldc_claim5} and~\ref{cl:ldc_claim6}.

\subsection{Claims}
\begin{claim} \label{cl:ldc_claim1}
For $\ijj$,
\begin{align*}
X_{i,t} \eqFunc \lp W_i, \wtild{V}_j^{t-1}, \ul{S}^{t-1}\rp \eqFunc \lp W_i, V_j^{t-1}, \ul{S}^{t-1}\rp
\end{align*}
\end{claim}
\begin{proof}
We focus on the case $(i,j)=(1,2)$ without loss of generality. Note that
\begin{align*}
X_{1,1} \eqFunc W_1
\end{align*}
and by the definition of the channel,
\begin{align*}
X_{1,t} &\eqFunc \lp W_1, \wtild Y_1^{t-1}, \ul{S}^t \rp \\
&\overset{\aaaa}{\eqFunc} \lp W_1, \wtild V_2^{t-1}, X_1^{t-1} \ul{S}^t \rp,
\end{align*}
hence the result follows by induction on $t$. (a) follows because $\wtild Y_1^{t-1} = S_1^{t-1}\mb{H}_{11}X_{1}^{t-1} + \wtild V_2^{t-1}$.
\end{proof}

\begin{claim}\label{cl:ldc_claim2}
For $\ijj$,
\begin{align*}
&H\lp Y_i^N|W_i,\ul{S}^N\rp = H\lp V_j^N,\wtild{V}_i^N|W_i,\ul{S}^N\rp.
\end{align*}
\end{claim}
\begin{proof}
Let us focus on the case $(i,j)=(1,2)$. 
\begin{align*}
&H\lp Y_1^N|W_1,\ul{S}^N\rp 
=\sum_{t=1}^N H\lp Y_{1,t} | W_1, \ul{S}^N,Y_1^{t-1}\rp\\
&\overset{\aaaa}{=}\sum_{t=1}^N H\lp Y_{1,t} | W_1, \ul{S}^N,Y_1^{t-1}, X_1^t\rp\\ 
&=\sum_{t=1}^N H\lp V_{2,t} | W_1, \ul{S}^N,V_2^{t-1}, X_1^t\rp\\ 
&\overset{\bbbb}{=}\sum_{t=1}^N H\lp V_{2,t} | W_1, \ul{S}^N,V_2^{t-1}\rp\\ 
&= H\lp V_2^N,\wtild{V}_1^N|W_1,\ul{S}^N\rp,
\end{align*}
where (a) is by definition, (b) is due to Claim~\ref{cl:ldc_claim1}.
The other holds similarly.
\end{proof}

\begin{claim}\label{cl:ldc_claim3}
For $\ijj$,
\begin{align*}
H\lp Y_i^N, V_i^N, \wtild V_j^N |W_i,\ul{S}^N\rp &= H\lp V_j^N,\wtild{V}_i^N|W_i,\ul{S}^N\rp \\
&= H\lp V_j^N|W_i,\ul{S}^N\rp
\end{align*}
\end{claim}
\begin{proof}
Let us focus on the case $(i,j)=(1,2)$. 
\begin{align*}
&H\lp Y_1^N, V_1^N, \wtild V_2^N |W_1,\ul{S}^N\rp  \\
&=\sum_{t=1}^N H\lp Y_{1,t}, V_{1,t}, \wtild V_{2,t} | W_1, \ul{S}^N, Y_1^{t-1}, V_1^{t-1}, \wtild V_2^{t-1}\rp\\
&\overset{\aaaa}{=}\sum_{t=1}^N H\lp Y_{1,t}, V_{1,t}, \wtild V_{2,t} | W_1, \ul{S}^N, Y_1^{t-1}, V_1^{t-1}, \wtild V_2^{t-1}, X_1^t \rp\\
&\overset{\bbbb}{=} \sum_{t=1}^N H\lp V_{2,t}, \wtild V_{1,t} | W_1, \ul{S}^N, V_2^{t-1}, \wtild V_1^{t-1}, X_1^t \rp\\
&\overset{\cccc}{=} \sum_{t=1}^N H\lp V_{2,t} | W_1, \ul{S}^N, V_2^{t-1}, X_1^t \rp\\
\end{align*}
where (a) follows by the fact that $X_{1,t} \eqFunc \lp W_1, \wtild Y_1^N, \ul{S}^N\rp$, (b) follows by subtracting $X_{1,t}$ from $Y_{1,t}$ and because $V_{1,t} \eqFunc X_{1,t}$ and $\wtild V_{2,t} \eqFunc \lp \ul{S}_t, V_{2,t}\rp$. Similarly, (c) follows since $\wtild V_{1,t} \eqFunc \lp X_{1,t}, \ul{S}_t \rp$. 

Now, the two equalities in the claim can be easily obtained from (b) and (c) respectively, by removing $X_{1}^t$ from the conditioning by virtue of Claim~\ref{cl:ldc_claim1}, and using chain rule.
\end{proof}

\begin{claim}\label{cl:ldc_claim4}
For $\ijj$,
\begin{align*}
&I\lp W_i;V_j^N,\wtild{V}_i^N|\ul{S}^N\rp \le Np_jn_{ji}.
\end{align*}
\end{claim}
\begin{proof}
Let us focus on the case $(i,j)=(1,2)$. 
\begin{align*}
&I\lp W_1;V_2^N,\wtild{V}_1^N|\ul{S}^N\rp\\
&\overset{\aaaa}\le I\lp W_1;W_2,\wtild{V}_1^N|\ul{S}^N\rp
= I\lp W_1;\wtild{V}_1^N|\ul{S}^N,W_2\rp\\
&= H\lp\wtild{V}_1^N|\ul{S}^N,W_2\rp
\le H\lp\wtild{V}_1^N|S_2^N\rp\\
&= \mbb{E}_{S_2^N}\lb H\lp (s_2V_1)^N \rp \big| S_2^N = s_2^N\rb\\
&\le \mbb{E}_{S_2^N}\lb \sum_{t=1}^N H\lp s_{2,t}V_{1,t} \rp \Bigg| S_2^N = s_2^N\rb\\
&\le \mbb{E}_{S_2^N}\lb N_1\lp s_2^N\rp n_{21} \Big| S_2^N = s_2^N\rb\\
&= Np_2n_{21}.
\end{align*}
Here $N_1(\cdot)$ denotes the number of $1$'s in the sequence. (a) follows because $V_2^N \eqFunc \lp W_2, \wtild{V}_1^{N}, \ul{S}^{N}\rp$.
\end{proof}

\begin{claim}\label{cl:ldc_claim5}
For $\ijj$,
\begin{align*}
&N^{-1}H\lp Y_i^N| V_i^N, \wtild{V}_j^N, \ul{S}^N\rp\\
&\le p_i (n_{ii}-n_{ji})^+ + (1-p_i)\max\lbp n_{ij}, (n_{ii}-n_{ji})^+\rbp,\\
&N^{-1}H\lp Y_i^N| V_j^N, \wtild{V}_i^N, \ul{S}^N\rp\\
&\le p_j(n_{ii}-n_{ji})^+ + (1-p_j)n_{ii},\\
\end{align*}
\end{claim}
\begin{proof}
Let us focus on the case $(i,j)=(1,2)$. 
\begin{align*}
&H\lp Y_1^N | V_1^N, \wtild{V}_2^N, \ul{S}^N\rp 
\le H\lp Y_1^N | V_1^N, \wtild{V}_2^N, S_1^N\rp\\
&= \mbb{E}_{S_1^N}\lb H\lp Y_1^N | V_1^N, (s_1V_2)^N \rp \big| S_1^N = s_1^N\rb\\
&\le \mbb{E}_{S_1^N}\lb \sum_{t=1}^N H\lp Y_{1,t} | V_{1,t}, s_{1,t}V_{2,t}\rp \Bigg| S_1^N = s_1^N\rb\\
&\le \mbb{E}_{S_1^N}\lb \left. \begin{array}{l}N_1\lp s_1^N\rp (n_{11}-n_{21})^+\\ + N_0\lp s_1^N\rp \max\lbp n_{12}, (n_{11}-n_{21})^+\rbp \end{array}\right | S_1^N = s_1^N\rb\\
&= Np_1 (n_{11}-n_{21})^+ + N(1-p_1)\max\lbp n_{12}, (n_{11}-n_{21})^+\rbp
\end{align*}
Here $N_1(\cdot)$ and $N_0(\cdot)$ denote the number of $1$'s and $0$'s respectively in the sequence.

For the second inequality,
\begin{align*}
&H\lp Y_1^N | V_2^N, \wtild{V}_1^N, \ul{S}^N\rp 
\le H\lp Y_1^N | V_2^N, \wtild{V}_1^N, S_2^N\rp\\
&= \mbb{E}_{S_2^N}\lb H\lp Y_1^N | V_2^N, (s_2V_1)^N \rp \big| S_2^N = s_2^N\rb\\
&\le \mbb{E}_{S_2^N}\lb \sum_{t=1}^N H\lp Y_{1,t} | V_{2,t}, s_{2,t}V_{1,t}\rp \Bigg| S_2^N = s_2^N\rb\\
&\le \mbb{E}_{S_2^N}\lb N_1\lp s_2^N\rp (n_{11}-n_{21})^+ + N_0\lp s_2^N\rp n_{11} \Big| S_2^N = s_2^N\rb\\
&= Np_2 (n_{11}-n_{21})^+ + N(1-p_2)n_{11}.
\end{align*}
The case $(i,j)=(2,1)$ follows similarly.
\end{proof}

\begin{claim}\label{cl:ldc_claim6}
For $\ijj$,
\begin{align*}
H\lp Y_i^N | \ul{S}^N, W_j, V_i^N \rp \leq N\lp n_{ii}-n_{ji}\rp^+
\end{align*}
\end{claim}
\begin{proof}
We focus on $(i,j)=(1,2)$ without loss of generality.
\begin{align*}
&H\lp Y_i^N | \ul{S}^N, W_j, V_i^N \rp \overset{\aaaa}{=} H\lp Y_i^N | \ul{S}^N, W_j, V_i^N, V_j^N \rp \\
&\quad \leq H\lp Y_i^N | \ul{S}^N,V_i^N, V_j^N \rp \leq N\lp n_{11}-n_{21}\rp^+
\end{align*}
where (a) follows because $V_j^N \eqFunc X_j^N \eqFunc \lp W_j, V_i^N, \ul{S}^N\rp$ by Claim~\ref{cl:ldc_claim1}.
\end{proof}

%% file: AP_OuterBoundGaussian.tex
In this section, we prove an outer bound region for the enhanced channel defined in Section~\ref{sec:converse}.

\subsection{Notation}
We define
\begin{align*}
\breve V_i = \left\{ 
\begin{array}{ll}
\bar V_i, & \text{ if $S_j = 0$}\\
Y_{ji}, & \text{ if $S_j = 1$}
\end{array}
\right.
\end{align*}
for $\ijj$, where $\bar V_i = Y_{ji} + Z_{jj} = h_{ji}X_i+\bar Z_j$, and
\begin{align*}
M_i &= \left| \lbp t: S_{i,t}=1 \rbp \right|, \\
L_i &= N-M_i.
\end{align*}
For any random vector $E^N$, we define
\begin{align*}
E^{(t)} &= \lbp E_{t '} \rbp_{t': S_{i,t'} = 1, t'\leq t} \\
E^{[t]} &= \lbp E_{t '} \rbp_{t': S_{i,t'} = 0, t'\leq t} \\
E_{i,(t)}&= \left\{ 
\begin{array}{ll}
\emptyset, & \text{ if $S_i = 0$}\\
E_{i,t}, & \text{ if $S_i = 1$}
\end{array}
\right. \\
E_{i,[t]}&= \left\{ 
\begin{array}{ll}
E_{i,t}, & \text{ if $S_i = 0$}\\
\emptyset, & \text{ if $S_i = 1$}
\end{array}
\right.
\end{align*}
for $i=1$ or 2. Note that in vector form, this notation omits any reference to user index $i$ for the sake of brevity. That is, although it is not clear whether $E^{(t)}$ is defined with respect to $S_1$ or $S_2$, in the proof this will be clear from the context. For instance, $Y_1^{(t)}$ and $V_2^{(t)}$ are defined with respect to $S_1$, since these variables refer to signals that pass through the feedback channel controlled by $S_1$. The partial average power for Tx$i$, $P_i^{(jk)}$, is a random variable defined as
\begin{align*}
P_i^{(j0)} &= \frac{1}{L_i} \sum_{t:S_{j,t=0}} P_{i,t} \\
P_i^{(j1)} &= \frac{1}{M_i} \sum_{t:S_{j,t=1}} P_{i,t}
\end{align*}
for $j=1,2$, where $P_{i,t}$ ie the power used by Tx$i$ at time slot $t$.

Finally, we define $h_{S}\lp\cdot\rp:= h\lp\cdot|\ul{S}^N=S^N\rp$ for convenience, where $h(\cdot)$ denotes differential entropy, and $S^N$ is a particular realization of $\ul{S}^N$. Similarly, we define $I_S\lp \cdot;\cdot\rp:=I_S\lp \cdot;\cdot|\ul{S}^N=S^N\rp$.

\subsection{Proof of the Bound \eqref{eq:ob_g_Ri}}
We focus on the case $(i,j)=(1,2)$. By Fano's inequality,
\begin{align*}
&N(R_1 - \epsilon_N) \leq I(W_1; Y_1^N, \ul{S}^N) \\
&\quad\leq I(W_1; Y_1^N, \wtild V_1^N, W_2, \ul{S}^N) \\
&\quad\leq I(W_1; Y_1^N, \wtild V_1^N, \ul{S}^N | W_2) \\
&\quad\leq \sum_{t=1}^N I(W_1; Y_{1,t}, \wtild V_{1,t}, \ul{S}_t | W_2, Y_1^{t-1}, \wtild V_1^{t-1}, \ul{S}^{t-1}) \\
&\quad= \sum_{t=1}^N I(W_1; Y_{1,t}, \wtild V_{1,t}| W_2, Y_1^{t-1}, \wtild V_1^{t-1}, \ul{S}^{t}) \\
&\quad\overset{\aaaa}{=} \sum_{t=1}^N I(W_1; Y_{1,t}, \wtild V_{1,t}| W_2, Y_1^{t-1}, \wtild V_1^{t-1}, \ul{S}^{t}, X_{2,t}) \\
&\quad= \sum_{t=1}^N I(W_1; Y_{1,t}| W_2, Y_1^{t-1}, \wtild V_1^{t}, \ul{S}^{t}, X_{2,t}) \\
&\qquad +I(W_1; \wtild V_{1,t}| W_2, Y_1^{t-1}, \wtild V_1^{t-1}, \ul{S}^{t}, X_{2,t})  \\
&\quad= \sum_{t=1}^N h(Y_{1,t}| W_2, Y_1^{t-1}, \wtild V_1^{t}, \ul{S}^{t}, X_{2,t}) \\
&\qquad - h(Y_{1,t}| W_2, Y_1^{t-1}, \wtild V_1^{t}, \ul{S}^{t}, X_{2,t}, W_1) \\
&\qquad +I(W_1; \wtild V_{1,t}| W_2, Y_1^{t-1}, \wtild V_1^{t-1}, \ul{S}^{t}, X_{2,t})  \\
&\overset{\bbbb}{=} \sum_{t=1}^N h(Y_{1,t}| W_2, Y_1^{t-1}, \wtild V_1^{t}, \ul{S}^{t}, X_{2,t}) \\
&\qquad - h(Y_{1,t}| W_2, Y_1^{t-1}, \wtild V_1^{t}, \ul{S}^{t}, X_{2,t}, W_1, X_{1,t}) \\
&\qquad +\sum_{t=1}^N I(W_1; \wtild V_{1,t}| W_2, Y_1^{t-1}, \wtild V_1^{t-1}, \ul{S}^{t}, X_{2,t})  \\
&\quad = \sum_{t=1}^N h(Y_{1,t}| W_2, Y_1^{t-1}, \wtild V_1^{t}, \ul{S}^{t}, X_{2,t}) \\
&\qquad - h(Z_{1,t}| W_2, Y_1^{t-1}, \wtild V_1^{t}, \ul{S}^{t}, X_{2,t}, W_1, X_{1,t}) \\
&\qquad + I(W_1; \wtild V_{1,t}| W_2, Y_1^{t-1}, \wtild V_1^{t-1}, \ul{S}^{t}, X_{2,t})  \\
&\quad\overset{\cccc}{=} \sum_{t=1}^N h(Y_{1,t}| W_2, Y_1^{t-1}, \wtild V_1^{t}, \ul{S}^{t}, X_{2,t}) - h(Z_{1,t})\\
&\qquad + I(W_1; \wtild V_{1,t}| W_2, Y_1^{t-1}, \wtild V_1^{t-1}, \ul{S}^{t}, X_{2,t})  \\
&\quad \leq \sum_{t=1}^N h(Y_{1,t}| \wtild V_{1,t}, S_{2,t}, X_{2,t}) - h(Z_{1,t}) \\
&\qquad + I(W_1; \wtild V_{1,t}| W_2, Y_1^{t-1}, \wtild V_1^{t-1}, \ul{S}^{t}, X_{2,t})  \\ 
&\quad \leq \sum_{t=1}^N h(Y_{1,t}| \wtild V_{1,t}, S_{2,t}, X_{2,t}) - h(Z_{1,t}) \\
&\qquad + I(W_1, X_{1,t}; \wtild V_{1,t}| W_2, Y_1^{t-1}, \wtild V_1^{t-1}, \ul{S}^{t}, X_{2,t})  \\
&\quad = \sum_{t=1}^N h(Y_{1,t}| \wtild V_{1,t}, S_{2,t}, X_{2,t}) - h(Z_{1,t}) \\
&\qquad + I(X_{1,t}; \wtild V_{1,t}| W_2, Y_1^{t-1}, \wtild V_1^{t-1}, \ul{S}^{t}, X_{2,t}) \\
&\qquad + I(W_1; \wtild V_{1,t}| W_2, Y_1^{t-1}, \wtild V_1^{t-1}, \ul{S}^{t}, X_{2,t}, X_{1,t})  \\
&\overset{\dddd}{=} \sum_{t=1}^N h(Y_{1,t}| \wtild V_{1,t}, S_{2,t}, X_{2,t}) - h(Z_{1,t}) \\
&\qquad + I(X_{1,t}; \wtild V_{1,t}| W_2, Y_1^{t-1}, \wtild V_1^{t-1}, \ul{S}^{t}, X_{2,t})  \\
&= \quad \sum_{t=1}^N h(Y_{1,t}| \wtild V_{1,t}, S_{2,t}, X_{2,t}) - h(Z_{1,t}) \\
&\qquad + h( \wtild V_{1,t}| W_2, Y_1^{t-1}, \wtild V_1^{t-1}, \ul{S}^{t}, X_{2,t}) \\
&\qquad - h( \wtild V_{1,t}| W_2, Y_1^{t-1}, \wtild V_1^{t-1}, \ul{S}^{t}, X_{2,t}, X_{1,t})  \\
&\quad \leq \sum_{t=1}^N h(Y_{1,t}| \wtild V_{1,t}, S_{2,t}, X_{2,t}) - h(Z_{1,t}) \\
&\qquad + h( \wtild V_{1,t}|S_{2,t}) - h( \wtild V_{1,t}| W_2, Y_1^{t-1}, \wtild V_1^{t-1}, \ul{S}^{t}, X_{2,t}, X_{1,t})  \\
&\quad \overset{\eeee}{=} \sum_{t=1}^N h(Y_{1,t}| \wtild V_{1,t}, S_{2,t}, X_{2,t}) - h(Z_{1,t}) \\
&\qquad + h( \wtild V_{1,t}|\ul{S}_{t}) - h( \wtild V_{1,t}| S_{2,t}, X_{1,t})  \\
&\quad= \sum_{t=1}^N h(Y_{1,t}| \wtild V_{1,t}, S_{2,t}, X_{2,t}) - h(Z_{1,t}) \\
&\qquad + I(X_{1,t};\wtild V_{1,t}|S_{2,t}) \\
&\quad \overset{\ffff}{=} p_2 \log \lp 1+\frac{\SNR_1}{1+\INR_2}  \rp + (1-p_2)\log \lp 1+\SNR_1 \rp \\
&\qquad + p_2 \log \lp 1 + \INR_2 \rp \\
&\quad =  \log \lp 1 + \SNR_1 \rp + p_2 \log \lp1 + \frac{\INR_2}{1+\SNR_1} \rp
\end{align*}
where
\begin{itemize}
\item (a) is due to Lemma~\ref{lem:xfunc}
\item (b) is because $X_{1,t} \eqFunc \lp \ul{S}^{t-1}, W_1, \wtild Y_1^{t-1} \rp \eqFunc \lp \ul{S}^{t-1}, W_1, Y_1^{t-1} \rp$,
\item (c) is because $Z_{1,t}$ is independent from all past signals and messages,
\item (d) is because $W_1 - X_{1,t} - \wtild V_{1,t}$ is a Markov chain, hence conditioned on $X_{1,t}$, $\wtild V_{1,t}$ is independent from $W_1$ and all the other past signals,
\item (e) is because given $\lp S_{2,t}, X_{1,t} \rp$, $\wtild V_{1,t}$ is independent from all the other variables in the conditioning,
\item (f) is due to Lemma~\ref{lem:hYtV}
\end{itemize}

\subsection{Proof of Bound \eqref{eq:ob_g_R1R2}}
In this section, we exclusively focus on the enhanced channel defined in Section~\ref{sec:converse}. By Fano's inequality.
\begin{align}
&N(R_1+R_2 - \epsilon_N) \notag\\
&\quad\leq I(W_1; \breve Y_1^N, \ul{S}^N) + I(W_2; \breve Y_2^N, \ul{S}^N) \notag\\
&\quad= I(W_1; \breve Y_1^N | \ul{S}^N) + I(W_2; \breve Y_2^N | \ul{S}^N) \notag\\
&\quad\leq I(W_1; \breve Y_1^N, \breve V_1^N | \ul{S}^N) + I(W_2; \breve Y_2^N, \breve V_2^N | \ul{S}^N) \notag\\
&\quad= h(\breve Y_1^N|\ul{S}^N, \breve V_1^N) + h(\breve Y_2^N|\ul{S}^N, \breve V_2^N) \label{eq:one}\\
&\qquad + h(\breve V_1^N|\ul{S}^N) + h(\breve V_2^N|\ul{S}^N) \label{eq:two}\\
&\qquad - h(\breve Y_1^N, \breve V_1^N | \ul{S}^N, W_1) - h(\breve Y_2^N, \breve V_2^N | \ul{S}^N, W_2) \label{eq:three}
\end{align}
Let us take one term from \eqref{eq:one}. 
\begin{align*}
&h(\breve Y_1^N|\ul{S}^N, \breve V_1^N) = \Es{ \hs{\breve Y_1^N|\breve V_1^N} } \\
&\quad \overset{\aaaa}{=} \Es{ \hs{\bar Y_1^{L_1}, Y_{11}^{M_1}, Y_{12}^{M_1}|\breve V_1^N} } \\
&\quad \leq \Es{ \hs{\bar Y_1^{L_1}|\breve V_1^N} }+ \Es{ \hs{Y_{11}^{M_1}|\breve V_1^N} } \\
&\qquad + \Es{ \hs{ Y_{12}^{M_1}|\breve V_1^N} }
\end{align*}
where (a) follows by (with a slight abuse of notation) decomposing $\breve Y_1^N$ into $\lp Y_{11}^{M_1}, Y_{12}^{M_1} \rp$ for time slots where $S_{1,t}=1$, and into $\bar Y_1^{L_1}$ for time slots for which $S_{1,t}=0$.

The other term in \eqref{eq:one} can be bounded similarly. Let us take one term from \eqref{eq:three}.
\begin{align*}
&- h(\breve Y_1^N, \breve V_1^N | \ul{S}^N, W_1) \\
&= -\Es{ \hs{\breve Y_1^N, \breve V_1^N |W_1}} \\
&\overset{\bbbb}{=} -\Es{ \hs{\breve V_2^N, Z_{11}^{M_1}, Z_2^{L_2}, Z_{21}^{M_2} |W_1}} \\
&\overset{\cccc}{=} -\sum_{t=1}^N \mbb{E}_{S^N}\lb h \lp \breve V_{2,t}, Z_{11,(t)}, Z_{2,[t]}, Z_{21,(t)}\right.\right.\\
&\qquad\qquad \left.\left| W_1, \breve V_{2}^{t-1}, Z_{11}^{(t-1)}, Z_{2}^{[t-1]}, Z_{21}^{(t-1)}\rp \rb \\
&\overset{\dddd}{=} -\mbb{E} \lb \sum_{t=1}^N \hs{\breve V_{2,t}|W_1, \breve V_2^{t-1}, Z_{11}^{(t-1)}, Z_2^{[t-1]} Z_{21}^{(t-1)}} \right. \\
&\quad + \sum_{t: S_{1,t}=1} \hs{Z_{11,t}|W_1, \breve V_2^{t-1}, Z_{11}^{(t-1)}, Z_2^{[t-1]} Z_{21}^{(t-1)}}  \\
&\quad + \sum_{t: S_{2,t}=0} \hs{Z_{2,t}|W_1, \breve V_2^{t-1}, Z_{11}^{(t-1)}, Z_2^{[t-1]} Z_{21}^{(t-1)}}  \\
&\quad +\left. \sum_{t: S_{2,t}=1} \hs{Z_{21,t}|W_1, \breve V_2^{t-1}, Z_{11}^{(t-1)}, Z_2^{[t-1]} Z_{21}^{(t-1)}} \rb \\
&\overset{\eeee}{=} -\Es{\sum_{t=1}^N \hs{\breve V_{2,t}|W_1, \breve V_2^{t-1}, Z_{11}^{(t-1)}, Z_2^{[t-1]} Z_{21}^{(t-1)}} } \\
&\quad -\Es{ L_2 h(Z_2) + M_2 h(Z_{21}) + M_1 h(Z_{11})}\\
&= -\Es{\sum_{t=1}^N \hs{\breve V_{2,t}|W_1, \breve V_2^{t-1}, Z_{11}^{(t-1)}, Z_2^{[t-1]} Z_{21}^{(t-1)}} }\\
&\quad  - N(1-p_2)h(Z_2) - Np_2 h(Z_{21}) - Np_1 h(Z_{11})
\end{align*}
where (b) follows by Lemma \ref{lem:YgW}, (c) follows by chain rule, (d) follows by the fact that for a given time slot $t$, the involved noise terms are independent from each other and from $\breve V_{2,t}$; and (e) is because the signals up to time $t-1$ are independent from the noise at time $t$, and because noise processes are i.i.d.

The other term in \eqref{eq:three} can be bounded similarly.

Putting everything together, we have
\begin{align}
&N(R_1+R_2-\epsilon_N) \notag\\
& \leq \Es{ \hs{\bar Y_1^{L_1}|\breve V_1^N} } + \Es{ \hs{Y_{11}^{M_1}|\breve V_1^N} } \notag\\
& + \Es{ \hs{\bar Y_2^{L_2}|\breve V_2^N} } + \Es{ \hs{Y_{22}^{M_2}|\breve V_2^N} } \notag\\
& + \Es{ \hsc{\breve V_1^N, Y_{12}^{M_1}} } \label{eq:four}\\
& + \Es{ \hsc{\breve V_2^N, Y_{21}^{M_2}} } \label{eq:five}\\
&  -\Es{\sum_{t=1}^N \hs{\breve V_{2,t}|W_1, \breve V_2^{t-1}, Z_{11}^{(t-1)}, Z_2^{[t-1]} Z_{21}^{(t-1)}} } \label{eq:six}\\
&  -\Es{\sum_{t=1}^N \hs{\breve V_{1,t}|W_2, \breve V_1^{t-1}, Z_{22}^{(t-1)}, Z_1^{[t-1]} Z_{12}^{(t-1)}} } \label{eq:seven}\\
&  - N(1-p_2)h(Z_2) - Np_2 h(Z_{21}) - Np_1 h(Z_{11}) \notag\\
& - N(1-p_1)h(Z_1) - Np_1 h(Z_{12}) - Np_2 h(Z_{22}) \notag
\end{align}
Let us combine \eqref{eq:four} and \eqref{eq:seven}.
\begin{align*}
&\eqref{eq:four}+\eqref{eq:seven} = \mbb{E}_{S^N} \lb \hsc{\breve V_1^N, Y_{12}^{M_1}} \right. \\
&\quad \left. - \sum_{t=1}^N \hs{\breve V_{1,t}|W_2, \breve V_1^{t-1}, Z_{22}^{(t-1)}, Z_1^{[t-1]} Z_{12}^{(t-1)}} \rb \\
&= \mbb{E}_{S^N} \lb \sum_{t=1}^N \hs{\breve V_{1,t}|\breve V_1^{t-1}, Y_{12}^{(t-1)}} \right.\\
& \quad - \hs{\breve V_{1,t}|W_2, \breve V_1^{t-1}, Z_{22}^{(t-1)}, Z_1^{[t-1]} Z_{12}^{(t-1)}} \\
& \quad + \left. \sum_{t:S_{1,t}=1} \hs{Y_{12,t}| \breve V_1^{t}, Y_{12}^{(t-1)}} \rb \\
& \leq  \mbb{E}_{S^N}  \lb \sum_{t:S_{1,t}=1} \hsc{Y_{12,t}} \right.\\
&\quad \left.\sum_{t=1}^N I_S(\breve V_{1,t};W_2, Z_{22}^{(t-1)}, Z_1^{[t-1]} Z_{12}^{(t-1)} |\breve V_1^{t-1}, Y_{12}^{(t-1)} \rb
\end{align*}
Similarly, we combine \eqref{eq:five} with \eqref{eq:six} to obtain the same expression with user indices swapped. Plugging these back, we get
\begin{align*}
&N(R_1+R_2-\epsilon_N) \\
& \leq \Es{ \hs{\bar Y_1^{L_1}|\breve V_1^N} } + \Es{ \hs{\bar Y_{11}^{M_1}|\breve V_1^N} } \\
& +  \mbb{E}_{S^N} \lb \hs{\bar Y_2^{L_2}|\breve V_2^N} + \hs{\bar Y_{22}^{M_2}|\breve V_2^N} \right. \\
& + \sum_{t=1}^N I_S(\breve V_{1,t};W_2, Z_{22}^{(t-1)}, Z_1^{[t-1]} Z_{12}^{(t-1)} |\breve V_1^{t-1}, Y_{12}^{(t-1)})  \\
& + \left. \sum_{t=1}^N I_S(\breve V_{2,t};W_1, Z_{11}^{(t-1)}, Z_2^{[t-1]} Z_{21}^{(t-1)} |\breve V_2^{t-1}, Y_{21}^{(t-1)}) \rb \\
& + \Es{ \sum_{t:S_{1,t}=1} \hsc{Y_{12,t}} } + \Es{ \sum_{t:S_{2,t}=1} \hsc{Y_{21,t}} } \\
&  - N(1-p_2)h(Z_2) - Np_2 h(Z_{21}) - Np_1 h(Z_{11}) \\
&  - N(1-p_1)h(Z_1)- Np_1 h(Z_{12}) - Np_2 h(Z_{22})
\end{align*}
We use Lemmas \ref{lem:hT}, \ref{lem:I0}, \ref{lem:YgV} and \ref{lem:UgV} to bound each of these terms, and use the fact that noise distribution is Gaussian to obtain the desired bound.

\subsection{Proof of the Bounds \eqref{eq:ob_g_2R1R2} and \eqref{eq:ob_g_R12R2}}
We excelusively focus on the enhanced channel, defined in Section~\ref{sec:converse}. By symmetry, it is sufficient to prove \eqref{eq:ob_g_2R1R2}. In addition to the enhanced interference channel in the case of sum rate bound, we consider an additional copy of Receiver 1, who always receives $\bar Y_1$ (\emph{i.e.}, as in the original channel). The feedback signal of Tx1 is still given by $S_1 \cdot \bar Y_1$, \emph{i.e.}, the same as the original channel. We would like to prove a sum rate upper bound on this new channel.

By Fano's inequality,
\begin{align}
&N(2R_1+R_2-\epsilon_N) \notag\\
&\quad \leq I(W_1;\bar Y_1^N, \ul{S}^N) +  I(W_2;\breve Y_2^N, \ul{S}^N) + I(W_1;\breve Y_1^N, \ul{S}^N) \notag\\
&\quad = I(W_1;\bar Y_1^N| \ul{S}^N) +  I(W_2;\breve Y_2^N| \ul{S}^N) + I(W_1;\breve Y_1^N| \ul{S}^N) \notag\\
&\quad = I(W_1;\bar Y_1^N| \ul{S}^N) +  I(W_2;\breve Y_2^N, \bar V_2^N| \ul{S}^N) \notag\\
&\qquad + I(W_1;\breve Y_1^N, \breve V_1^N| W_2, \ul{S}^N) \label{eq:eight}
\end{align}
The first mutual information term in \eqref{eq:eight} can be bounded as follows
\begin{align*}
I(W_1;\bar Y_1^N| \ul{S}^N) &= h(\bar Y_1^N| \ul{S}^N) - h(\bar Y_1^N|W_1, \ul{S}^N) \\
&= h(\bar Y_1^N| \ul{S}^N) - \sum_{t=1}^N h(\bar Y_{1,t}|W_1, \bar Y_1^{t-1}, \ul{S}^N) \\
&\overset{\aaaa}{=} h(\bar Y_1^N| \ul{S}^N) - \sum_{t=1}^N h(\bar V_{2,t}|W_1, \bar V_2^{t-1}, \ul{S}^N) \\
&\leq \sum_{t=1}^N h(\bar Y_{1,t}) - h(\bar V_{2,t}|W_1, \bar V_2^{t-1}, \ul{S}^N)
\end{align*}
where (a) follows by the fact that $X_{1,t}\eqFunc\lp W_1, \bar Y_1^{t-1}, \ul{S}^{t-1}\rp$ and by subtracting $X_{1,t}$ from $\bar Y_{1,t}$.

Let us consider the second mutual information term from \eqref{eq:eight}.
\begin{align*}
&I(W_2;\breve Y_2^N, \bar V_2^N| \ul{S}^N)\\
& = \Es{ \hsc{\breve Y_2^N, \bar V_2^N} - \hs{\breve Y_2^N, \bar V_2^N|W_2} } \\
&\overset{\aaaa}{=} \mbb{E}_{S^N} \lb \hsc{Y_{22}^{M_2}, Y_{21}^{M_2}, \bar Y_2^{L_2}, \bar V_2^N} \right.\\
&\quad \left.- \hs{\breve V_1^N, Z_1^N, Z_{22}^{M_2}|W_2} \rb \\
&= \Es{ \hs{Y_{22}^{M_2}, \bar Y_2^{L_2}|Y_{21}^{M_2}, \bar V_2^N}}  \\
&\quad + \Es{ \hsc{Y_{21}^{M_2}, \bar V_2^N} - \hs{\breve V_1^N, Z_1^N, Z_{22}^{M_2}|W_2} } \\
&\overset{\bbbb}{\leq}  \Es{ \hs{Y_{22}^{M_2}| \bar V_2^N} + \hs{ \bar Y_2^{L_2}|\bar V_2^N} } \\
&\quad + \mbb{E}_{S^N} \lb \sum_{t=1}^N \hs{\bar V_{2,t}|\bar V_2^{t-1}, Y_{21}^{(t-1)}} \right.\\
&\quad  + \sum_{t:S_{2,t}=1}^N \hs{Y_{21,t}|\bar V_2^{t-1}, Y_{21}^{(t-1)}} \\
&\quad  - \left.\sum_{t=1}^N \hs{\breve V_{1,t}, Z_{1,t}, Z_{22,(t)}|W_2, \breve V_1^{t-1}, Z_1^{t-1}, Z_{22}^{(t-1)} }\rb \\
&\overset{\cccc}{=}  \Es{ \hs{Y_{22}^{M_2}| \bar V_2^N} + \hs{ \bar Y_2^{L_2}|\bar V_2^N} } \\
&\quad + \mbb{E}_{S^N} \lb \sum_{t=1}^N \hs{\bar V_{2,t}|\bar V_2^{t-1}, Y_{21}^{(t-1)}} \right.\\
&\quad \left. + \sum_{t:S_{2,t}=1}^N \hs{Y_{21,t}|\bar V_2^{t-1}, Y_{21}^{(t-1)}} \rb \\
&\quad  - \mbb{E}_{S^N}\lb\sum_{t=1}^N \hs{\breve V_{1,t}|W_2, \breve V_1^{t-1}, Z_1^{t-1}, Z_{22}^{(t-1)} }\right.  \\
&\quad  + \sum_{t=1}^N \hs{Z_{1,t}|W_2, \breve V_1^{t-1}, Z_1^{t-1}, Z_{22}^{(t-1)} }  \\
&\quad + \left. \sum_{t: S_{2,t}=1} \hs{Z_{22,t}|W_2, \breve V_1^{t-1}, Z_1^{t-1}, Z_{22}^{(t-1)} } \rb \\
&\overset{\dddd}{\leq} \Es{ \hs{Y_{22}^{M_2}| \bar V_2^N} + \hs{ \bar Y_2^{L_2}|\bar V_2^N} } \\
&\quad + \Es{ \sum_{t=1}^N \hs{\bar V_{2,t}|\bar V_2^{t-1}, Y_{21}^{(t-1)}} + \sum_{t:S_{2,t}=1} \hsc{Y_{21,t}}}\\
&\quad - \Es{\sum_{t=1}^N \hs{\breve V_{1,t}|W_2, \breve V_1^{t-1}, Z_1^{t-1}, Z_{22}^{(t-1)} } } \\
&\quad - Nh(Z_1) - Np_2h(Z_{22})
\end{align*}
where (a) is by decomposing $\breve Y_2^N$ into $\lp Y_{22}^{M_2}, Y_{21}^{M_2}\rp$ for time slots where $S_{2,t}=1$, and to $\bar Y_2^{L_2}$ for time slots where $S_{2,t}=0$, and by Lemma \ref{lem:YgW}. (b) is because conditioning reduces entropy and by chain rule. (c) is because for a given time slot $t$, the noise terms involved are independent from each other and from $\breve V_{1,t}$. (d) is because conditioning reduces entropy and the noise processes are i.i.d., and because the noise terms are independent from the signals up to time $t-1$.

Next, we consider the third mutual information term from \eqref{eq:eight}. 
\begin{align*}
&I(W_1;\breve Y_1^N, \breve V_1^N| W_2, \ul{S}^N) \\
&= \Es{ I_S(W_1;\breve Y_1^N, \breve V_1^N| W_2 )} \\
&\overset{\aaaa}{\leq} \Es{ I(W_1;\breve Y_1^N, \breve V_1^N| W_2, Z_{22}^{M_2}, Z_{11}^{M_1})} \\
&= \Es{ \hs{\breve Y_1^N, \breve V_1^N| W_2, Z_{22}^{M_2}, Z_{11}^{M_1}} }\\
& - \Es{\hs{\breve Y_1^N, \breve V_1^N| W_2, Z_{22}^{M_2}, Z_{11}^{M_1}, W_1} } \\
&= \sum_{t=1}^N \Es{ \hs{\breve Y_{1,t}, \breve V_{1,t}| \breve Y_1^{t-1},\breve V_1^{t-1}, W_2, Z_{22}^{M_2}, Z_{11}^{M_1}} }\\
& - \Es{\hs{\breve Y_{1,t}, \breve V_{1,t}| \breve Y_1^{t-1},\breve V_1^{t-1}, W_2, Z_{22}^{M_2}, Z_{11}^{M_1}, W_1} } \\
&= \sum_{t=1}^N \Es{ \hs{\breve Y_{1,t}| \breve Y_1^{t-1},\breve V_1^{t}, W_2, Z_{22}^{M_2}, Z_{11}^{M_1}}} \\
& + \Es{\hs{\breve V_{1,t}| \breve Y_1^{t-1},\breve V_1^{t-1}, W_2, Z_{22}^{M_2}, Z_{11}^{M_1}}}\\
&  - \Es{\hs{\breve Y_{1,t}, \breve V_{1,t}| \breve Y_1^{t-1},\breve V_1^{t-1}, W_2, Z_{22}^{M_2}, Z_{11}^{M_1}, W_1} } \\ 
&\overset{\bbbb}{=} \sum_{t=1}^N \Es{ \hs{\breve Y_{1,t}| \breve Y_1^{t-1},\breve V_1^{t}, W_2, Z_{22}^{M_2}, Z_{11}^{M_1}, X_{2,t}}} \\
& + \Es{\hs{\breve V_{1,t}| \breve Y_1^{t-1},\breve V_1^{t-1}, W_2, Z_{22}^{M_2}, Z_{11}^{M_1}}}\\
&  - \mbb{E}_{S^N} \lb h_S \lp \breve Y_{1,t}, \breve V_{1,t}\right.\right.\\
&\qquad \quad \left.\left. | \breve Y_1^{t-1},\breve V_1^{t-1}, W_2, Z_{22}^{M_2}, Z_{11}^{M_1}, W_1, X_{1,t}, X_{2,t} \rp \rb  \\
&\overset{\cccc}{\leq} \sum_{t=1}^N \Es{ \hs{\breve Y_{1,t}| \breve V_{1,t}, X_{2,t}} }\\
& + \Es{\hs{\breve V_{1,t}| Y_{12}^{(t-1)},\breve V_1^{t-1}, W_2, Z_{22}^{M_2}, Z_{11}^{M_1}}}\\ 
&  - \Es{\sum_{t:S_{2,t}=0} \hsc{Z_{2,t}} + \sum_{t:S_{2,t}=1} \hsc{Z_{21,t}}}\\
&  - \Es{\sum_{t:S_{1,t}=0} \hsc{Z_{1,t}}+ \sum_{t:S_{1,t}=1} \hsc{Z_{11,t}, Z_{12,t}}}\\
&\overset{\dddd}{\leq} \sum_{t=1}^N \Es{ \hs{\breve Y_{1,t}| \breve V_{1,t}, X_{2,t}} }\\
& + \Es{\hs{\breve V_{1,t}| Y_{12}^{(t-1)},\breve V_1^{t-1}, W_2, Z_{22}^{(t-1)}, Z_{11}^{(t-1)}}}\\
&  - N(1-p_1)h(Z_1) - Np_1h(Z_{11})- Np_1h(Z_{12})\\
&  - Np_2h(Z_{21}) - N(1-p_2)h(Z_2)
\end{align*}
where (a) follows by the fact that $\lp Z_{11}^{M_1}, Z_{22}^{M_2} \rp$ is independent from $W_1$ given $W_2$, (b) is because $X_{1,t}\eqFunc\lp W_1, \breve Y_1^{t-1}, \ul{S}^{t-1}\rp$ and $X_{2,t}\eqFunc\lp W_2, \breve V_1^{t-1}, Z_{22}^{(t-1)}, \ul{S}^{t-1}\rp$. In (c), the first two terms are upper bounded using the fact that $\breve Y_1^{t-1} = \lp Y_{11}^{(t-1)}, Y_{12}^{(t-1)}, \bar Y_1^{[t-1]} \rp$, and that conditioning reduces entropy. The noise terms are obtained by subtracting $X_{1,t}$ and $X_{2,t}$ from $\breve Y_{1,t}$ and $\breve V_{1,t}$, and using the fact that noise variables at time $t$ are independent from the variables up to time $t-1$. (d) is because conditioning reduces entropy and because noise processes are i.i.d.

Putting everything back together, we have
\begin{align*}
&N(2R_1+R_2-\epsilon_N) \\
& \leq \sum_{t=1}^N h(\bar Y_{1,t}) + \Es{ \hs{Y_{22}^{M_2}| \bar V_2^N} + \hs{ \bar Y_2^{L_2}|\bar V_2^N} } \\
&\quad + \mbb{E}_{S^N} \lb \vphantom{\sum_{t:S_{2,t}=1}}
\sum_{t=1}^N I_S(\bar V_{2,t}; W_1|\bar V_2^{t-1}, Y_{21}^{(t-1)})\right. \\
&\quad +  \sum_{t=1}^N  I_S(\breve V_{1,t}; Z_{1}^{t-1}| Y_{12}^{(t-1)},\breve V_1^{t-1}, W_2, Z_{22}^{(t-1)}, Z_{11}^{(t-1)})\\
&\quad + \left. \sum_{t=1}^N \hs{\breve Y_{1,t}| \breve V_{1,t}, X_{2,t}}+\sum_{t:S_{2,t}=1} \hsc{Y_{21,t}} \rb\\
&\quad - N(1-p_1)h(Z_1) - Np_1h(Z_{11})- Np_1h(Z_{12}) \\
&\quad - Np_2h(Z_{21}) - N(1-p_2)h(Z_2)- Nh(Z_1)\\
&\quad  - Np_2h(Z_{22})
\end{align*}
Using Lemma \ref{lem:hT}, \ref{lem:I0}, \ref{lem:I0_1}, \ref{lem:YgV}, \ref{lem:UgV} and \ref{lem:priv} to bound each of these terms, we get the desired bound.

\subsection{Lemmas}
In this subsection, we prove the lemmas that have been used in the proofs of the previous subsections.
\begin{lemma}\label{lem:xfunc}
$X_{2,t} \eqFunc \lp W_2, \wtild V_1^{t-1}, \ul{S}^t \rp$
\end{lemma}
\begin{proof}
Note that
\begin{align*}
X_{2,1} \eqFunc W_2
\end{align*}
and by the definition of the channel,
\begin{align*}
X_{2,t} &\eqFunc \lp W_2, \wtild Y_2^{t-1}, \ul{S}^t \rp \\
&\overset{\aaaa}{\eqFunc} \lp W_2, \wtild V_1^{t-1}, X_2^{t-1}, \ul{S}^t \rp,
\end{align*}
hence the result follows by induction on $t$. (a) follows because $\wtild Y_2^{t-1} = S_2^{t-1}h_{22}X_{2}^{t-1} + \wtild V_1^{t-1}$.
\end{proof}

\begin{lemma}\label{lem:hYtV}
\begin{align*}
\sum_{t=1}^N h(Y_{1,t}| \wtild V_{1,t}, S_{2,t}, X_{2,t}) &\leq Np_2 \log \lp 1+\frac{\SNR_1}{1+\INR_2}  \rp \\
&\quad + N(1-p_2)\log \lp 1+\SNR_1 \rp 
\end{align*}
\end{lemma}
\begin{proof}
\begin{align*}
&\sum_{t=1}^N h(Y_{1,t}| \wtild V_{1,t}, S_{2,t}, X_{2,t}) = \sum_{t=1}^N p_2 h(Y_{1,t}| V_{1,t}, X_{2,t}) \\
&\qquad + (1-p_2) h(Y_{1,t}| X_{2,t}) \\
&\quad \overset{\aaaa}{=}\sum_{t=1}^N p_2 h(Y_{1,Q}| V_{1,Q}, X_{2,Q},Q=t)\\
&\qquad + (1-p_2) h(Y_{1,Q}| X_{2,Q},Q=t) \\
&\quad = N p_2 h(Y_{1,Q}| V_{1,Q}, X_{2,Q},Q)\\
&\qquad + N(1-p_2) h(Y_{1,Q}| X_{2,Q},Q) \\
&\quad \overset{\bbbb}{=} N p_2 h(Y_1| V_1, X_2, Q)+ N(1-p_2) h(Y_1| X_2,Q) \\
&\quad \overset{\cccc}{\leq} Np_2 \log \lp 1+\frac{\SNR_1}{1+\INR_2}  \rp \\
&\qquad + N(1-p_2)\log \lp 1+\SNR_1 \rp 
\end{align*}
where (a) follows by introducing a time-sharing variable $Q$ uniformly distributed between 1 and $N$, (b) is by defining $Y_{i}:=Y_{i,Q}$, $X_{i}:=X_{i,Q}$ and $V_{i}:=V_{i,Q}$. (c) follows by the fact that choosing jointly Gaussian input distribution with correlation coefficient $\rho=0$ for $p(x_1,x_2)$ maximizes the given conditional differential entropy.
\end{proof}

\begin{lemma} \label{lem:hT}
\begin{align*}
\Es{\sum_{t:S_{i,t}=1} \hsc{Y_{ij,t}}} = Np_i \log\lp 2\pi e \lp \frac{1}{2}+ \INR_i\rp \rp
\end{align*} for $\ijj$.
\end{lemma}
\begin{proof}
\begin{align*}
&\Es{\sum_{t:S_{i,t}=1} \hsc{Y_{ij,t}}}  = \Es{M_i \frac{1}{M_i}\sum_{t:S_{i,t}=1} \hsc{Y_{ij,t}} } \\
& \leq \Es{M_i \frac{1}{M_i}\sum_{t:S_{i,t}=1} \log 2\pi e \lp \frac{1}{2}+ |h_{ij}|^2 P_{j,t}\rp } \\
& =  \mathbb{E} \lb M_i \mathbb{E} \lb \left.  \frac{1}{M_i}\sum_{t:S_{i,t}=1} \log 2\pi e \lp \frac{1}{2}+ |h_{ij}|^2 P_{j,t}\rp \right|M_i \rb \rb \\
& \overset{\aaaa}{\leq} \mathbb{E}_{M_i} \lb M_i \mathbb{E}\lb \left. \log 2\pi e \lp \frac{1}{2}+ |h_{ij}|^2 P_{j}^{(i1)} \rp \right|M_i \rb \rb \\
& \overset{\bbbb}{\leq} \mathbb{E}_{M_i} \lb M_i \log 2\pi e \lp \frac{1}{2}+ |h_{ij}|^2 \E{P_{j}^{(i1)}|M_i} \rp  \rb \\
&\overset{\cccc}{\leq} \mathbb{E}_{M_i} \lb M_i \log 2\pi e \lp \frac{1}{2}+ |h_{ij}|^2 P_j \rp  \rb \\
& = \E{M_i} \log\lp 2\pi e \lp \frac{1}{2}+ \INR_i\rp \rp \\
& = Np_i \log\lp 2\pi e \lp \frac{1}{2}+ \INR_i \rp \rp
\end{align*}
where (a) and (b) follow by Jensen's inequality (since $\log(\cdot)$ is concave), and (c) follows since $P_{j}^{(i1)}$ averaged over the realizations of $S_i$ is the average power, which is less than the power constraint $P_j$.
\end{proof}

\begin{lemma} \label{lem:I0}
\begin{align*}
&\Es{I_S(V_{i,t}; W_j, Z_{jj}^{(t-1)}, Z_i^{[t-1]}, Z_{ij}^{(t-1)}|V_i^{t-1}, Y_{ij}^{(t-1)})} \\
&=\Es{I_S(V_{i,t}; Z_{i}^{t-1}| Y_{ij}^{(t-1)},V_i^{t-1}, W_j, Z_{jj}^{(t-1)}, Z_{ii}^{(t-1)}) } \\
&= 0
\end{align*}
for $(i,j)=(1,2),(2,1)$.
\end{lemma}
\begin{proof}
Since all variables involved are related to $V_{i,t}$ through $X_{i,t}$; by data processing inequality,
\begin{align*}
&\Es{I_S(V_{i,t}; W_j, Z_{jj}^{(t-1)}, Z_i^{[t-1]}, Z_{ij}^{(t-1)}|V_i^{t-1}, Y_{ij}^{(t-1)})} \\
& \leq \Es{I_S(X_{i,t}; W_j, Z_{jj}^{(t-1)}, Z_i^{[t-1]}, Z_{ij}^{(t-1)}|V_i^{t-1}, Y_{ij}^{(t-1)})} \\
&\leq \mbb{E}_{S^N} \lb I_S(W_i, Z_{ii}^{(t-1)}; W_j, Z_{jj}^{(t-1)}, Z_i^{[t-1]}, Z_{ij}^{(t-1)} \right. \\
&\qquad\qquad\qquad\qquad\qquad\qquad\qquad\left. |V_i^{t-1}, Y_{ij}^{(t-1)}) \rb
\end{align*}
where the latter inequality follows by the fact that $X_{i,t} \eqFunc \lp W_i, Y_{ij}^{(t-1)}, Z_{ii}^{(t-1)}, \ul{S}^{t-1}\rp$. Similarly, the second mutual information can be bounded by
\begin{align*}
&\Es{I_S(V_{i,t}; Z_{i}^{t-1}| Y_{ij}^{(t-1)},V_i^{t-1}, W_j, Z_{jj}^{(t-1)}, Z_{ii}^{(t-1)})} \\
&\leq  \mbb{E}_{S^N} \lb I_S(X_{i,t}; Z_{i}^{[t-1]}, Z_{ij}^{(t-1)} \right.\\
&\qquad \qquad\qquad\left. | Y_{ij}^{(t-1)},V_i^{t-1}, W_j, Z_{jj}^{(t-1)}, Z_{ii}^{(t-1)}) \rb \\
&\leq \mbb{E}_{S^N} \lb I_S(W_i, Z_{ii}^{(t-1)}; Z_{jj}^{(t-1)}, Z_i^{[t-1]}, Z_{ij}^{(t-1)}, W_j \right. \\
&\qquad\qquad\qquad \qquad \qquad \qquad \left. |V_i^{t-1}, Y_{ij}^{(t-1)}) \rb
\end{align*}
where the first step is because $Z_{i}^{t-1} \eqFunc \lp Z_{i}^{[t-1]}, Z_{ij}^{(t-1)}, Z_{ii}^{(t-1)}, \ul{S}^{t-1}\rp$, and second step is because for random variables $A,B,C$; $I(A,B;C) \geq I(A;C|B)$; and  $X_{i,t} \eqFunc \lp W_i, Y_{ij}^{(t-1)}, Z_{ii}^{(t-1)}, \ul{S}^{t-1}\rp$. Note that we have the same upper bound for both mutual information terms. We will next show that this upper bound is zero.

To show conditional independence, we will use the property that $X$ and $Y$ are independent given $Z$ if and only if the probability distribution $p(X,Y,Z)$ can be factorized as
\begin{align*}
p(X,Y,Z) = f(X,Z)g(Y,Z)
\end{align*}
for some functions $f$ and $g$. Consider the joint distribution of all the variables involved in the above mutual information (we define $p_S(\cdot):=p(\cdot|\ul{S}^N=S^N)$).
\begin{align*}
&p_S(W_i, Z_{ii}^{(t-1)}, Z_{jj}^{(t-1)}, Z_i^{[t-1]}, Z_{ij}^{(t-1)}, W_j, V_i^{t-1}, Y_{ij}^{(t-1)}) =\\
& p(W_i)p(W_j) \prod_{\tau = 1}^{t-1} p_S(Z_{ii,(\tau)}, Z_{jj,(\tau)}, Z_{i,[\tau]}, Z_{ij,(\tau)}, V_{i,\tau}, Y_{ij,(\tau)} \\
&\quad|Z_{ii}^{(\tau-1)}, Z_{jj}^{(\tau-1)}, Z_i^{[\tau-1]}, Z_{ij}^{(\tau-1)}, V_i^{\tau-1}, Y_{ij}^{(\tau-1)}, W_i, W_j) \\
&\overset{\aaaa}{=} p(W_i)p(W_j) \prod_{\tau = 1}^{t-1} p_S(Z_{ii,(\tau)})p_S(Z_{jj,(\tau)})p_S(Z_{i,[\tau]}) \\
&\qquad \cdot p_S(Z_{ij,(\tau)}) p_S(V_{i,\tau}|Z_{ii}^{(\tau-1)}, Y_{ij}^{(\tau-1)}, W_i) \\
&\qquad \cdot p_S(Y_{ij,(\tau)}|Z_{jj}^{(\tau-1)}, Z_{ij}^{(\tau)}, V_i^{\tau-1}, W_j)\\
&= f(W_i, Z_{ii}^{(t-1)}, V_i^{t-1}, Y_{ij}^{(t-1)}) \\
&\qquad \cdot g(Z_{jj}^{(t-1)}, Z_i^{(t-1)}, Z_{ij}^{(t-1)}, W_j, V_i^{t-1}, Y_{ij}^{(t-1)})
\end{align*} 
where (a) follows since 
\begin{align*}
Y_{ij,(\tau)} &\eqFunc \lp X_{j,\tau}, Z_{ij}^{(\tau)}, \ul{S}^{\tau} \rp \\
&\eqFunc \lp Z_{jj}^{(\tau-1)}, Z_{ij}^{(\tau)}, V_i^{\tau-1}, \ul{S}^{\tau}, W_j \rp
\end{align*}
and 
\begin{align*}
V_{i,\tau} &\eqFunc \lp X_{i,\tau}, Z_{j,(\tau)}, Z_{ji, (\tau)} \rp \\
&\eqFunc \lp  Z_{ii}^{(\tau-1)}, Y_{ij}^{(\tau-1)}, \ul{S}^{\tau-1}, W_i, Z_{j,(\tau)}, Z_{ji, (\tau)} \rp
\end{align*}
and $\lp Z_{j,(\tau)}, Z_{ji, (\tau)} \rp$ is independent of everything else. In the last line, we define
\begin{align*}
&f(W_i, Z_{ii}^{(t-1)}, V_i^{t-1}, Y_{ij}^{(t-1)}) \\
&\quad = p(W_i) \prod_{\tau = 1}^{t-1} p_S(Z_{ii,(\tau)}) p_S(V_{i,\tau}|Z_{ii}^{(\tau-1)}, Y_{ij}^{(\tau-1)}, W_i) \\
&g(Z_{jj}^{(t-1)}, Z_i^{(t-1)}, Z_{ij}^{(t-1)}, W_j, V_i^{t-1}, Y_{ij}^{(t-1)}) \\
&\quad = p(W_j)\prod_{\tau = 1}^{t-1} p_S(Z_{jj,(\tau)}) p_S(Z_{i,(\tau)})p_S(Z_{ij,(\tau)})\\
&\quad \cdot p_S(Y_{ij,\tau}|Z_{jj}^{(\tau-1)}, Z_{ij}^{(\tau)}, V_i^{\tau-1}, W_j)
\end{align*}
from which the result follows.
\end{proof}

\begin{lemma}\label{lem:I0_1}
\begin{align*}
\Es{I_S(\bar V_{j,t}; W_i | \bar V_j^{t-1}, Y_{ji}^{(t-1)})} = 0
\end{align*}
\end{lemma}
\begin{proof}
\begin{align*}
&\Es{I_S(\bar V_{j,t}; W_i | \bar V_j^{t-1}, Y_{ji}^{(t-1)})} \\
& \quad\leq \Es{I_S(X_{j,t}; W_i | \bar V_j^{t-1}, Y_{ji}^{(t-1)})} \\
& \quad \leq \Es{I_S(W_j, Z_{jj}^{(t-1)}; W_i | \bar V_j^{t-1}, Y_{ji}^{(t-1)})} \\
\end{align*}
where the first step follows by data processing inequality, and the second one follows by $X_{j,t} \eqFunc \lp W_j, Y_{ji}^{(t-1)}, Z_{jj}^{(t-1)}, \ul{S}^{t-1} \rp$. The proof technique is similar to that of Lemma~\ref{lem:I0}. The probability distribution of the involved variables is
\begin{align*}
&p(W_i, W_j, Z_{jj}^{(t-1)}, \bar V_j^{t-1}, Y_{ji}^{(t-1)}) \\
& = p(W_i)p(W_j) \prod_{\tau=1}^{t-1} p_S(Z_{jj,(\tau)} ) \\
&\qquad p_S(\bar V_{j,\tau}, Y_{ji,(\tau)}|\bar V_{j}^{\tau-1}, Y_{ji}^{(\tau-1)},Z_{jj}^{(\tau-1)},W_i,W_j) \\
& \overset{\aaaa}{=} p(W_i)p(W_j) \prod_{\tau=1}^{t-1} p_S(Z_{jj,(\tau)}) \\
&\qquad p_S(\bar V_{j,\tau}|Y_{ji}^{(\tau-1)},Z_{jj}^{(\tau-1)},W_j) p_S(Y_{ji,(\tau)}|W_i, \bar V_j^{t-1})
\end{align*}
where (a) follows by the fact that 
\begin{align*}
 \bar V_{j,\tau} &\eqFunc \lp W_j, Y_{ji}^{(\tau-1)}, Z_{jj}^{(t-1)}, \ul{S}^{t-1}, Z_{i,\tau} \rp, \\
 Y_{ji,(\tau)} &\eqFunc \lp  W_i, \bar V_j^{t-1}, \ul{S}^{t-1}, Z_{ji,(\tau)}\rp
\end{align*}
and that $Z_{i,\tau}$ and $Z_{ji,(\tau)}$ are independent of everything else. Then the result follows by defining
\begin{align*}
&f(W_i, \bar V_j^{t-1}, Y_{ji}^{(t-1)}) \\
&\quad = p(W_i) \prod_{\tau=1}^{t-1}p_S(Y_{ji,(\tau)}|W_i, \bar V_j^{t-1}) \\
&g(W_j, Z_{jj}^{(t-1)}, \bar V_j^{t-1}, Y_{ji}^{(t-1)}) \\
&\quad = p(W_j) \prod_{\tau=1}^{t-1} p_S(Z_{jj,(\tau)}) p_S(\bar V_{j,\tau}|Y_{ji}^{(\tau-1)},Z_{jj}^{(\tau-1)},W_i)
\end{align*}
and noting that the above probability distribution factorizes as $f \cdot g$.
\end{proof}

\begin{lemma} \label{lem:YgW}
\begin{align*}
&\Es{\hs{\breve Y_i^N, \breve V_i^N|W_i}} \\
&\qquad = \Es{\hs{\breve V_j^N, Z_{ii}^{M_i}, Z_j^{L_j}, Z_{ji}^{M_j}|W_i} } \\
&\Es{\hs{\breve Y_i^N, \bar V_i^N|W_i}} \\
&\qquad = \Es{\hs{\breve V_j^N, Z_{ii}^{M_i}, Z_j^N|W_i} }
\end{align*}
 for $\ijj$.
\end{lemma}
\begin{proof}
\begin{align*}
&\Es{\hs{\breve Y_i^N, \breve V_i^N|W_i}} \\
&\quad = \Es{\sum_{t=1}^N\hs{\breve Y_{i,t}, \breve V_{i,t}|W_i, \breve Y_i^{t-1}, \breve V_i^{t-1} }} \\
&\quad \overset{\aaaa}{=} \Es{\sum_{t=1}^N\hs{\breve Y_{i,t}, \breve V_{i,t}|W_i, \breve Y_i^{t-1}, \breve V_i^{t-1}, X_{i}^t }} \\
&\quad = \mbb{E}_{S^N} \lb \sum_{t=1}^N h_S \lp \bar Y_{i,[t]}, Y_{ii,(t)}, Y_{ij,(t)}, \bar V_{i,[t]}, Y_{ji,(t)}\right.\right.\\
&\qquad\left.\left|W_i, \bar Y_i^{[t-1]}, Y_{ii}^{(t-1)}, Y_{ij}^{(t-1)}, \bar V_i^{[t-1]}, Y_{ji}^{(t-1)}, X_{i}^t  \rp \vphantom{\sum_{t=1}^N}\rb \\ 
&\quad =\mbb{E}_{S^N} \lb \sum_{t=1}^N h_S \lp\bar V_{j,[t]}, Z_{ii,(t)}, Y_{ij,(t)}, Z_{j,[t]}, Z_{ji,(t)} \right.\right.\\
&\qquad \left.\left|W_i, \bar V_j^{[t-1]}, Z_{ii}^{(t-1)}, Y_{ij}^{(t-1)}, Z_j^{[t-1]}, Z_{ji}^{(t-1)} \rp\vphantom{\sum_{t=1}^N}\rb \\ 
&\quad = \Es{\hs{\bar V_{j}^{L_i}, Z_{ii}^{M_i}, Y_{ij}^{M_i}, Z_{j}^{L_j}, Z_{ji}^{M_j} |W_i }} \\ 
&\quad \overset{\bbbb}{=} \Es{\hs{\breve V_{j}^{N}, Z_{ii}^{M_i}, Z_{j}^{L_j}, Z_{ji}^{M_j} |W_i }}
\end{align*}
where (a) is because $X_{i}^t \eqFunc \lp W_i, \breve Y_i^{t-1}, \mcal{S}^{t-1} \rp$, and (b) is because $\breve V_j^N = \lp \bar V_j^{L_i}, Y_{ij}^{M_i} \rp$. The second equality can be proved using similar steps.
\end{proof}

\begin{lemma} \label{lem:YgV}
\begin{align*}
\Es{ \hs{\bar Y_i^{L_i}|\breve V_i^{N} } } &\leq  a\\
\Es{ \hs{\bar Y_i^{L_i}|\bar V_i^{N} } } &\leq a
\end{align*}  
for $\ijj$, where
\begin{align*}
a &= N(1-p_i)\\
&\quad \cdot \log 2\pi e \lp 1+ \INR_i+\frac{\SNR_i  + 2\sqrt{ \SNR_i \cdot \INR_i }}{1 + \INR_j} \rp
\end{align*}
\end{lemma}
\begin{proof}
\begin{align*}
&\Es{ \hs{\bar Y_i^{L_i}|\breve V_i^{N} } } \\
&\leq \Es{ \sum_{t: S_{i,t}=0} \hs{\bar Y_{i,t}|\breve V_{i,t} } } \\
&\overset{\aaaa}{=} \Es{ L_i \frac{1}{L_i}\sum_{t: S_{i,t}=0} \hs{\bar Y_{i,Q}|\breve V_{i,Q}, Q=t } } \\
&\overset{\bbbb}{=} \Es{ L_i \hs{\bar Y_{i}|\breve V_{i}, Q} } \\
&= \Es{ L_i \lp p_j \hs{\bar Y_{i}|Y_{ji}, Q} + (1-p_j) \hs{\bar Y_{i}|\bar V_{i}, Q} \rp} \\
&\overset{\cccc}{\leq} \mbb{E} \lb L_i p_j \log 2\pi e \lp 1+ \INR_i+\frac{\SNR_i  + 2\sqrt{ \SNR_i \cdot \INR_i }}{1 + 2\INR_j} \rp \right. \\
&+ L_i(1-p_j)  \\
&\quad \cdot \left. \log 2\pi e\lp 1+ \INR_i+\frac{\SNR_i  + 2\sqrt{ \SNR_i \cdot \INR_i }}{1 + \INR_j} \rp  \rb \\
&\leq (1-p_i)\log 2\pi e \lp 1+ \INR_i+\frac{\SNR_i  + 2\sqrt{ \SNR_i \cdot \INR_i }}{1 + \INR_j} \rp
\end{align*}
where (a) is by introducing a time-sharing random variable $Q$ with uniform distribution over the set $\lbp t: S_{i,t}=0 \rbp$, and (b) follows by setting $\bar Y_i = \bar Y_{i,Q}$ and $\breve V_i = \breve V_{i,Q}$. (c) follows by the fact that choosing jointly Gaussian input distribution with correlation coefficient $\rho=0$ for $p(x_1,x_2)$ maximizes the given conditional differential entropy.

By following similar steps, we can show that
\begin{align*}
&\Es{ \hs{\bar Y_i^{L_i}|\bar V_i^{N} } } \leq \Es{ L_i \hs{\bar Y_i|\bar V_i, Q}} \\
& \leq (1-p_i)\log 2\pi e \lp 1+ \INR_i+\frac{\SNR_i  + 2\sqrt{ \SNR_i \cdot \INR_i }}{1 + \INR_j} \rp
\end{align*}
\end{proof}

\begin{lemma} \label{lem:UgV}
\begin{align*}
\Es{ \hs{Y_{ii}^{M_i}|\breve V_i^{N} } } &\leq b \\
\Es{ \hs{Y_{ii}^{M_i}|\bar V_i^{N} } } &\leq b
\end{align*}
where
\begin{align*}
b=Np_i \log 2\pi e \lp \frac{1}{2}+\frac{\SNR_i}{2\INR_j+1} \rp
\end{align*}
for $\ijj$.
\end{lemma}
\begin{proof}
\begin{align*}
&\Es{ \hs{Y_{ii}^{M_i}|\breve V_i^{N} } } \leq \Es{ \sum_{t:S_{i,t}=1} \hs{Y_{ii,t}|\breve V_{i,t} } } \\
&\overset{\aaaa}{=} \Es{ M_i \frac{1}{M_i}\sum_{t:S_{i,t}=1} \hs{Y_{ii,Q}|V_{i,Q}, Q=t } } \\
&\overset{\bbbb}{=} \Es{ M_i \hs{Y_{ii}|\breve V_i, Q } } \\
&= \Es{ M_i \lp (1-p_j) \hs{Y_{ii}|\bar V_i, Q } + p_j \hs{Y_{ii}|Y_{ji}, Q } \rp } \\
&\overset{\cccc}{\leq} \Es{ M_i \lp (1-p_j)\log 2\pi e \lp \frac{1}{2}+\frac{2\SNR_i+\frac{1}{2}}{2\INR_j+1} \rp \right. \right. \\
&\quad \left. \left. + p_j \log 2\pi e \lp \frac{1}{2}+\frac{\SNR_i}{2\INR_j+1} \rp \rp } \\
&\leq p_i \log 2\pi e \lp \frac{1}{2}+\frac{\SNR_i}{2\INR_j+1} \rp
\end{align*}
where (a) is by introducing a time-sharing random variable $Q$ with uniform distribution over the set $\lbp t: S_{i,t}=1 \rbp$, and (b) follows by setting $Y_{ii} = Y_{ii,Q}$ and $\breve V_i = \breve V_{i,Q}$. (c) follows by the fact that choosing jointly Gaussian input distribution with correlation coefficient $\rho=0$ for $p(x_1,x_2)$ maximizes the given conditional differential entropy. Similarly,
\begin{align*}
\Es{ \hs{Y_{ii}^{M_i}|\bar V_i^{N} } } &\leq \Es{M_i \hs{Y_{ii}|\bar V_i, Q }} \\
&\leq p_i \log 2\pi e \lp \frac{1}{2}+\frac{\SNR_i}{2\INR_j+1} \rp
\end{align*}
\end{proof}

\begin{lemma}\label{lem:priv}
\begin{align*}
&\Es{\sum_{t=1}^N \hs{\breve Y_{i,t}|\breve V_{i,t}, X_{j,t}}}  \\
&\qquad\leq p_i \log2\pi e\lp \frac{1}{2}+ \frac{\SNR_i}{2\INR_j+1}\rp + p_i\log2\pi e\frac{1}{2} \\
&\qquad \quad  + (1-p_i)\log2\pi e \lp1 + \frac{\SNR_i}{1+\INR_j}\rp
\end{align*}
for $\ijj$.
\end{lemma}
\begin{proof}
\begin{align*}
&\Es{\sum_{t=1}^N \hs{\breve Y_{i,t}|\breve V_{i,t}, X_{j,t}}} \\
&\overset{\aaaa}{=} \Es{N \frac{1}{N}\sum_{t=1}^N \hs{\breve Y_{i,Q}|\breve V_{i,Q}, X_{j,Q}, Q=t}} \\
&= \Es{N \hs{\breve Y_{i,Q}|\breve V_{i,Q}, X_{j,Q}, Q}} \\
&\overset{\bbbb}{=} N h(\breve Y_{i}|\breve V_{i}, X_{j}, Q) \\
&= p_i h(Y_{ii}, Y_{ij}|\breve V_i, X_j, Q) + (1-p_i) h(\bar Y_i|\breve V_i, X_j, Q) \\
&\leq p_i h(Y_{ii}|V_i) + p_i h(Y_{ij}|X_j) \\
&\quad + (1-p_i) \lb (1-p_j) h(\bar Y_i|\bar V_i, X_j) +p_jh(\bar Y_i|Y_{ji}, X_j)  \rb\\
&\overset{\cccc}{\leq} p_i \log2\pi e\lp \frac{1}{2} + \frac{\SNR_i}{2\INR_j+1}\rp + p_i\log2\pi e\frac{1}{2} \\
&\quad + (1-p_i)\log2\pi e \lp1 + \frac{\SNR_i}{1+\INR_j}\rp
\end{align*}
where (a) is by introducing a uniformly distributed time-sharing random variable $Q$, and (b) is by defining $\breve Y_i = \breve Y_{i,Q}$, $\breve V_i = \breve V_{i,Q}$ and $X_j = X_{j,Q}$. (c) follows by the fact that choosing jointly Gaussian input distribution with correlation coefficient $\rho=0$ for $p(x_1,x_2)$ maximizes the given conditional differential entropy.
\end{proof}

%% file: AP_GapAnalysis.tex
In this section, we give upper bounds for the gap terms $\delta_1$ and $\delta_2$ from Theorem~\ref{th:gaussian}. We also compare our achievable region with the outer bound of \cite{SuhTse_11} for the case $p_1=p_2=1$.

\subsection{Bounding $\delta_1$}
We will show that, each of the bounds \eqref{eq:ib_Ri}, \eqref{eq:ib_RiRj}, and \eqref{eq:ib_2RiRj} are within a constant gap of the region given in \eqref{eq:g_Ri}--\eqref{eq:g_2RiRj}. Without loss of generality, we focus on the case $(i,j)=(1,2)$, and start with the first bound in \eqref{eq:ib_Ri}.
\begin{align*}
&\msf{A}_1 + \msf{B}_2 = \log \lp 3 + \frac{\SNR_1}{1+\INR_2} \rp \\
&\quad + \log \lp 2 + \INR_2 \rp - 2\log 3 - C_1 - C_2 \\
&\geq \log \lp 1 + \frac{\SNR_1}{1+\INR_2} \rp \\
&\quad + \log \lp 1 + \INR_2 \rp - 2\log 3 - C_1 - C_2 \\
&= \log \lp 1 + \SNR_1 + \INR_2 \rp - 2\log 3 - C_1 - C_2 \\
&= \log \lp 1 + \SNR_1 \rp +\log\lp 1+\frac{ \INR_2}{1+\SNR_1} \rp \\
&\quad - 2\log 3 - C_1 - C_2 \\
&\geq \eqref{eq:g_Ri}_{(1,2),R} - 2\log 3 - C_1 - C_2
\end{align*}
where $\eqref{eq:g_Ri}_{(1,2),R}$ refers to bound on the right-hand side of \eqref{eq:g_Ri}, evaluated with $(i,j)=(1,2)$. Next, we consider the second bound in \eqref{eq:ib_Ri}. If $\SNR_1 \geq \INR_1$,
\begin{align*}
&\msf{D}_1 =\log\lp 3+\SNR_1\rp - \log 3 - C_1 \\
&\geq \log\lp 1+\SNR_1\rp - \log 3 - C_1 \\
&\geq \log\lp 1+\SNR_1 + \INR_1\rp - \log 3 - C_1 -1 \\
&= \eqref{eq:g_Ri}_{(1,2),L}- \log 3 - C_1 -1
\end{align*}
where $\eqref{eq:g_Ri}_{(1,2),L}$ refers to bound on the left-hand side of \eqref{eq:g_Ri}, evaluated with $(i,j)=(1,2)$. If $\SNR_1 < \INR_1$,
\begin{align*}
&\msf{D}_1 =\log\lp 3+\SNR_1\rp + p_2 \log \lp 1+\frac{\INR_2}{ 3 +\SNR_1 } \rp \\
&\quad  - p_2 \log \frac{5}{3} - \log 3 - C_1 \\
&\geq  \log\lp 1+\SNR_1\rp + p_2 \log \lp 1+\frac{\INR_2}{ 1 +\SNR_1 } \rp \\
&\quad - p_2 \log \frac{5}{3} -p_2\log 3 - \log 3 - C_1 \\
&= \eqref{eq:g_Ri}_{(1,2),R} -p_2 \log 5 - \log 3 - C_1
\end{align*}
Next, we consider the first bound in \eqref{eq:ib_RiRj}.
\begin{align*}
&\msf{A}_1+\msf{G}_2 = \log \lp 3 + \frac{\SNR_1}{1+\INR_2} \rp \\
&\quad + \log \lp 2 + \SNR_2+\INR_2\rp- 2\log 3 - C_1-C_2 - \kappa_1 \\
&\geq \log \lp 1 + \frac{\SNR_1}{1+\INR_2} \rp \\
&\quad + \log \lp 1 + \SNR_2+\INR_2\rp- 2\log 3 - C_1-C_2 - \kappa_1 \\
&= \eqref{eq:g_RiRj_1}_{(1,2)} - 2\log 3 - C_1-C_2 - \kappa_1
\end{align*}
where $\eqref{eq:g_RiRj_1}_{(1,2)}$ refers to the bound \eqref{eq:g_RiRj_1} evaluated with $(i,j)=(1,2)$. For the bound $\msf{F}_1+\msf{F}_2$, we first consider the case when $\INR_1 >\SNR_1$.
\begin{align*}
&\msf{F}_1+\msf{F}_2 \geq \log \lp 2 + \INR_1 + \frac{\SNR_1}{1+\INR_2} \rp \\
&\quad + \log \lp 2 + \INR_2 + \frac{\SNR_2}{1+\INR_1} \rp -2\log 3 - C_1 - C_2 \\
&\geq \log \lp 1 + \INR_1 \rp + \log \lp 1 + \INR_2 + \frac{\SNR_2}{1+\INR_1} \rp \\
&\quad  -2\log 3 - C_1 - C_2 \\
&\geq \log \lp 1 + \SNR_1 +\INR_1\rp + \log \lp 1 + \INR_2 + \frac{\SNR_2}{1+\INR_1} \rp \\
&\quad  -2\log 3 - C_1 - C_2 -1\\
&= \eqref{eq:g_RiRj_1}_{(2,1)}-2\log 3 - C_1 - C_2 -1
\end{align*}
By symmetry, we can show that when $\INR_2 >\SNR_2$,
\begin{align*}
\msf{F}_1+\msf{F}_2 \geq \eqref{eq:g_RiRj_1}_{(1,2)}-2\log 3 - C_1 - C_2 -1
\end{align*}
Next we consider the only remaining case of $\INR_1 \leq \SNR_1$, $\INR_2 \leq \SNR_2$.
\begin{align*}
&\msf{F}_1+\msf{F}_2 = \log \lp 2 + \INR_1 + \frac{\SNR_1}{1+\INR_2} \rp \\
&\quad + \log \lp 2 + \INR_2 + \frac{\SNR_2}{1+\INR_1} \rp \\ 
&\quad + p_1 \log \lp \frac{\lp 2+\INR_1\rp\lp 3+\frac{\SNR_1}{1+\INR_2}\rp}{2+\frac{\SNR_1}{1+\INR_2}+\INR_1} \rp \\
&\quad + p_2 \log \lp \frac{\lp 2+\INR_2\rp\lp 3+\frac{\SNR_2}{1+\INR_1}\rp}{2+\frac{\SNR_2}{1+\INR_1}+\INR_2} \rp \\
&\quad -2\log 3 - 2C_1 - 2C_2 - \lp p_1+p_2\rp\log 6 \\
&\overset{\aaaa}{\geq} \log \lp 1 + \INR_1 + \frac{\SNR_1}{1+\INR_2} \rp \\
&\quad + \log \lp 1 + \INR_2 + \frac{\SNR_2}{1+\INR_1} \rp \\ 
&\quad + p_1 \log \lp \frac{\lp 1+\INR_1\rp\lp 1+\frac{\SNR_1}{1+\INR_2}\rp}{1+\frac{\SNR_1}{1+\INR_2}+\INR_1} \rp \\
&\quad + p_2 \log \lp \frac{\lp 1+\INR_2\rp\lp 1+\frac{\SNR_2}{1+\INR_1}\rp}{1+\frac{\SNR_2}{1+\INR_1}+\INR_2} \rp \\
&\quad -2\log 3 - 2C_1 - 2C_2 - \lp p_1+p_2\rp\log 6 \\
&= \eqref{eq:g_RiRj_2}-2\log 3 - 2C_1 - 2C_2 - \lp p_1+p_2\rp\log 6
\end{align*}
where in (a), we used the fact that the function $\log\lp \frac{x+a}{x+a+b}\rp$ is monotonically increasing in $x$, for $x,a,b>0$. Finally, we consider the bound \eqref{eq:ib_2RiRj}. Again, we distinguish the cases $\INR_2>\SNR_2$ and $\INR_2\leq \SNR_2$. For the former case,
\begin{align*}
&\msf{A}_1+\msf{F}_2+\msf{G}_1 = \log \lp 3 + \frac{\SNR_1}{1+\INR_2} \rp \\
&\quad + \log \lp 1 + \INR_2 + \frac{\SNR_2}{1+\INR_1} \rp  \\
&\quad + \log \lp 2 + \SNR_1+\INR_1\rp- 3\log 3 - 2C_1-C_2 - \kappa_2 \\
&\geq \log \lp 3 + \frac{\SNR_1}{1+\INR_2} \rp \\
&\quad + \log \lp 1 + \INR_2 \rp + \log \lp 2 + \SNR_1+\INR_1\rp \\
&\quad - 3\log 3 - 2C_1-C_2 - \kappa_2 \\
&\geq \log \lp 3 +  \frac{\SNR_1}{1+\INR_2} \rp \\
&\quad + \log \lp 1 + \INR_2 +\SNR_2 \rp + \log \lp 2 + \SNR_1+\INR_1\rp \\
&\quad - 3\log 3 - 2C_1-C_2 - \kappa_2 -1 \\
&\geq \log \lp 1 + \frac{\SNR_1}{1+\INR_2} \rp \\
&\quad + \log \lp 1 + \INR_2 +\SNR_2 \rp + \log \lp 1 + \SNR_1+\INR_1\rp \\
&\quad - 3\log 3 - 2C_1-C_2 - \kappa_2 -1 \\
&= \eqref{eq:g_RiRj_1}_{(1,2)} + \eqref{eq:g_Ri}_{(1,2),L}- 3\log 3 - 2C_1-C_2 - \kappa_2 -1 
\end{align*}
For the case $\INR_2\leq \SNR_2$,
\begin{align*}
&\msf{A}_1+\msf{F}_2+\msf{G}_1 = \log \lp 3 + \frac{\SNR_1}{1+\INR_2} \rp \\
& + \log \lp 1 + \INR_2 + \frac{\SNR_2}{1+\INR_1} \rp + \log \lp 2 + \SNR_1+\INR_1\rp \\
& + p_2 \log \lp \frac{\lp 2+\INR_2\rp\lp 3+\frac{\SNR_2}{1+\INR_1}\rp}{2+\frac{\SNR_2}{1+\INR_1}+\INR_2} \rp \\
& - 3\log 3 - 2C_1-C_2 - p_2\log 6 -\kappa_2 \\
&\overset{\aaaa}{\geq} \log \lp 1 + \frac{\SNR_1}{1+\INR_2} \rp \\
& + \log \lp 1 + \INR_2 + \frac{\SNR_2}{1+\INR_1} \rp + \log \lp 1 + \SNR_1+\INR_1\rp \\
& + p_2 \log \lp \frac{\lp 1+\INR_2\rp\lp 1+\frac{\SNR_2}{1+\INR_1}\rp}{1+\frac{\SNR_2}{1+\INR_1}+\INR_2} \rp \\
& - 3\log 3 - 2C_1-C_2 - p_2\log 6 -\kappa_2 \\
&= \eqref{eq:g_2RiRj}- 3\log 3 - 2C_1-2C_2 - p_2\log 6 -\kappa_2
\end{align*}
where in (a), as before, we used the fact that the function $\log\lp \frac{x+a}{x+a+b}\rp$ is monotonically increasing in $x$, for $x,a,b>0$. By symmetry, similar gaps apply to the case $(i,j)=(2,1)$. Now, we can upper bound $\delta_1$ by noting that it cannot be larger than the maximum of the gaps found above (after proper normalization, \emph{e.g.}, the gap found for the bound on $R_1+R_2$ is divided by 2, and the one on $2R_1+R_2$ is divided by 3). Hence, using the fact that $C_i=2p_j+p_i$ and $\kappa_i=p_i$, we find
\begin{align*}
\delta_1 < 2\log 3 + 3\lp p_1+p_2\rp \text{ bits.}
\end{align*}

\subsection{Bounding $\delta_2$}
In order to bound $\delta_2$, we compare the bounds obtained in Section~\ref{sec:converse} with the bounds \eqref{eq:g_Ri}--\eqref{eq:g_2RiRj} one by one. Without loss of generality, we focus on $(i,j)=(1,2)$, and begin with the first bound in \eqref{eq:g_Ri}.
\begin{align*}
&\eqref{eq:g_Ri}_{(1,2),L} = \log\lp 1+\SNR_1+\INR_1\rp \\
&\geq \log\lp 1+\SNR_1+\INR_1 + 2\sqrt{\SNR_1\cdot\INR_1}\rp - \log 3 \\
&\geq \eqref{eq:ob_p_Ri} - \log 3
\end{align*}
Next, we consider the second bound in \eqref{eq:g_Ri}, and note that $\eqref{eq:ob_g_Ri}=\eqref{eq:g_Ri}_R$, where $\eqref{eq:g_Ri}_R$ refers to the right-hand side of the minimization in \eqref{eq:g_Ri}. We now consider the bound \eqref{eq:g_RiRj_1}.
\begin{align*}
\eqref{eq:g_RiRj_1} &= \log \lp 1 + \frac{\SNR_1}{1+\INR_2} \rp \\
&\quad + \log \lp 1 + \SNR_2 + \INR_2 \rp \\
&\geq  \log \lp 1 + \frac{\SNR_1}{1+\INR_2} \rp \\
&\quad + \log \lp 1 + \SNR_2 + \INR_2 + 2\sqrt{\SNR_2\cdot\INR_2} \rp - \log 3 \\
&\geq \eqref{eq:ob_p_RiRj} - \log 3
\end{align*}
Let us take \eqref{eq:g_RiRj_2}.
\begin{align*}
\eqref{eq:g_RiRj_2} &= \log \lp 1 + \frac{\SNR_1}{1+\INR_2} + \INR_1 \rp \\
&\quad + \log \lp 1 + \frac{\SNR_2}{1+\INR_1} + \INR_2 \rp \\
&\quad + p_1 \log \lp \frac{\lp 1+\INR_1 \rp\lp 1+\frac{\SNR_1}{1+\INR_2}\rp}{1+\frac{\SNR_1}{1+\INR_2}+\INR_1} \rp\\
&\quad + p_2 \log \lp \frac{\lp 1+\INR_2 \rp\lp 1+\frac{\SNR_2}{1+\INR_1}\rp}{1+\frac{\SNR_2}{1+\INR_1}+\INR_2} \rp\\
&\geq \log \lp 1 + \frac{\SNR_1}{1+\INR_2} + \INR_1 +2\sqrt{\SNR_1 \cdot \INR_1}\rp \\
&\quad + \log \lp 1 + \frac{\SNR_2}{1+\INR_1} + \INR_2 +2\sqrt{\SNR_2 \cdot \INR_2}\rp \\
&\quad + p_1 \log \lp \frac{\lp 1+\INR_1 \rp\lp 1+\frac{\SNR_1}{1+\INR_2}\rp}{1+\frac{\SNR_1}{1+\INR_2}+\INR_1} \rp\\
&\quad + p_2 \log \lp \frac{\lp 1+\INR_2 \rp\lp 1+\frac{\SNR_2}{1+\INR_1}\rp}{1+\frac{\SNR_2}{1+\INR_1}+\INR_2} \rp\\
&\quad -2\log 3 \\
&\geq \log \lp 1 + \frac{\SNR_1}{1+\INR_2} + \INR_1 +2\sqrt{\SNR_1 \cdot \INR_1}\rp \\
&\quad + \log \lp 1 + \frac{\SNR_2}{1+\INR_1} + \INR_2 +2\sqrt{\SNR_2 \cdot \INR_2}\rp \\
&\quad + p_1 \log \lp \frac{\lp 1+\INR_1 \rp\lp 1+\frac{\SNR_1}{1+\INR_2}\rp}{1+\frac{\SNR_1+2\sqrt{\SNR_1 \cdot \INR_1}}{1+\INR_2}+\INR_1} \rp\\
&\quad + p_2 \log \lp \frac{\lp 1+\INR_2 \rp\lp 1+\frac{\SNR_2}{1+\INR_1}\rp}{1+\frac{\SNR_2+2\sqrt{\SNR_2 \cdot \INR_2}}{1+\INR_1}+\INR_2} \rp\\
&\quad -2\log 3 \\
&\geq \eqref{eq:ob_g_R1R2} - 2\log 3 - 2p_1-2p_2
\end{align*}
Finally, we consider \eqref{eq:g_2RiRj}
\begin{align*}
\eqref{eq:g_2RiRj} &= \log \lp 1 + \frac{\SNR_1}{1+\INR_2} \rp \\
&\quad + \log \lp 1 + \frac{\SNR_2}{1+\INR_1} + \INR_2 \rp \\
&\quad + \log \lp 1 + \SNR_1 + \INR_1 \rp \\
&\quad + p_2 \log \lp \frac{\lp 1+\INR_2 \rp\lp 1+\frac{\SNR_2}{1+\INR_1}\rp}{1+\frac{\SNR_2}{1+\INR_1}+\INR_2} \rp \\
&\geq \log \lp 1 + \frac{\SNR_1}{1+\INR_2} \rp \\
&\quad + \log \lp 1 + \frac{\SNR_2+2\sqrt{\SNR_2 \cdot \INR_2}}{1+\INR_1} + \INR_2 \rp \\
&\quad + \log \lp 1 + \SNR_1 + \INR_1 +2\sqrt{\SNR_1 \cdot \INR_1}\rp \\
&\quad + p_2 \log \lp \frac{\lp 1+\INR_2 \rp\lp 1+\frac{\SNR_2}{1+\INR_1}\rp}{1+\frac{\SNR_2}{1+\INR_1}+\INR_2} \rp \\
&\quad -2\log 3 \\
&\geq \log \lp 1 + \frac{\SNR_1}{1+\INR_2} \rp \\
&\quad + \log \lp 1 + \frac{\SNR_2+2\sqrt{\SNR_2 \cdot \INR_2}}{1+\INR_1} + \INR_2 \rp \\
&\quad + \log \lp 1 + \SNR_1 + \INR_1 +2\sqrt{\SNR_1 \cdot \INR_1}\rp \\
&\quad + p_2 \log \lp \frac{\lp 1+\INR_2 \rp\lp 1+\frac{\SNR_2}{1+\INR_1}\rp}{1+\frac{\SNR_2+2\sqrt{\SNR_2 \cdot \INR_2}}{1+\INR_1}+\INR_2} \rp \\
&\quad -2\log 3 \\
&\geq \eqref{eq:ob_g_2R1R2} - 2\log 3 - 1 -2p_2
\end{align*}
In order to bound $\delta_2$, we note that it cannot be larger than the maximum of the gaps found above, after normalization as done with bounding $\delta_1$. Hence, we find
\begin{align*}
\delta_2 < \log 3 + p_1 + p_2 \text{ bits.}
\end{align*}

\subsection{Comparison with Suh-Tse Outer Bound}
In this subsection, we compare our inner bound for the case $p_1=p_2=1$ with the perfect feedback outer bound of \cite{SuhTse_11}. Looking at the region \eqref{eq:g_Ri}--\eqref{eq:g_2RiRj}, we see that if we set $p_1=p_2=1$, then the bounds \eqref{eq:g_RiRj_2} and \eqref{eq:g_2RiRj} become redundant, and the region reduces to the outer bound region of \cite{SuhTse_11}, with the following differences:
\begin{itemize}
\item The outer bounds in \cite{SuhTse_11} are parameterized by the parameter $\rho$, which captures the correlation between the symbols of two users. In the region \eqref{eq:g_Ri}--\eqref{eq:g_2RiRj}, supremum values over all possible values of $\rho$ is given.
\item The bounds in \eqref{eq:g_Ri}--\eqref{eq:g_2RiRj} do not contain beamforming gain terms $2\rho\sqrt{\SNR_i\cdot\INR_i}$, which appear in the outer bounds of \cite{SuhTse_11}. 
\end{itemize}
It is easy to see that the first item does not result in a rate penalty, while the second one gives a penalty of $\log 3$ since
\begin{align*}
\eqref{eq:g_Ri}_L &= \log \lp 1+\SNR_1+\INR_1\rp \\
&\geq \log \lp 1+\SNR_1+\INR_1 + 2\sqrt{\SNR_1\cdot\INR_1}\rp - \log 3
\end{align*}
Hence, the region $\mcal{\bar C}(1,1)$ is within at most $\log 3$ bits of the outer bound region of \cite{SuhTse_11}. From the results of this section, we also know that our scheme achieves the region $\mcal{\bar C}(1,1)-\delta_1$. Evaluating $\delta_1$ for $p_1=p_2=1$, we see that the proposed scheme achieves within $3+3\log3 \approx 7.75$ bits of the Suh-Tse outer bound region.

%% file: AP_Threshold.tex
In this section, we prove that when the feedback probabilities are sufficiently high, perfect feedback sum-capacity can be achieved (approximately for Gaussian case, exactly for linear deterministic case). The precise statements for the two models are given in Corollaries~\ref{cor:threshold_ldc} and \ref{cor:threshold}.

\subsection{Proof of Corollary~\ref{cor:threshold_ldc}}
We will show that when feedback is perfect, the bounds on the sum rate that involve the feedback probabilities become strictly redundant. That is, setting $p_1=p_2=p$, we will prove that the bounds $\eqref{eq:ldc_R1R2_3}$, $\eqref{eq:ldc_R1}+\eqref{eq:ldc_R2}$, $\frac{\eqref{eq:ldc_R1}+\eqref{eq:ldc_R12R2}}{2}$, $\frac{\eqref{eq:ldc_R2}+\eqref{eq:ldc_2R1R2}}{2}$, and $\frac{\eqref{eq:ldc_2R1R2}+\eqref{eq:ldc_R12R2}}{3}$ are all strictly larger than the perfect feedback bound \eqref{eq:ldc_R1R2_1} when $p=1$ and $n_{12}, n_{21}>0$. Then the result follows by noting that all such bounds are continuous and monotonically increasing functions of $p$, and hence there must exist a $p^*<1$ such that whenever $p=p^*$, perfect feedback sum-rate bound \eqref{eq:ldc_R1R2_1} is exactly matched by these bounds.

We first prove a claim that will be used in the main proof.
\begin{claim}\label{cl:ldc_cor}
For $n_{21},n_{12}>0$,
\begin{align*}
\eqref{eq:ldc_R1R2_1} < n_{12} + n_{21} + \lp n_{11}-n_{21}\rp^+ + \lp n_{22}-n_{12}\rp^+
\end{align*}
\end{claim}
\begin{proof}
\begin{align*}
& n_{12} + n_{21} + \lp n_{11}-n_{21}\rp^+ + \lp n_{22}-n_{12}\rp^+ \\
&\quad = n_{12} + \max\lp n_{11}, n_{21}\rp + \lp n_{22}-n_{12}\rp^+ \\
&\quad = \max\lp n_{12} + n_{11}, n_{12} + n_{21}\rp + \lp n_{22}-n_{12}\rp^+ \\
&\quad > \max\lp n_{11}, n_{12} \rp + \lp n_{22}-n_{12}\rp^+ \geq \eqref{eq:ldc_R1R2_1}\\
\end{align*}
where the strict inequality follows by the fact that $n_{21},n_{12}>0$. 
\end{proof}

Next, we consider the bound \eqref{eq:ldc_R1R2_3}.
\begin{align*}
\eqref{eq:ldc_R1R2_3} &= \max \lbp n_{12}, \lp n_{11}-n_{21}\rp^+\rbp+\max \lbp n_{21}, \lp n_{22}-n_{12}\rp^+\rbp \\
& +\min \lbp n_{12}, \lp n_{11}-n_{21}\rp^+\rbp+\min \lbp n_{21}, \lp n_{22}-n_{12}\rp^+\rbp \\
&= n_{12} + n_{21} + \lp n_{11}-n_{21}\rp^+ + \lp n_{22}-n_{12}\rp^+ \\
&> \eqref{eq:ldc_R1R2_1}\\
\end{align*}
where the last line follows by Claim~\ref{cl:ldc_cor}. Hence, the bound \eqref{eq:ldc_R1R2_3} becomes strictly redundant. Next, consider
\begin{align*}
\eqref{eq:ldc_2R1R2}+\eqref{eq:ldc_R12R2} &= \max\lp n_{11},n_{12}\rp + \max\lp n_{22},n_{21}\rp \\
&\quad + \lp n_{11}-n_{21}\rp^+ + \lp n_{22}-n_{12}\rp^+ \\
&\quad + \max\lbp n_{12},\lp n_{11}-n_{21}\rp^+\rbp \\
&\quad + \max\lbp n_{21},\lp n_{22}-n_{12}\rp^+\rbp \\
&\quad + \min\lbp n_{12},\lp n_{11}-n_{21}\rp^+\rbp \\
&\quad + \min\lbp n_{21},\lp n_{22}-n_{12}\rp^+\rbp \\
&= \max\lp n_{11},n_{12}\rp + \max\lp n_{22},n_{21}\rp \\
&\quad + 2\lp n_{11}-n_{21}\rp^+ + 2\lp n_{22}-n_{12}\rp^+ n_{12} + n_{21} \\
&\geq 2\min \lbp \max\lp n_{11},n_{12}\rp + \lp n_{22}-n_{12}\rp^+ ,\right.\\
&\quad \left. +\max\lp n_{22},n_{21}\rp + \lp n_{11}-n_{21}\rp^+\rbp \\
&\quad + \lp n_{11}-n_{21}\rp^+ + \lp n_{22}-n_{12}\rp^+ n_{12} + n_{21} \\
&> 3 \cdot \eqref{eq:ldc_R1R2_1}
\end{align*}
where the last line follows by Claim~\ref{cl:ldc_cor}.

Next, we consider the sum of individual rate bounds. Since these bounds consist of the minimum of two terms, we consider each case separately. In what follows, $\eqref{eq:ldc_R1}_R$ denotes the term on the right-hand side of the minimization in \eqref{eq:ldc_R1}, while $\eqref{eq:ldc_R1}_L$ denotes the term on the left-hand side ($\eqref{eq:ldc_R2}_R$ and $\eqref{eq:ldc_R2}_L$ are also defined similarly). By symmetry, it is sufficient to prove that $\eqref{eq:ldc_R1}_R+\eqref{eq:ldc_R2}_R$ and $\eqref{eq:ldc_R1}_L+\eqref{eq:ldc_R2}_R$ are strictly redundant for $p_1=p_2=1$. We show this as follows.
\begin{align*}
\eqref{eq:ldc_R1}_R+\eqref{eq:ldc_R2}_R &= n_{11} + n_{22} + \lp n_{21} - n_{11}\rp^+ \\
&\quad + \lp n_{12} - n_{22}\rp^+ \\
&= \max \lp n_{11}, n_{21}\rp +  \max \lp n_{22}, n_{12}\rp \\
&= n_{12} + n_{21} + \lp n_{11}-n_{21}\rp^++ \lp n_{22}-n_{12}\rp^+ \\
&> \eqref{eq:ldc_R1R2_1}
\end{align*}
by Claim~\ref{cl:ldc_cor}, and
\begin{align*}
\eqref{eq:ldc_R1}_L+\eqref{eq:ldc_R2}_R &= \max\lp n_{11},n_{12} \rp + n_{22} + \lp n_{12} - n_{22}\rp^+ \\
&= \max\lp n_{11},n_{12} \rp +  \max \lp n_{22}, n_{12}\rp \\
&> \max\lp n_{11},n_{12} \rp +  \lp n_{22} - n_{12}\rp^+ \geq \eqref{eq:ldc_R1R2_1} \\
\end{align*}
since $n_{12}>0$.

Finally, we consider the bounds $\eqref{eq:ldc_2R1R2}+\eqref{eq:ldc_R2}$, and $\eqref{eq:ldc_R12R2}+\eqref{eq:ldc_R1}$. By symmetry, it is sufficient to show the redundancy of  $\eqref{eq:ldc_2R1R2}+\eqref{eq:ldc_R2}_R$ and $\eqref{eq:ldc_2R1R2}+\eqref{eq:ldc_R2}_L$. The former is shown by
\begin{align*}
\eqref{eq:ldc_2R1R2}+\eqref{eq:ldc_R2}_R &= \max\lp n_{11},n_{12}\rp + \lp n_{11}-n_{21}\rp^+ \\
&\quad + \max\lbp n_{21},\lp n_{22}-n_{12}\rp^+\rbp \\
&\quad + \min\lbp n_{21},\lp n_{22}-n_{12}\rp^+\rbp+ n_{22} \\
&\quad + \lp n_{12}-n_{22}\rp^+ \\
&= \max\lp n_{11},n_{12}\rp + \lp n_{11}-n_{21}\rp^+ \\
&\quad + n_{21} + \lp n_{22}-n_{12}\rp^+ + \max\lp n_{22}, n_{12}\rp \\
&= \max\lp n_{11},n_{12}\rp + + \lp n_{22}-n_{12}\rp^+ \\
&\quad + n_{21}  + n_{12}+ \lp n_{11}-n_{21}\rp^+ \lp n_{22}- n_{12}\rp^+ \\
&\geq \eqref{eq:ldc_R1R2_1} + n_{21}  + n_{12}+ \lp n_{11}-n_{21}\rp^+ \lp n_{22}- n_{12}\rp^+\\
&> 2\cdot \eqref{eq:ldc_R1R2_1} 
\end{align*}
where the last line follows by Claim~\ref{cl:ldc_cor}, and
\begin{align*}
&\eqref{eq:ldc_2R1R2}+\eqref{eq:ldc_R2}_L = \max\lp n_{11},n_{12}\rp + \lp n_{11}-n_{21}\rp^+ \\
& + \max\lbp n_{21},\lp n_{22}-n_{12}\rp^+\rbp + \min\lbp n_{21},\lp n_{22}-n_{12}\rp^+\rbp \\
& + \max \lbp n_{22}, n_{21}\rbp \\
&= \max\lp n_{11},n_{12}\rp + \lp n_{11}-n_{21}\rp^+ \\
&\quad + n_{21} + \lp n_{22}-n_{12}\rp^+ + \max\lp n_{22}, n_{21}\rp \\
&> \max\lp n_{11},n_{12}\rp + \lp n_{11}-n_{21}\rp^+ \\
&\quad+  \max\lp n_{22}, n_{21}\rp + \lp n_{22}-n_{12}\rp^+ \\
&\geq 2\min \lbp \max\lp n_{11},n_{12}\rp + \lp n_{22}-n_{12}\rp^+ ,\right. \\
&\quad+ \left. \max\lp n_{22}, n_{21}\rp + \lp n_{11}-n_{21}\rp^+  \rbp\\
&= 2\cdot \eqref{eq:ldc_R1R2_1} 
\end{align*}
where the strict inequality follows by the fact that $n_{21}>0$.

\subsection{Proof of Corollary~\ref{cor:threshold}}
Similar to the proof of Corollary~\ref{cor:threshold_ldc}, we will show that when feedback is perfect, the bounds on the sum rate that involve the feedback probabilities become redundant for the set $\mcal{\bar C}(p_1, p_2)$. Since the capacity region is within a constant gap of the region $\mcal{\bar C}(p_1, p_2)$, for all channel parameters, the result will follow. 

Specifically, setting $p_1=p_2=p$, we will show that when $p=1, \INR_1,\INR_2>0$, the bounds $\eqref{eq:g_RiRj_2}$, $\eqref{eq:g_Ri}_{(1,2)}+\eqref{eq:g_Ri}_{(2,1)}$, $\frac{\eqref{eq:g_Ri}_{(1,2)}+\eqref{eq:g_2RiRj}_{(2,1)}}{2}$, $\frac{\eqref{eq:g_Ri}_{(2,1)}+\eqref{eq:g_2RiRj}_{(1,2)}}{2}$, and $\frac{\eqref{eq:g_2RiRj}_{(1,2)}+\eqref{eq:g_2RiRj}_{(2,1)}}{3}$ are all strictly larger than the perfect feedback bound \eqref{eq:g_RiRj_1}, where the subscript $\lp a,b\rp$ denotes the evaluation of the relevant bound with $(i,j)=(a,b)$.

We first prove a claim that will be useful in the proof of the corollary.
\begin{claim}\label{cl:g_cor}
For $\INR_1,\INR_2> 0$,
\begin{align*}
&\min_{\ijj}\eqref{eq:g_RiRj_1} < \log \lp 1+\INR_1\rp + \log \lp 1+\INR_2\rp \\
&\quad +\log \lp 1+\frac{\SNR_1}{1+\INR_2}\rp + \log \lp 1+\frac{\SNR_2}{1+\INR_1}\rp
\end{align*}
\begin{align*}
&\log \lp 1+\INR_1\rp + \log \lp 1+\INR_2\rp \\
&\quad +\log \lp 1+\frac{\SNR_1}{1+\INR_2}\rp + \log \lp 1+\frac{\SNR_2}{1+\INR_1}\rp \\
&= \log \lp 1+\frac{\SNR_2}{1+\INR_1}\rp \\
&  + \log\lp 1+\INR_1+\INR_2+\SNR_1+\INR_1\INR_2 + \INR_1\SNR_1 \rp \\
&> \log\lp 1+\SNR_1+\INR_1\rp + \log \lp1+\frac{\SNR_2}{1+\INR_1} \rp \\
&\geq \min_{\ijj}\eqref{eq:g_RiRj_1}
\end{align*}
\end{claim}
Next, we show that under the condition $\INR_1,\INR_2> 0$ and $p=1$, all of the mentioned bounds are strictly redundant. We start with \eqref{eq:g_RiRj_2}:
\begin{align*}
\eqref{eq:g_RiRj_2} &= \log \lp 1+\INR_1\rp + \log \lp 1+\INR_2\rp \\
&\quad +\log \lp 1+\frac{\SNR_1}{1+\INR_2}\rp + \log \lp 1+\frac{\SNR_2}{1+\INR_1}\rp \\
&> \min_{\ijj}\eqref{eq:g_RiRj_1}
\end{align*}
by Claim~\ref{cl:g_cor}. Next, 
\begin{align*}
&\eqref{eq:g_2RiRj}_{(1,2)}+\eqref{eq:g_2RiRj}_{(2,1)} =  \log\lp 1+\SNR_1+\INR_1\rp \\
&\quad  +\log\lp 1+\SNR_2+\INR_2\rp + 2\log \lp1+\frac{\SNR_1}{1+\INR_2} \rp \\
&\quad + 2\log \lp1+\frac{\SNR_2}{1+\INR_1} \rp + \log\lp 1 + \INR_1\rp \\
&\quad  + \log\lp 1 + \INR_2\rp \\
&= \eqref{eq:g_RiRj_1}_{(1,2)}+ \eqref{eq:g_RiRj_1}_{(2,1)} + \log \lp 1+\INR_i\rp + \log \lp 1+\INR_j\rp \\
&\quad +\log \lp 1+\frac{\SNR_i}{1+\INR_j}\rp + \log \lp 1+\frac{\SNR_j}{1+\INR_i}\rp \\
&> 2\cdot \eqref{eq:g_RiRj_1}_{(1,2)}+ \eqref{eq:g_RiRj_1}_{(2,1)} \geq 3\cdot \min_{\ijj} \eqref{eq:g_RiRj_1}
\end{align*}
by Claim~\ref{cl:g_cor}. Next, we consider the bounds $\eqref{eq:g_Ri}_{(1,2), R}+\eqref{eq:g_Ri}_{(2,1), R}$, $\eqref{eq:g_Ri}_{(1,2), L}+\eqref{eq:g_Ri}_{(2,1), R}$, and $\eqref{eq:g_Ri}_{(1,2), R}+\eqref{eq:g_Ri}_{(2,1), L}$, where the subscript $L$ and $R$ refer to the left-hand and right-hand side of the minimization in the relevant bounds, respectively. By symmetry, it is sufficient to show the redundancy of the former two.
\begin{align*}
&\eqref{eq:g_Ri}_{(1,2), R}+\eqref{eq:g_Ri}_{(2,1), R}  = \log \lp 1+\INR_1\rp + \log \lp 1+\INR_2\rp \\
&\quad +\log \lp 1+\frac{\SNR_1}{1+\INR_2}\rp + \log \lp 1+\frac{\SNR_2}{1+\INR_1}\rp \\
&> \min_{\ijj}\eqref{eq:g_RiRj_1}
\end{align*}
by Claim~\ref{cl:g_cor}, and 
\begin{align*}
&\eqref{eq:g_Ri}_{(1,2), L}+\eqref{eq:g_Ri}_{(2,1), R} = \log \lp 1+\SNR_1 + \INR_1\rp \\
&\quad + \log \lp 1+\SNR_2+\INR_1\rp \\
&>  \log \lp 1+\SNR_1 + \INR_1\rp + \log \lp 1+\frac{\SNR_2}{1+\INR_1}\rp \\
&\geq \min_{\ijj}\eqref{eq:g_RiRj_1}
\end{align*}
since $\INR_1,\INR_2>0$. Finally, we show the redundancy of $\frac{\eqref{eq:g_Ri}_{(1,2)}+\eqref{eq:g_2RiRj}_{(2,1)}}{2}$ and $\frac{\eqref{eq:g_Ri}_{(2,1)}+\eqref{eq:g_2RiRj}_{(1,2)}}{2}$. By symmetry, we only consider the former, and in particular, $\frac{\eqref{eq:g_Ri}_{(1,2),L}+\eqref{eq:g_2RiRj}_{(2,1)}}{2}$ and $\frac{\eqref{eq:g_Ri}_{(1,2),R}+\eqref{eq:g_2RiRj}_{(2,1)}}{2}$, where the subscript $L$ and $R$ refer to the left-hand and right-hand side of the minimization in the relevant bounds, respectively. Then
\begin{align*}
&\eqref{eq:g_Ri}_{(1,2),L}+\eqref{eq:g_2RiRj}_{(2,1)} = \log \lp 1+\SNR_1 + \INR_1\rp \\
&\quad + \log \lp 1+\SNR_2 + \INR_2\rp +  \log \lp 1+\frac{\SNR_2}{1+\INR_1} \rp \\
&\quad + \log\lp 1+\INR_1\rp+ \log \lp 1+\frac{\SNR_1}{1+\INR_2} \rp \\
&> \log \lp 1+\SNR_1 + \INR_1\rp \\
&\quad + \log \lp 1+\SNR_2 + \INR_2\rp +  \log \lp 1+\frac{\SNR_2}{1+\INR_1} \rp \\
&\quad + \log \lp 1+\frac{\SNR_1}{1+\INR_2} \rp = \eqref{eq:g_RiRj_1}_{(1,2)} + \eqref{eq:g_RiRj_1}_{(2,1)} \\
&\geq \min_{\ijj}\eqref{eq:g_RiRj_1}
\end{align*}
and
\begin{align*}
&\eqref{eq:g_Ri}_{(1,2),R}+\eqref{eq:g_2RiRj}_{(2,1)} = \log \lp 1+ \INR_2\rp \\
&\quad + \log \lp 1+\SNR_2 + \INR_2\rp +  \log \lp 1+\frac{\SNR_2}{1+\INR_1} \rp \\
&\quad + \log\lp 1+\INR_1\rp+ 2\log \lp 1+\frac{\SNR_1}{1+\INR_2} \rp \\
&> \eqref{eq:g_RiRj_1}_{(2,1)} + \log \lp 1+ \INR_2\rp + \log \lp 1+ \INR_1\rp \\
&\quad \log \lp 1+\frac{\SNR_2}{1+\INR_1} \rp + \log \lp 1+\frac{\SNR_1}{1+\INR_2} \rp \\
&> 2\cdot \min_{\ijj} \eqref{eq:g_RiRj_1}.
\end{align*}

%% file: GICIFB.bbl
\begin{thebibliography}{10}

\bibitem{SuhTse_11}
C.~Suh and D.~N.~C. Tse, ``Feedback capacity of the {G}aussian interference
  channel to within $2$ bits,'' {\em IEEE Transactions on Information Theory},
  vol.~57, pp.~2667--2685, May 2011.

\bibitem{Shannon_56}
C.~E. Shannon, ``The zero error capacity of a noisy channel,'' {\em Information
  Theory, IRE Transactions on}, vol.~2, no.~3, pp.~8--19, 1956.

\bibitem{Ozarow_84}
L.~H. Ozarow, ``The capacity of the white gaussian multiple access channel with
  feedback,'' {\em Information Theory, IEEE Transactions on}, vol.~30, no.~4,
  pp.~623--629, 1984.

\bibitem{VahidSuh_12}
A.~Vahid, C.~Suh, and A.~S. Avestimehr, ``Interference channels with
  rate-limited feedback,'' {\em IEEE Transactions on Information Theory},
  vol.~58, pp.~2788--2812, May 2012.

\bibitem{LeTandon_12}
S.-Q. Le, R.~Tandon, M.~Motani, and H.~V. Poor, ``The capacity region of the
  symmetric linear deterministic interference channel with partial feedback,''
  {\em Proceedings of Allerton Conference on Communication, Control, and
  Computing}, October 2012.

\bibitem{SahaiAggarwal_09}
A.~Sahai, V.~Aggarwal, M.~Yuksel, and A.~Sabharwal, ``On channel output
  feedback in deterministic interference channels,'' in {\em Information Theory
  Workshop, 2009. ITW 2009. IEEE}, pp.~298--302, IEEE, 2009.

\bibitem{SuhWang_12}
C.~Suh, I.-H. Wang, and D.~N.~C. Tse, ``Two-way interference channels,'' {\em
  Proceedings of IEEE International Symposium on Information Theory},
  pp.~2811--2815, July 2012.

\bibitem{AvestimehrDiggavi_09}
A.~S. Avestimehr, S.~N. Diggavi, and D.~N.~C. Tse, ``Wireless network
  information flow: A deterministic approach,'' {\em IEEE Transactions on
  Information Theory}, vol.~57, pp.~1872--1905, April 2011.

\bibitem{LimKim_11}
S.~Lim, Y.-H. Kim, A.~El~Gamal, and S.-Y. Chung, ``Noisy network coding,'' {\em
  Information Theory, IEEE Transactions on}, vol.~57, no.~5, pp.~3132--3152,
  2011.

\bibitem{OzgurDiggavi_10}
A.~Ozgur and S.~Diggavi, ``Approximately achieving gaussian relay network
  capacity with lattice codes,'' {\em arXiv preprint arXiv:1005.1284}, 2010.

\bibitem{OzgurDiggavi_13}
A.~Ozgur and S.~Diggavi, ``Approximately achieving gaussian relay network
  capacity with lattice-based qmf codes,'' {\em IEEE Transactions on
  Information Theory}, vol.~59, pp.~8275--8294, Dec 2013.

\bibitem{KramerHou_11}
G.~Kramer and J.~Hou, ``On message lengths for noisy network coding,'' in {\em
  Information Theory Workshop (ITW), 2011 IEEE}, pp.~430--431, IEEE, 2011.

\bibitem{HanKobayashi_81}
T.~S. Han and K.~Kobayashi, ``A new achievable rate region for the interference
  channel,'' {\em IEEE Transactions on Information Theory}, vol.~27,
  pp.~49--60, January 1981.

\bibitem{KarakusWang_13}
C.~Karakus, I.-H. Wang, and S.~Diggavi, ``Interference channel with
  intermittent feedback,'' in {\em Information Theory Proceedings (ISIT), 2013
  IEEE International Symposium on}, pp.~26--30, IEEE, 2013.

\bibitem{KarakusWang_13_2}
C.~Karakus, I.-H. Wang, and S.~Diggavi, ``An achievable rate region for
  gaussian interference channel with intermittent feedback,'' in {\em
  Communication, Control, and Computing (Allerton), 2013 51st Annual Allerton
  Conference on}, pp.~203--210, Oct 2013.

\bibitem{KhistiLapidoth_13}
A.~Khisti and A.~Lapidoth, ``Multiple access channels with intermittent
  feedback and side information,'' in {\em Information Theory Proceedings
  (ISIT), 2013 IEEE International Symposium on}, pp.~2631--2635, IEEE, 2013.

\bibitem{EtkinTse_07}
R.~Etkin, D.~N.~C. Tse, and H.~Wang, ``Gaussian interference channel capacity
  to within one bit,'' {\em IEEE Transactions on Information Theory}, vol.~54,
  pp.~5534--5562, December 2008.

\bibitem{Zaidi_14}
A.~Zaidi, ``Achievable regions for interference channels with generalized and
  intermittent feedback,'' in {\em Information Theory Proceedings (ISIT), 2014
  IEEE International Symposium on}, pp.~1026–--1030, IEEE, 2014.

\bibitem{El-GamalKim_11}
A.~E. Gamal and Y.-H. Kim, {\em Network Information Theory}.
\newblock Cambridge University Press, 2011.

\bibitem{Durrett_10}
R.~Durrett, {\em Probability: theory and examples}, vol.~3.
\newblock Cambridge university press, 2010.

\end{thebibliography}
